\documentclass[aps,prx,reprint,superscriptaddress,longbibliography,nofootinbib]{revtex4-2}

\usepackage{hyperref}
\usepackage{enumitem}
\usepackage{xcolor}
\usepackage{epstopdf}
\usepackage[normalem]{ulem} 
\usepackage{comment}
\usepackage{subfigure}
\usepackage{tikz}
\usepackage{tabularx}

\usepackage[T1]{fontenc}
\usepackage[utf8]{inputenc}
\usepackage{babel}

\usepackage{amsmath,amssymb,amsthm,mathtools}
\usepackage{bm}
\usepackage{braket}

\usepackage{graphicx}
\usepackage{dcolumn}
\usepackage{booktabs}
\usepackage{multirow}

\theoremstyle{plain}
\newtheorem{theorem}{Theorem}[section]
\newtheorem*{theorem*}{Theorem}
\newtheorem{lemma}[theorem]{Lemma}
\newtheorem{proposition}[theorem]{Proposition}

\theoremstyle{definition}
\newtheorem{definition}[theorem]{Definition}

\newtheorem{axiom}{Axiom}

\usepackage{soul}

\newcommand{\cneg}{\mathord{\sim}}

\begin{document}

\title{Bell's \emph{go} theorem: regaining local causality with a new notion of realism}

\author{Lorenzo Catani}\email{lorenzo.catani4@gmail.com}
\affiliation{INESC TEC, Rua da Universidade, 4710-057 Braga, Portugal}
\affiliation{International Iberian Nanotechnology Laboratory, Av. Mestre José Veiga s/n, 4715-330 Braga, Portugal}

\author{Nesta van der Schaaf}
\affiliation{Université Paris-Saclay, CNRS, CentraleSupélec, ENS Paris-Saclay, Inria,
Laboratoire Méthodes Formelles, 91190, Gif-sur-Yvette, France}
\author{Rui Soares Barbosa}
\affiliation{International Iberian Nanotechnology Laboratory, Av. Mestre José Veiga s/n, 4715-330 Braga, Portugal}

\begin{abstract}

Bell's theorem is popularly stated as the inconsistency of quantum theory with local realism. More precisely, realism here is mathematically described by the ontological models framework, where the elements of reality are represented by points in a measurable space, and probabilities obey the Kolmogorov axioms. In the Bell scenario, locally causal ontological models posit a common cause supporting a global valuation---a joint assignment of outcomes to all measurements at once, across mutually incompatible contexts. Within this notion of realism, Bell's theorem implies that any realist explanation of the quantum predictions must be conspiratorial: any underlying mechanism, whether superluminal, superdeterministic, or retrocausal, must be finely tuned so as never to manifest at the operational level.
In this work, we redefine realism by developing a new ontological framework that generalizes the standard one. The elements of reality are no longer points, but parts of a point-free  generalization of a topological space (sublocales of a locale). These parts, which we call records, are primitive and inherently relational. A record shared by two incompatible contexts establishes the sameness of a measurement across them, without such pairwise identities composing into a global valuation, and the common cause is then a web of such records.  Within this framework, Bell's no-go theorem becomes a \emph{go} theorem: local causality, restated for records, is no longer at odds with the quantum statistics. This is possible because modularity, the inclusion--exclusion law of classical probability, fails across incompatible contexts, so that the underlying probability theory generalizes Kolmogorov's. We exhibit an explicit family of models in the new framework, reproducing several quantum strategies, including the optimal CHSH violation. More broadly, locally causal models of this kind attain every Bell--CHSH value up to the algebraic maximum. A realist explanation of the quantum predictions need not be conspiratorial; every causal influence posited by these models is reflected in the operational predictions.

\end{abstract} 

\maketitle

\tableofcontents

\section{Introduction}
\label{sec:Intro}


If one asks researchers in the foundations of quantum theory to name the most important result in the field, the answer is almost unanimous: Bell's no-go theorem~\cite{Bell_1964}. It formally establishes, via the violation of an inequality, a way in which quantum theory departs from the classical worldview, by showing that the statistics predicted by quantum theory are inconsistent with those of any ontological model satisfying local causality.
Ontological models~\cite{Harrigan2010}, also known as hidden-variable models, constitute a mathematically rigorous framework for providing a realist explanation of the statistical predictions of a theory.
The scenario in which the theorem is formulated involves two spacelike-separated parties, Alice and Bob, who measure parts of a system prepared in their common past, choosing settings $x$ and $y$ and obtaining outcomes~$a$ and~$b$. Each round realizes a measurement \emph{context}, consisting of the pair of measurements jointly performed, one for each party. Local causality requires a model of the correlations between $a$ and $b$ to respect relativistic causality (no event being influenced by anything outside its past light cone) and, by Reichenbach's principle, to explain them by a common cause in the shared past, conditioning on which the two outcomes become independent.
Despite the theorem's established importance, and more than sixty years after its first formulation, there is still no consensus on what lesson to draw from it.

\smallskip


Bell's no-go theorem forces us to abandon one of its two assumptions.
One possibility is to abandon local causality while retaining the ontological models framework as the mathematical embodiment of realism.
A realist model of this kind must then involve a superluminal causal influence---say, from Alice's setting $x$ to Bob's outcome $b$.
Yet the quantum statistics are no-signalling, so this influence leaves no operational trace, with Bob's marginal outcome distribution being exactly the same for any choice of Alice's measurement. The question this realist explanation cannot escape is then: why should nature arrange a genuine causal influence so as to be precisely undetectable at the operational level?
More generally, every classical causal explanation of the quantum predictions becomes unavoidably conspiratorial in this way, whether the mechanism is superluminal, retrocausal, or superdeterministic,\footnote{We take, as in~\cite{CataniLeifer2023}, the measurement-independence assumption (capturing no-retrocausality and no-superdeterminism) to be an assumption of the ontological models framework.} positing causal influences that are not reflected at the level of the operational predictions~\cite{Wood2015}.
This mismatch between the statistical predictions and the realist explanation of the theory is known as the problem of fine-tuning, since the influence the model posits is concealed, and no-signalling (or, in different scenarios, some other operational equivalence) is recovered only through unjustified adjustments to the parameters of the ontological model.
Taking this problem seriously, the lesson of the theorem is that the ontological models framework itself must be replaced, as our way of modelling reality, by one that reproduces the quantum statistics without fine-tuning.

\smallskip


How, though, can such a framework be developed? One difficulty is that the ontological models framework seems to capture what realism itself means. Its elements of reality, complete specifications of the properties of physical systems, are individual units, each individuated on its own, independently of any other and of any whole they may belong to. They are thus mathematically represented by points of a measurable space, and the statistics of a theory are recovered through classical (Kolmogorov) probability theory. This simplicity does not limit the framework's generality. It carries no commitment as to what the basic entities are, be they particles, fields, wave functions, or something else.

The framework is, moreover, theory-independent, providing a realist account of any operational theory, not just quantum theory. It has also enabled remarkable progress, allowing one to rigorously pose questions such as whether the quantum state describes reality directly or merely our knowledge of it (the $\psi$-ontic vs.\ $\psi$-epistemic debate~\cite{LeiferOntology,PBR}), and to formally define departures from classicality through the no-go theorems, of which Bell's is one example among many~\cite{KochenSpecker1967,Spekkens2005,LeiferMaroney2013,Chaturvedi2020}.
Some realist interpretations of quantum theory can be cast within the framework~\cite{BohmianMechanics}, while others lie outside it~\cite{Wallace}; we discuss these in App.~\ref{Appendix:Comparison}.
To the best of our knowledge, however, no existing proposal to replace the framework meets the requirements the task imposes. We articulate them as three desiderata. The new framework must (i) embody a precise notion of reality, represented through a rigorous mathematical structure; (ii) be theory-independent, applicable to any operational theory for which one seeks a realist explanation; and (iii) reproduce the quantum statistics without fine-tuning.


In this work, we introduce the \emph{localic ontological models framework}: a new, theory-independent mathematical framework that generalizes the standard one and embodies an inherently relational notion of reality. Within it, local causality no longer forces a contradiction with the statistics of quantum theory, so that all three desiderata are met, the third in the paradigmatic case of the Bell scenario. Specifically, we define locally causal localic models in the Bell scenario and show that the bound attainable by the Bell--CHSH expression~\cite{CHSH69} rises from the classical value all the way to the algebraic maximum. Within one and the same family of models we exhibit explicit representations of quantum strategies, including the one achieving the Tsirelson bound. The lesson we draw from Bell's theorem is then not that local causality fails, but that the standard framework smuggles in a hidden commitment about the common cause, namely that it must be a global valuation, a joint assignment of outcomes to all measurements at once, despite their belonging to mutually incompatible contexts, which can never be jointly realized.

We keep the causal assumptions of the theorem, local causality and measurement independence, restated for the new framework. What we redefine is realism, and with it what the common cause is. We replace the independently individuated units that support a global valuation with \emph{records}, relational units through which incompatibility is reflected at the ontological level. The common cause becomes a web of records fixed in the shared past, and the record a round realizes screens the outcomes off.


In the localic ontological models framework, the elements of reality (the records) are no longer points of a measurable space but \emph{parts} (sublocales) of a \emph{locale}, a point-free generalization of a topological space~\cite{PicadoPultr2012,Vickers1989}; hence the name. A part is individuated by how it restricts the whole, not by a list of points, so a locale can even possess parts that are not empty and yet contain no point. It is exactly these non-spatial parts that do the new work.\footnote{A locale is \emph{spatial} when its frame is the set of open subsets of some topological space, so it is an ordinary topological space in disguise. Nevertheless, even spatial locales may possess parts that are \emph{non-spatial}, and so cannot be described by collections of points; the extreme case is a pointless part, containing no points at all. Section~\ref{sec:pointfree} exhibits one inside the spatial locale coming from the real line.} 
They establish the sameness of a measurement across the two incompatible contexts in which it appears, while lying in those two alone; in the standard framework, by contrast, every point lies in all contexts at once and answers every measurement, which is what makes it a global valuation.

Preparations are still represented by weights over the elements of reality and measurements by response functions, with the operational statistics recovered by integration, now over records.  The departure from Kolmogorov probability lies in the weight, which remains non-negative but is not obliged to satisfy the axioms in full. Within each context the weight is classical, additive over the disjoint outcomes and normalized. Across incompatible contexts, which overlap precisely in the records they share, the weights need not fit into a single additive account, and the \emph{modularity} law of classical probability may fail there. The failure is thereby both explained and confined: it descends from each context being normalized on its own while sharing only part of its weight with the others, and it shows only on incompatible contexts taken together, which no round realizes, precisely where no experiment can ever check it. The standard framework is recovered as a special case, in which each record is made of points, each point answers every measurement, and the weight of a preparation is a single probability distribution over the points (Sec.~\ref{sec:localic_OMF}).


The family of localic models with which we prove our results has a symmetric structure (Fig.~\ref{fig:localicOM_Bell}). Each of the four Bell contexts contains \emph{core} parts, which belong to that context alone, and each pair of contexts sharing a measurement is linked by \emph{glue} parts, i.e., records belonging to both. It is the glues that make factorizability, the mathematical expression of local causality, a substantive condition, as a glue carries a response function in each of its two contexts and the condition requires the two to give the same answer to the measurement those contexts share. The glue structure itself is fixed, once for every state and measurement, and states and measurements differ only in how they distribute weight over it and read its records out.

As a single weight parameter varies, the Bell--CHSH value takes every value from the classical bound up to the algebraic maximum, with each model so obtained remaining locally causal.
We also exhibit representations of the quantum strategies within this family of models. In the optimal strategy, the weight representing the maximally entangled state populates the glues, and the response functions representing the optimal incompatible measurements read out exactly the records that enhance the winning correlations, yielding the Tsirelson value $2\sqrt{2}$. The state \emph{plants} the weight and the measurements \emph{harvest} it. Keeping the same state but measuring local Paulis, the harvest misaligns and the stabilizer strategy stays within the classical bound: the difference between the two strategies lies entirely in the response functions, not in the ontological representation of the state. Finally, a product state populates no glue at all, and its statistics are equally well reproduced by a standard locally causal ontological model.


This work is the first step of a research programme, and what it establishes is an existence result. It exhibits locally causal localic models reproducing the statistics of quantum strategies in the Bell scenario, with three worked out explicitly.
Nonlocality is thereby shown to be an artefact of the notion of realism the standard ontological models framework encodes, rather than a necessity for realist explanations of the Bell correlations. 
The models reproduce the no-signalling statistics with no influence to conceal, since each party's response at a record depends on that party's setting alone, and no equally fine-tuned retrocausal or superdeterministic mechanism is needed either. No-signalling, derived from local causality in the standard framework, is here required of the weight (Sec.~\ref{sec:go}), and we expect the rule composing systems to derive it, as we discuss in Sec.~\ref{sec:programme}.

That models of this kind reach the algebraic maximum shows that, once the global valuation is given up, local causality bounds nothing by itself, and it relocates the central question: what singles out quantum theory is no longer the degree of nonlocality an ontological model must posit to reproduce its correlations---a degree that operational approaches seek to explain through information-theoretic principles~\cite{Popescu1994,hardy2001quantum,ChiribellaDArianoPerinotti2011,MasanesMuller2011}---but which allocations of weight to the records nature permits. We speculate in Sec.~\ref{sec:programme} on the origin of the Tsirelson bound, the composition rule for locales, and the extension of the same recipe to contextuality and extended Wigner's friend scenarios.

The framework may sound abstract, mostly because point-free topology is not part of the standard toolbox of physics, but the story it tells about Bell's theorem is simple, and we now tell it in full, with no mathematics: the global valuation becomes a shared plan, the records need no translation, and the failure of modularity becomes the observation that weights across incompatible contexts have lost their reason to add.

\subsection*{A simple story}


Imagine that two friends, Alice and Bob, are sent to opposite ends of the Earth. Each is later asked a question and gives an answer. When the answers are compared, they are correlated. The natural explanation is simple: before separating, the friends agreed on a plan, a complete list of answers for each question. This shared past explains the correlation, and nothing needs to travel between them to provide the correlated answers.


Imagine now the same scenario, but where the questions are what we call incompatible measurements.\footnote{We use ``question'' and ``measurement'' interchangeably: a question put to the system is a measurement performed on it, and its outcome is the answer.} These are measurements that have no counterpart in classical everyday life, but exist in quantum theory; \emph{incompatible} means they cannot both be performed in the same experiment.  Alice chooses between two such measurements $A_0$ and $A_1$ (for example, the spin of an electron along the $X$ or the $Z$ direction), and Bob likewise has two of his own, $B_0$ and $B_1$. Whenever Alice and Bob each choose one measurement, their two choices define a particular experimental situation, which we call a \emph{context}. There are four: $\{A_0,B_0\}$, $\{A_0,B_1\}$, $\{A_1,B_0\}$, and $\{A_1,B_1\}$. A context is the set of measurements performed together in a single round, so in any one round exactly one context is realized, and different rounds may realize different ones. Notice that two contexts may still share one measurement. For instance, $\{A_0,B_0\}$ and $\{A_0,B_1\}$ share Alice's measurement $A_0$ and differ only in Bob's; since Bob cannot perform $B_0$ and $B_1$ together, these two contexts can never occur in the same round. We may picture two contexts as two incompatible perspectives, two experimental situations each individually realizable, but never both at once.


Bell's theorem concerns explanations of this kind, which rest on a shared past. If the correlations between Alice's and Bob's outcomes come from a plan laid down before the two separated, then, taken across these contexts together, the correlations cannot exceed a certain limit.
Quantum systems, however, can be correlated more strongly than this limit allows. Faced with this, one \emph{may} conclude that the strong correlations require some influence between Alice's region and Bob's. Since the two are far apart, such an influence would have to travel faster than light, even though no signal could actually be sent between these two regions,\footnote{To send a message one must control what is transmitted, but neither Alice nor Bob can control their own outcomes: each, looking only at their own side, sees nothing but an unpredictable sequence of results, and the correlation surfaces only when the two later compare their results. Nothing can thus be sent from one region to the other, in accordance with relativity, which forbids any faster-than-light signal.} giving rise to what we called earlier the problem of fine-tuning.


But look at what such a shared plan assumes. A plan gives every measurement a single outcome, fixed in advance,\footnote{We take the plan to fix a definite answer for each measurement only because that is the most vivid picture; it costs no generality. Were the friends instead to share a random recipe fixing merely the \emph{probability} of each answer, the resulting correlations would be exactly those of some mixture of definite plans~\cite{Fine}. What matters below is therefore not that the answers are fixed in advance, but the next feature, that one assignment covers all the measurements at once.} and it does this for all of them at once---even though they belong to different, mutually incompatible contexts. It is one universal list, and because the contexts overlap, each measurement belonging to two of them, it makes every context agree with every other on the outcomes of the measurements they share. Yet no experiment can perform incompatible measurements together, so none can ever record such a list, not even in principle. It can only be \emph{posited} to exist. 


In this work, we reject this notion of plan and replace it with a different notion of shared past, corresponding to a web of records. The role of a record is to identify the same measurement across the two incompatible contexts it appears in. While a plan fixes an outcome for every measurement in every context at once, a record does far less. Suppose Alice measures $A_0$. That same measurement appears in two different contexts, according to which measurement Bob chooses alongside it, and a record simply says that these two appearances of~$A_0$ are one and the same. Such a record is a relational unit, and a web of such records ties the overlapping contexts without ever fusing them into one universal list: the overlaps are kept, the global list is not.


One may wonder how a record can establish the sameness of a measurement across two contexts without amounting to a plan. The difference lies in what is registered. In a plan, the sameness of $A_0$ across its two contexts comes for free, because one entry answers every measurement at once. Such an entry is not a full account of a round, which consists of the pair of questions posed and the answers they gave. By registering the answers without the questions, the entry belongs to all four contexts at once, and this is the step we refuse. Only one of the four contexts is realized in a round, and which questions were posed is part of what happened. A full account therefore belongs to the one context realized and to no other, and an entry belonging to all four contains answers to questions the round never posed.

A record shared by two contexts registers only part of what happened instead. Take the two contexts containing $A_0$. Their shared record registers Alice's question and nothing of the question posed to Bob with it, and it answers $A_0$ alike in both. Leaving Bob's question out is exactly what lets it belong to both contexts, so that it can be part of a round of either. Bob's answer results from the record together with the question he posed, $B_0$ in one context and $B_1$ in the other. The shared record thus belongs to two alternatives that cannot both occur, the round that took place and the one that did not, and this gives it a counterfactual character.

Even taken together, shared records do not amount to a plan. Each ties the two contexts that share a question and only those, since contexts that share no question have nothing to be the same about, and $\{A_0,B_0\}$ and $\{A_1,B_1\}$ share none. Every pair sharing a question is tied, each context to two others, so the four form a ring. But the four ties are different records, and none answers more than the measurements of its own two contexts. The record tying the two contexts that contain $A_0$, for instance, answers one pair in each of them, $(A_0,B_0)$ and $(A_0,B_1)$, with the same answer to $A_0$ in both, so three measurements and never $A_1$. 
A full account of a round answers fewer still, just the two measurements of its own context. Nothing answers all four measurements at once, so the contexts are never gathered into a single account.


Carrying sameness this thinner way changes how the probabilities combine. In standard classical probability, the chances of the mutually exclusive outcomes of one measurement add up. When Alice measures along $Z$, the electron comes out either up or down, and the two chances sum to one. Up and down are alternatives within one and the same situation, so adding their chances is exactly right. Across two incompatible contexts that reason is gone. The measurements $A_0$ and $A_1$, or $B_0$ and $B_1$, cannot even be performed together, so their outcomes are not alternatives within any one situation. Within each context the chances add as usual, since there the measurements are compatible; across contexts, nothing requires them to, and in fact additivity fails. This is what we give up of classical probability, and it is given up exactly where no experiment can object. To check additivity across two contexts one would have to perform them together, which is what incompatibility forbids. The chances stay non-negative, and within each context they still add to one.


The failure of additivity across incompatible contexts opens room for stronger correlations. The chances of incompatible contexts do not add into a single total, and that is the same fact as there being no single description covering them all. Bell's limit constrains such a description, so where there is none it does not apply. Notice that the shared record is not what breaks additivity; on the contrary, it is the only thing that brings two contexts closer to a single description. Classical probability combines two descriptions by counting what they share once, so if what two contexts share has chance $s$, their combined total is $1+1-s$. This equals~$1$, as a single description requires, only if $s=1$, when the two share everything, as all contexts do in a plan. With no shared record $s=0$ and the total is two, and a shared record brings it closer to one by exactly its chance.

This is also why records are needed and why a plan is too much. Without records the four contexts would be four unrelated descriptions, with nothing carried over from the shared past and so nothing explained; with a plan they would be a single description, and Bell's limit would apply. Records sit between the two: enough sameness for the shared past to explain the correlations, never enough to make the four into a single description.

Bell's limit arises from the ring the contexts form. Suppose we want the answers of Alice and Bob to follow a pattern, agreeing in three contexts and disagreeing in the fourth. No list of answers can do this. Each shared measurement carries its answer from one context to the next, so agreement in three contexts around the ring forces agreement in the fourth. A plan therefore satisfies the pattern in at most three contexts in four, and that is the limit. Records carry an answer only between the two contexts they tie, never around the whole ring, so nothing forces the fourth context to agree, and the four contexts, each with a unit of chance of its own, can satisfy the pattern in full. A suitable choice of quantum state and incompatible measurements satisfies it more often than any plan allows. Yet, given the record, Alice's outcome depends only on her own measurement and Bob's only on his, and neither depends on what is done in the distant region.


What had seemed to demand a spooky connection between Alice and Bob rested on the universal list a plan assumes. Without that list, the chances of incompatible contexts no longer add into a single description, and correlations stronger than any plan allows become possible. The conclusion that tempted us---a faster-than-light influence between the two regions---no longer follows. From their past Alice and Bob carry no list of answers fixed in advance. They carry a web of records, tying contexts two at a time without gathering them into one, and that is enough to explain the correlations entirely.

\smallskip

The remainder of the paper translates this story into mathematics.

\section{Background}
\label{sec:background}

This section makes mathematically precise the two assumptions from which Bell's theorem proceeds, the standard ontological models framework and local causality, and states the theorem itself. Both were already at work in the simple story. The framework is where the story's notion of a shared past lives. What Alice and Bob share from their past is carried by a physical system, taken to possess properties of its own, and the framework represents the complete specification of these properties as a point of a measurable space, with all of the properties determinate at once. That specification is the common cause of the friends' correlated answers. This all-at-once determination is what makes the story's universal list available. Local causality is the other half of the story. It rules out influence between the two distant regions and asks the shared past to account for the correlation on its own, so that conditioning on it leaves none. (Recall that measurement independence, ruling out retrocausality and superdeterminism, is also built into the ontological models framework.) The section closes with the no-go theorem, the limit the two assumptions jointly impose on the correlations and its violation by quantum theory.

\subsection{Standard ontological models framework}
\label{sec:standard_OMF}

The programme of explaining the quantum statistics through underlying variables is as old as quantum theory itself. Its most influential argument is due to Einstein, Podolsky and Rosen, who concluded, from the correlations of an entangled pair, that the quantum-state description of physical systems is incomplete, i.e., that there are elements of reality it fails to capture~\cite{EPR}. Von Neumann had offered a proof that no completion by ``hidden variables'' is possible~\cite{VonNeumann}, but an assumption in it that hidden variable models have no reason to satisfy was identified by Hermann in 1935~\cite{CrullBacciagaluppi2016} and, independently, by Bell~\cite{Bell1966}. Bohm settled the question of existence by exhibiting a hidden-variable model, anticipated by de Broglie, reproducing the predictions of nonrelativistic quantum mechanics~\cite{Bohm1952}; the model, however, is manifestly nonlocal. Bell's theorem~\cite{Bell_1964,bell1976} established that this nonlocality is no artefact of Bohm's construction: no hidden-variable model satisfying local causality reproduces the quantum statistics.
The modern formulation of hidden-variable models, which we adopt, is the \emph{ontological models framework}~\cite{Spekkens2005,Harrigan2010}. The mathematical object is unchanged; the new name is more apt, for the underlying states of a theory need not be hidden. In classical Hamiltonian mechanics, for example, they are the points of phase space, in principle perfectly accessible.

What notion of reality does the framework encode? An ontological model attempts to explain the success of a physical theory by positing that the experiment is performed on physical systems presumed to possess properties of their own, regardless of whether any experiment probes them and regardless of what anyone knows about them; the \emph{ontic state} of a system is a full specification of its properties. Two features of this notion of reality deserve emphasis, because they are precisely the ones the present work modifies. First, it is not a relational notion: the elements of reality are individual units, each individuated on its own, independently of any other and of any whole they may belong to. Relations among them, where present, are external to what they are. Second, and relatedly, the framework posits determinate facts without requiring that they be registered anywhere. Not only, as already said, independently of measurement and of anyone's knowledge, but without anything in the universe standing in any relation to them---a system possesses its properties even if no record of them exists.

The ontological models framework provides a mathematical representation of this notion of reality. The ontic state of a system is represented by a \emph{point} $\lambda$; the collection of all possible ontic states is a \emph{set} $\Lambda$, the \emph{ontic state space}. A property holds or fails at each individual ontic state, and is therefore represented by a measurable part of the state space, the \emph{subset} of $\Lambda$ collecting exactly the ontic states at which it holds. Such a part is fixed \emph{extensionally}, i.e., by nothing but which points belong to it. The logic of properties is accordingly the Boolean algebra of measurable subsets:\footnote{Here and throughout, \emph{Boolean}, \emph{intuitionistic} and \emph{paraconsistent} describe the algebra of parts of the ontic state space, read as a logic of propositions about the ontic state (App.~\ref{Appendix:Paraconsistency}).} each point settles every property, and all of them at once---$\lambda\in U$ or $\lambda\in\Lambda\setminus U$, exactly one of the two. This is the sense in which an ontic state is a complete specification of the properties of a system: it assigns a truth value to every one of them. Uncertainty about which ontic state the system is in is represented by a probability measure on $\Lambda$, obeying the Kolmogorov axioms; its finite additivity is equivalently expressed by the \emph{modular law}
\begin{equation}
\label{eq:modularity}
p(U\cup V)+p(U\cap V)\;=\;p(U)+p(V)
\end{equation}
for all measurable parts $U,V$, of which additivity on disjoint parts is the special case $U\cap V=\varnothing$. We display the law because it will be central to what follows. Point-set structure, a Boolean algebra of parts, Kolmogorov probability: this is the mathematics grounding an ontological model.

The framework is theory-independent: it presupposes no particular physical theory, but applies to any theory admitting a formulation in operational terms. Specifically, the operational description of an experimental scenario\footnote{For simplicity, we here consider prepare-and-measure scenarios (transformations can also be included, but they are not needed for our purposes, and we omit them).} consists of a set $\mathcal{P}$ of preparations, a set $\mathcal{M}$ of measurements with outcomes $k$, and the probabilities $\wp(k|P,M)$ of obtaining the outcome $k$ when the measurement $M\in\mathcal{M}$ is performed on a system prepared according to $P\in\mathcal{P}$;\footnote{Similarly to~\cite{bong2020strong}, we maintain a notational distinction between empirical probability distributions, which can be obtained by a single agent from the collected data and are denoted by $\wp$, and theoretical probability distributions, which are required to exist given certain assumptions but are potentially not observable by any agent, and are denoted by $p$ (or by $\mu$, $\xi$ for the distinguished roles introduced below).} a specification of this kind is what is standardly called an \emph{operational theory}.

\begin{definition}[Ontological model]
\label{def:ontological_model}
An \emph{ontological model} of a physical theory, specified operationally as above, consists of:
\begin{enumerate}[label=(\roman*)]
\item a measurable space $\Lambda$, the \emph{ontic state space}, whose elements $\lambda\in\Lambda$, the \emph{ontic states}, each represent a complete specification of the physical properties of the system;
\item for every preparation $P\in\mathcal{P}$, a probability distribution $\mu(\lambda|P)$ on $\Lambda$, the \emph{epistemic state} representing $P$;
\item for every measurement $M\in\mathcal{M}$ and outcome $k$, a conditional probability distribution $\xi(k|\lambda,M)$, the \emph{response function} representing $M$: the probability of obtaining $k$ when $M$ is performed on a system in the ontic state $\lambda$, whatever the preparation that produced it;
\end{enumerate}
such that the empirical predictions are reproduced via the classical law of total probability,
\begin{equation}
\label{eq:opdata}
\wp(k|P,M)\;=\;\sum_{\lambda\in\Lambda}\xi(k|\lambda,M)\,\mu(\lambda|P).
\end{equation}
\end{definition}
Note also that the epistemic state depends on the preparation alone: writing $\mu(\lambda|P)$ rather than $\mu(\lambda|P,M)$ already encodes \emph{measurement independence}, the absence of superdeterministic and retrocausal mechanisms, as part of the framework rather than as a separate assumption~\cite{CataniLeifer2023}.

The epistemic states are so called because they represent states of knowledge about which ontic state the system is in, in contrast with the ontic states themselves. The ontic state spaces employed in this work are finite, so we write sums; in general, Eq.~\eqref{eq:opdata} is an integral over $\Lambda$. We also use, as is customary in the Bell literature, the generic letter $p$, with the arguments making the role clear: $p(\lambda)$ for the epistemic state of a fixed preparation, $p(a|x,\lambda)$ for a response function.

\subsubsection*{Ontological models of the Bell scenario}
\label{sec:standard_bell}

In the Bell scenario, two spacelike-separated parties, Alice and Bob, each hold one part of a bipartite system prepared, in their common past, according to a preparation $P$, and each operates on their own side, or \emph{wing}, of the experiment. Alice chooses a setting $x\in\{0,1\}$ and performs the corresponding measurement $A_x$, obtaining an outcome $a\in\{0,1\}$; Bob chooses $y\in\{0,1\}$ and performs $B_y$, obtaining $b\in\{0,1\}$. Repeating the protocol yields the empirical statistics $\wp(a,b|x,y)$. These statistics are \emph{no-signalling} when each party's marginal does not depend on the other party's setting,
\begin{equation}
\label{eq:nosignalling}
\textstyle\sum_b \wp(a,b|x,y)=\wp(a|x),\qquad
\textstyle\sum_a \wp(a,b|x,y)=\wp(b|y),
\end{equation}
as they are in quantum theory, and as any statistics compatible with relativistic causality at the operational level must be.

As in the simple story, each round realizes a \emph{context}, the set of measurements jointly performed in that round (here one of Alice's and one of Bob's), so the contexts are the four sets $\{A_x,B_y\}$, which we label by the pair of settings $(x,y)$. Two measurements are \emph{incompatible} when they cannot be jointly performed: there is no single measurement whose outcome statistics return those of both, by marginalization, for every preparation. Two contexts are \emph{incompatible} when they cannot be jointly realized in a single round, which is the case exactly when, among the measurements by which the two contexts differ, some pair is incompatible (measurements on different wings are always jointly performable, one on each side). Throughout we consider the Bell scenarios in which Bell inequalities can be violated. There each party's two measurements are incompatible, so the four contexts are pairwise incompatible. They are not, for that, unrelated: two contexts differing in a single setting, which we call \emph{neighbouring}, share one measurement; for instance, $\{A_0,B_0\}$ and $\{A_0,B_1\}$ share Alice's $A_0$ and differ by the incompatible pair $B_0,B_1$, which is why they can never occur in the same round. The two \emph{diagonal} pairs of contexts, differing in both settings ($\{A_0,B_0\}$ with $\{A_1,B_1\}$, and $\{A_0,B_1\}$ with $\{A_1,B_0\}$), share no measurement. This pattern of sharing is displayed in Fig.~\ref{fig:standardOM_Bell}.


What must an ontological model of the physical theory associated with the Bell scenario (hereafter, simply, ontological model of the Bell scenario) supply? Exactly the data of Definition~\ref{def:ontological_model}: an ontic state space for the bipartite system, an epistemic state representing the preparation, and response functions for the measurements of the two wings, reproducing $\wp(a,b|x,y)$ through Eq.~\eqref{eq:opdata}. In the present scenario, the measurement independence already built into Definition~\ref{def:ontological_model} reads as independence from the choice of settings,
\begin{equation}
\label{eq:measurement_independence}
\mu(\lambda|P,x,y)\;=\;\mu(\lambda|P)\;\equiv\;p(\lambda).
\end{equation}
The analysis of locally causal models can be conducted, without loss of generality, on a single canonical model, and the precise sense in which nothing is lost by this choice is Fine's theorem, recalled in the next subsection. Its ontic state space collects the sixteen deterministic assignments of outcomes to the four measurements,
\begin{equation}
\label{eq:standard_Lambda}
\Lambda:=\{0,1\}^{4},\qquad \lambda=(a_0,a_1,b_0,b_1),
\end{equation}
where $a_x$ is the outcome assigned to Alice's measurement $A_x$ and $b_y$ the outcome assigned to Bob's $B_y$. The response function of the context $(x,y)$ is outcome-deterministic (valued in $\{0,1\}$) and simply reads the assignment off:
\begin{equation}
\label{eq:standard_responses}
\xi(a,b|\lambda,x,y)\;=\;\delta_{a,a_x}\,\delta_{b,b_y}.
\end{equation}
An ontic state of this model is, then, exactly what the simple story called a \emph{plan}: a universal list assigning outcomes to all four measurements at once, including pairs of them no single round can pose together.
An epistemic state $p(\lambda)$ is a mixture of plans. By Eqs.~\eqref{eq:opdata} and~\eqref{eq:standard_responses}, the empirical statistics are reproduced as
\begin{equation}
\begin{split}
\label{eq:standard_reproduction}
\wp(a,b|x,y)\;&=\;\sum_{\lambda\in\Lambda}p(\lambda)\,\delta_{a,a_x}\,\delta_{b,b_y}
\\ &=\;\sum_{\lambda\in\Lambda^{xy}_{ab}}p(\lambda)\;=:\;p\big(\Lambda^{xy}_{ab}\big),
\end{split}
\end{equation}
where
\begin{equation}
\label{eq:sectors}
\Lambda^{xy}_{ab}:=\{\lambda\in\Lambda:\ a_x=a,\ b_y=b\}
\end{equation}
is the \emph{outcome sector}, or simply \emph{sector}, of the outcome pair $(a,b)$ in the context $(x,y)$. It is the set of plans that answer $(a,b)$ when the context $(x,y)$ is realized. Each sector contains exactly four of the sixteen points and, for fixed $(x,y)$, the four sectors partition $\Lambda$: the sectors of a context are disjoint and their weights add, an instance of the modular law~\eqref{eq:modularity}.

Across contexts the picture is different, as sectors of neighbouring contexts \emph{overlap}. For the same outcome pair,
\begin{equation}
\label{eq:overlap}
\begin{gathered}
\Lambda^{00}_{00}=\{(0,a_1,0,b_1):a_1,b_1\in\{0,1\}\},\\
\Lambda^{01}_{00}=\{(0,a_1,b_0,0):a_1,b_0\in\{0,1\}\},\\
\Lambda^{00}_{00}\cap\Lambda^{01}_{00}=\{(0,a_1,0,0):a_1\in\{0,1\}\}\neq\varnothing.
\end{gathered}
\end{equation}
More generally, the region where Alice's $A_x$ yields the outcome $a$, namely $\{\lambda:a_x=a\}$, is one and the same subset of $\Lambda$ regardless of which measurement Bob performs. In the standard ontological model, the sameness of a measurement across the two contexts it belongs to is thus set-theoretic identity of regions, carried by shared points. Nothing in the model has to arrange it. It is automatic: a point of $\{0,1\}^4$ is, as a single element of reality, a global assignment, i.e., a \emph{global valuation} fixing outcomes for all four measurements; so a point shared by two sectors is not merely shared by those two, but it lies in one sector of every context, and thereby answers all four measurements at once. This is the mathematical form of the observation made in the simple story: in a plan, the sameness of $A_0$ across its two contexts is purchased together with something much stronger---a single element of reality belonging to all four contexts at once.
Figure~\ref{fig:standardOM_Bell} shows the sectors, their overlaps, and the shared points that carry them.

Finally, the response functions~\eqref{eq:standard_responses} already exhibit the key feature: the joint response factorizes into a factor for each wing, each depending only on the local setting---Alice's on $x$ alone, Bob's on $y$ alone. This factorization, imposed on ontological models of the scenario, is the content of local causality, which we now formalize.

\begin{figure*}[t]
    \centering
    \includegraphics[width=\textwidth]{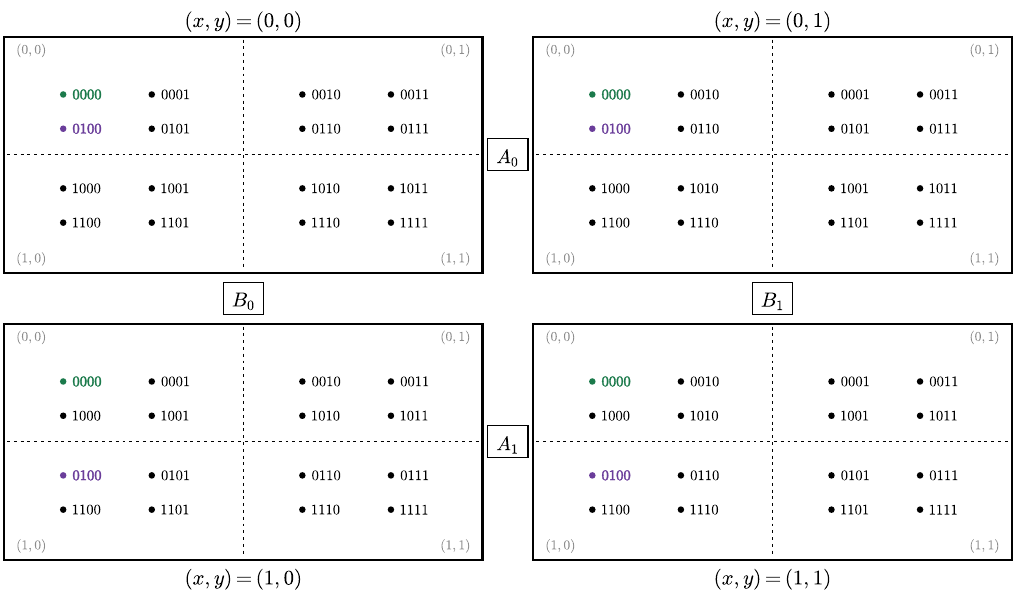}
    \caption{\textbf{The standard ontological model of the Bell scenario.} The four measurement contexts $(x,y)$ are arranged in a $2\times 2$ grid following the choice of settings by Alice (rows) and Bob (columns). Neighbouring contexts share the measurement labelling the edge between them; diagonal contexts share none. Within each context, the four outcome sectors are arranged in a further $2\times 2$ sub-grid indexed by the outcome pair $(a,b)$, which labels its outer corner in grey; the sector labelled $(a,b)$ in the context $(x,y)$ is $\Lambda^{xy}_{ab}$ of Eq.~\eqref{eq:sectors}. The ontic states are the sixteen points $\lambda=(a_0,a_1,b_0,b_1)$ of $\Lambda=\{0,1\}^4$, the plans of the simple story, each fixing outcomes for all four measurements at once. Within a context the four sectors partition $\Lambda$, so every point appears once in every context and four times in the figure, the four copies of one string denoting one and the same point. Across neighbouring contexts the sectors overlap, as the two points of $\Lambda^{00}_{00}\cap\Lambda^{01}_{00}$ of Eq.~\eqref{eq:overlap} show, highlighted in green ($0000$) and purple ($0100$). The first occupies the same sector in all four contexts and the second does not, since which sector a point falls in depends on the context. Cross-contextual sameness of a measurement is thus carried by shared points, and every shared point is a global assignment. Figure~\ref{fig:localicOM_Bell} is the localic counterpart of this picture.}
    \label{fig:standardOM_Bell}
\end{figure*}

\subsection{Bell's local causality}
\label{sec:local_causality}

In the simple story, locality entered as a statement about causes. Once the friends separate, nothing that happens in one region influences the other, and whatever correlates their answers must already be contained in what they share. Bell's local causality is the formalization of this statement within the ontological models framework. It combines relativistic causality, which places the causes of an event in its past light cone, with Reichenbach's common cause principle~\cite{Reichenbach1956}, which traces every correlation to a direct causal link or to a common cause, conditioning on which the correlation disappears. In the Bell scenario the direct link is excluded by the spacelike separation of the wings, and the common cause is the ontic state $\lambda$ of the shared system (Fig.~\ref{fig:dags}, left). By measurement independence, Eq.~\eqref{eq:measurement_independence}, $\lambda$ cannot involve the settings. Local causality then demands that $\lambda$ does the whole explanatory job: once $\lambda$ is completely specified, no correlation between the wings remains.
\begin{definition}[Local causality]
\label{def:local_causality}
An ontological model of the Bell scenario is \emph{locally causal} if
\begin{equation}
\label{eq:factorizability}
p(a,b|x,y,\lambda)\;=\;p(a|x,\lambda)\,p(b|y,\lambda)
\qquad\text{for all } a,b,x,y,\lambda.
\end{equation}
\end{definition}
Summing Eq.~\eqref{eq:factorizability} over one party's outcome shows that, given the ontic state, each party's outcome distribution depends on their own setting alone,
\begin{equation}
\label{eq:parameter_independence}
p(a|x,y,\lambda)=p(a|x,\lambda),
\qquad
p(b|x,y,\lambda)=p(b|y,\lambda),
\end{equation}
a condition known as \emph{parameter independence}~\cite{Jarrett1984,Shimony1986}.

Definition~\ref{def:local_causality} reflects Bell's mature formulation of the condition, phrased by him in terms of \emph{local beables}---the elements a theory posits as existing, as opposed to what is observed, each attached to a spacetime region. Once the beables in the backward light cone of an event are specified, its probability is unaltered by anything spacelike separated from it~\cite{bell1976,BellUnspeakable}. His 1964 paper~\cite{Bell_1964} instead proceeded via determinism together with locality; the relation between the two formulations is analysed in~\cite{Wiseman2017}. Inserting Eq.~\eqref{eq:factorizability} into the reproduction formula~\eqref{eq:opdata}, the correlations of a locally causal model take the form
\begin{equation}
\label{eq:locality}
    \wp(a,b|x,y) = \sum_\lambda p(\lambda)\, p(a|x,\lambda)\,p(b|y,\lambda);
\end{equation}
conversely, correlations are consistent with local causality precisely when they admit such a decomposition.\footnote{Notice that no-signalling, Eq.~\eqref{eq:nosignalling}, is an automatic consequence of Eq.~\eqref{eq:locality}: summing over $b$ leaves $\sum_\lambda p(\lambda)\,p(a|x,\lambda)$, in which $y$ does not appear.}  This is the story's claim in formulas: given the shared past, Alice's answer depends on her question alone, and likewise for Bob.

Which correlations admit such a decomposition is settled by Fine's theorem.
\begin{theorem}[Fine~\cite{Fine,Fine1982}]
\label{thm:fine}
For the correlations $\wp(a,b|x,y)$ of the Bell scenario, the following are equivalent:
\begin{enumerate}[label=(\roman*)]
\item $\wp(a,b|x,y)$ admits of a locally causal ontological model;
\item there exists a joint probability distribution $p(a_0,a_1,b_0,b_1)$ over the outcomes of all four measurements, returning the correlations as marginals;
\item $\wp(a,b|x,y)$ admits of a locally causal ontological model on the canonical state space~\eqref{eq:standard_Lambda}, with the deterministic response functions~\eqref{eq:standard_responses}.\footnote{Only the responses are deterministic, not the epistemic state. Fine's construction absorbs the randomness of the responses into the mixture over $\lambda$, which must therefore remain free to carry it.}
\end{enumerate}
\end{theorem}
The equivalence of (ii) and (iii) is a change of notation, since a joint distribution \emph{is} an epistemic state on the canonical state space, $p(\lambda)=p(a_0,a_1,b_0,b_1)$, and marginalizing over the unperformed measurements is summing over a sector~\cite{Fine,abramsky2011sheaf}, for example,
$\wp(a,b|x{=}0,y{=}0)=\sum_{a_1,b_1}p(a,a_1,b,b_1)=p\big(\Lambda^{00}_{ab}\big)$,
which is Eq.~\eqref{eq:standard_reproduction}. The substantive implication is (i)$\Rightarrow$(ii): allowing the responses to be stochastic does not enlarge the set of correlations a locally causal model can reproduce. This is the precise sense in which the canonical model of Sec.~\ref{sec:standard_OMF} involves no loss of generality. Within the standard framework, a locally causal explanation of the correlations can always be taken to be a mixture of plans. Whether the quantum correlations admit of such an explanation is the question Bell answered in the negative, to which we now turn.

\subsection{Bell's no-go theorem}
\label{sec:bell_nogo}

The limit the two assumptions jointly impose on the correlations is a bound on the Clauser--Horne--Shimony--Holt (CHSH) combination~\cite{CHSH69} of the correlators,
\begin{equation}
\begin{split}
\label{eq:CHSH_S}
& E_{xy}:=\sum_{a,b}(-1)^{a+b}\,\wp(a,b|x,y),
\\
& S:=E_{00}+E_{01}+E_{10}-E_{11}.
\end{split}
\end{equation}

\begin{theorem}[Bell--CHSH inequality~\cite{Bell_1964,CHSH69}]
\label{thm:chsh}
If $\wp(a,b|x,y)$ admits of a locally causal ontological model, then $|S|\le 2$.
\end{theorem}

\begin{proof}
By Eq.~\eqref{eq:locality}, set $A_x(\lambda):=\sum_a(-1)^a\,p(a|x,\lambda)$ and $B_y(\lambda):=\sum_b(-1)^b\,p(b|y,\lambda)$, both in $[-1,1]$ as averages of $\pm1$. Then $E_{xy}=\sum_\lambda p(\lambda)\,A_x(\lambda)B_y(\lambda)$, so that $S=\sum_\lambda p(\lambda)\,S(\lambda)$ with, suppressing the argument $\lambda$,
\begin{equation*}
\begin{split}
S(\lambda)&=A_0B_0+A_0B_1+A_1B_0-A_1B_1 \\ &=A_0(B_0+B_1)+A_1(B_0-B_1).
\end{split}
\end{equation*}
Since $|A_0|,|A_1|\le1$ and $|B_0+B_1|+|B_0-B_1|=2\max(|B_0|,|B_1|)\le2$, we have $|S(\lambda)|\le2$ for every $\lambda$, and $|S|\le2$ follows because $p$ is a probability distribution.
\end{proof}

Notice that no determinism was assumed in reaching the bound, given that the responses $A_x,B_y$ may take any value in $[-1,1]$. In the canonical model of Sec.~\ref{sec:standard_OMF} they are $A_x(\lambda)=(-1)^{a_x}$ and $B_y(\lambda)=(-1)^{b_y}$, so $S(\lambda)=\pm2$ exactly: each plan scores the extreme value, and mixing plans cannot do better. What the derivation requires, in either form, is that one and the same $\lambda$ supply all four responses $A_0,A_1,B_0,B_1$ at once.

The CHSH expression can equivalently be read as the success probability of a game in which the parties win a round when $a\oplus b=xy$ ($\oplus$ denoting addition mod~$2$). With uniformly random settings,
\begin{equation}
\label{eq:pwin}
p_{\mathrm{win}}:=\frac14\sum_{x,y}\,\sum_{\substack{a,b:\\ a\oplus b=xy}}\wp(a,b|x,y)
\;=\;\frac12+\frac{S}{8},
\end{equation}
so Theorem~\ref{thm:chsh} reads $p_{\mathrm{win}}\le 3/4$. The parity identity $\bigoplus_{x,y}(a_x\oplus b_y)=0\neq 1=\bigoplus_{x,y}\,xy$ shows directly that no plan wins all four contexts.

Quantum theory violates the bound. Two qubits in a maximally entangled state, with suitable choices of the four measurements, yield $S=2\sqrt{2}$ (the maximal quantum value~\cite{Cirelson1980}), that is, $p_{\mathrm{win}}=\cos^2(\pi/8)\approx0.854$, and the violation has been confirmed experimentally~\cite{Hensen,Shalm2015,Giustina2015}. Quantum theory is therefore said to manifest \emph{nonlocality}: it does not admit of a locally causal ontological model~\cite{Bell_1964}. By Theorem~\ref{thm:fine}, this is equivalently the statement that no joint probability distribution over the outcomes of the four measurements returns the quantum correlations as its marginals.

The bound of Theorem~\ref{thm:chsh} is not the largest value $S$ can take. Since $|E_{xy}|\le1$, the algebraic maximum is $|S|=4$, and it is attained without signalling. The Popescu--Rohrlich box~\cite{Popescu1994}, $\wp(a,b|x,y)=\tfrac12$ when $a\oplus b=xy$ and $0$ otherwise, has uniform marginals on both wings and wins the game in every context, $p_{\mathrm{win}}=1$. Quantum theory therefore sits strictly between the local bound $2$ and the algebraic maximum $4$.

Let us record where each assumption did its work. Local causality itself entered only once, as the factorization~\eqref{eq:factorizability}. Everything else was supplied by the framework: that one and the same ontic state bears the four responses $A_0,A_1,B_0,B_1$---in the canonical model, a plan; and that the four contexts are averaged with one and the same weight $p(\lambda)$, additive in the sense of Eq.~\eqref{eq:modularity}. The theorem therefore bounds the conjunction of local causality with the standard ontological models framework. The usual response to the violation keeps the framework and gives up local causality; we keep local causality and measurement independence, and modify the framework instead.

\section{The localic ontological models framework}
\label{sec:localic_framework}

Section~\ref{sec:standard_OMF} singled out two features of the notion of reality encoded in the standard ontological models framework as the ones this work modifies. The first is that its elements of reality are individuated each on its own, independently of any whole they belong to. The second is that they carry determinate facts without those facts being registered anywhere. We replace both, and with them the mathematics that carried them, namely point-set structure, a Boolean algebra of parts and Kolmogorov probability. The elements of reality become the records of the simple story, represented mathematically by the parts of a locale.
In the same order as Sec.~\ref{sec:standard_OMF}, Sec.~\ref{sec:records} motivates and sets out the notion of reality, Sec.~\ref{sec:pointfree} introduces the mathematical structure that represents it, and Sec.~\ref{sec:localic_OMF} defines a localic ontological model of an operational theory. Two independent consequences then follow: a record shared by two incompatible contexts contains no points (Sec.~\ref{sec:pointlessness}), and the weights of incompatible contexts admit no global account that is both modular and normalized unless the contexts share the whole of their weight (Sec.~\ref{sec:weight}). The Bell scenario comes only in Sec.~\ref{sec:localic_bell}, and nothing before it is special to Bell. Like the framework it generalizes, ours applies to any operational theory, and its axioms and theorems are stated for an arbitrary scenario. Desiderata~(i) and~(ii) of the introduction, a precise notion of reality carried by a rigorous mathematical structure and theory-independence, are therefore met here. Desideratum~(iii), reproducing the quantum statistics without fine-tuning, is the business of Sec.~\ref{sec:go}.

\subsection{A reality made of records}
\label{sec:records}

Is there any difference, in the world, between a property that is nowhere recorded and no property at all? The standard framework says there is a difference: a system possesses its properties even with nothing in the universe standing in any relation to them. We say there is none: for something to exist is for something to be recorded. We call the elements of reality \emph{records}, and no observer and no agent enters into what it is to be one.

Records register what occurred in a round of an experiment, both the procedures performed on a system and the outcomes they gave, and may register the whole of a round or only a part of it. Treating the procedures as part of what occurred is the step the standard framework does not take. There an ontic state specifies the system alone. A measurement enters only as a further argument of the response functions (Definition~\ref{def:ontological_model}), which are defined at every ontic state for every measurement of the theory, performed or not. What is described is in either case a single round, and the procedures that could have been performed instead are alternative ways for that round to come about, not further rounds. Registering the procedures does not make what exists depend on which of them are chosen. What records there could be and how they stand to one another are fixed once and for all, and what the procedures of a round settle is only which of those records that round realizes.

Records are of two kinds. Some register everything that occurred in a round, procedures and outcomes. No record of this kind is common to two rounds in which different procedures were performed, since it would then register both. Others register only a part of a round, the identity of a procedure across the rounds in which it could have been performed. Such a record registers a given measurement as one and the same measurement, and says nothing of which other measurements the round performed. It is a constituent of the round that occurred, and would equally have been a constituent of a round in which that same measurement was performed together with a different one. Being the very same record in both is what the sameness of a measurement across rounds consists in. Such records matter most where the two rounds realize incompatible contexts, for only they can then be common to both. Unlike a record of the first kind, one of the second does not by itself fix the outcome of the whole context, which is fixed by it together with what the round performed. In the simple story it was records of this second kind that mattered, and we only spoke of those, but the word covers every element of reality, records of the first kind included.

The sameness of a measurement in every context containing it is something any ontology must account for. A measurement appearing in two contexts yields the same statistics in both, and if nothing in the ontology corresponds to that sameness, the agreement is an unexplained coincidence. This is particularly important when the correlations are to be given a common-cause explanation, as they are here. To explain correlations across contexts by a common cause is to explain how one and the same measurement behaves in each of them. If the contexts have nothing in common there is no common cause to speak of, the ontology falls apart into as many unrelated worlds as there are contexts, and the demand that a cause screen the outcomes off constrains nothing across them (Sec.~\ref{sec:localic_local_causality}). The demand for such sameness is not peculiar to our framework. The standard one meets it automatically, since the region where a given measurement yields a given outcome is one and the same subset of $\Lambda$ in every context containing that measurement (Sec.~\ref{sec:standard_bell}).

The sharing between contexts there is carried by the points those contexts have in common, and a reality made of records cannot be shared that way. A point is a complete specification of a round, so a point serving as a record would register which measurements were performed and what they gave. A record of the second kind registers a measurement and not which others the round performed, and leaving this out is precisely what lets it belong to two rounds at once. No point is a record of the second kind, and where the elements of reality are points there are no such records to be had.

So not all elements of reality can be points. What, then? To register only part of what occurred in a round is to stand to that round as a part does to a whole, and what a record is, is fixed by which part it registers. The elements of reality are thus not self-standing units but parts, individuated by how they restrict a whole, and the relations by which an element is individuated are internal to what it is. Our notion of reality is relational in exactly this sense. Of the two features of the standard framework this work modifies (Sec.~\ref{sec:standard_OMF}), one was that its elements carry facts registered nowhere; relationality is the other, and it costs no further assumption, following from the requirement that whatever is real be recorded. It also suggests how the sameness of a measurement might survive: two contexts could share a part without sharing any point of it. Whether a mathematics of parts is available that delivers this, parts that are not empty and yet contain no point, is what the next subsection~answers.

\subsection{From points to parts}
\label{sec:pointfree}

A reality made of records asks for a mathematics of parts rather than points. In the framework of Sec.~\ref{sec:standard_OMF}, the ontic state space is a set of points, and its parts, also made of points, are measurable subsets. Topology is standardly built in the same way, from a set of points of which some subsets are declared \emph{open}, but it can also be built directly from the opens themselves. The usual motivation for doing so, different from ours, is operational. No experiment pins down a point, since an apparatus has a finite resolution and what it detects is a region, a point being only the idealization of a region shrunk below every resolution~\cite{Vickers1989}.

The opens of a topological space are ordered by inclusion, and in that order the union of any family of them is its least upper bound and the intersection of a finite family is its greatest lower bound. Both bounds are defined by the order alone, so they survive when the points are forgotten, and they are then called \emph{join} ($\vee$) and \emph{meet}~($\wedge$). An ordered collection in which every family has a join, and hence also a meet, is a \emph{complete lattice}; the join of the empty family is its least element $0$, which here is the empty set. A complete lattice is a \emph{frame} when finite meets distribute over arbitrary joins,
\begin{equation}
\label{eq:frame}
U\wedge\bigvee_{i\in I}V_{i}\;=\;\bigvee_{i\in I}\,(U\wedge V_{i})
\end{equation}
for every element $U$ and every family $\{V_{i}\}_{i\in I}$, finite or not, a law the opens of a space obey because intersection distributes over union. A frame read as a space whose points are not presupposed is a \emph{locale}~$L$. Frame and locale are one and the same object, and we write $\mathcal{O}(L)$ when the locale $L$ is regarded as a frame, a lattice of elements we keep calling opens. The theory of locales is \emph{point-free topology}~\cite{PicadoPultr2012,Vickers1989,Johnstone1983}.

Point-free topology generalizes point-set topology, the familiar theory of topological spaces and their subspaces, without discarding it. The opens of any topological space form a frame, so every topological space defines a locale, while not every locale arises in this way. Points survive too, as a derived notion modelled on the point-set case. In a topological space, a point $x$ decides every open $U$ by lying in it or not. It lies in a finite intersection exactly when it lies in each of the opens intersected, and in a union exactly when it lies in at least one of the opens united, however many they are. A \emph{point} of a locale is any way of deciding its opens with these two properties,\footnote{Equivalently, a point is a frame homomorphism $\mathcal{O}(L)\to\{0,1\}$, a map into the frame of truth values preserving finite meets and arbitrary joins, which gives for every open $U$ the yes-or-no answer to ``$x\in U$'' of a would-be point $x$.} and $\mathrm{pt}(L)$ denotes the set of points of~$L$. When any two distinct opens are told apart by some point, lying in one and not in the other, the locale is \emph{spatial}, and this happens exactly when its frame is the frame of opens of a topological space. At the other extreme, a locale can have no points at all without being empty (the empty locale is the one whose frame has $0$ as its only element).

The opens are not yet the parts we need, since even on a space they are only some of its subsets. A part of a locale cannot be a set of points, as there may be none to collect. It is given instead by a \emph{quotient} of the frame of the locale, obtained by declaring certain opens of the whole to be identified, compatibly with arbitrary joins and finite meets, and the quotient frame is the frame of the part. The opens encode the distinctions present in the whole, and restricting attention to a part amounts to giving up some of them. Opens that differ only in distinctions given up become one and the same open of the part. Naming what can no longer be told apart is thus naming the part. If nothing is identified the part is the whole, and if everything is, it is the empty part $0$. These quotients of $\mathcal{O}(L)$ are the parts of $L$, its \emph{sublocales},\footnote{A part is a locale $L'$ embedded in $L$, and its frame is a quotient of $\mathcal{O}(L)$ because maps of locales run opposite to maps of their frames. A continuous map $X\to Y$ sends each open of $Y$ back to its preimage, an open of $X$, so a map of locales $L'\to L$ is by definition a frame homomorphism $\mathcal{O}(L)\to\mathcal{O}(L')$, and it embeds $L'$ in $L$ exactly when this homomorphism is surjective, which makes $\mathcal{O}(L')$ a quotient of $\mathcal{O}(L)$. The identification of the main text is a \emph{frame congruence} on $\mathcal{O}(L)$. Equivalently, the part is given by a \emph{nucleus} $j\colon\mathcal{O}(L)\to\mathcal{O}(L)$ (inflationary, idempotent and preserving finite meets), sending each open to the largest open identified with it, whose fixed points are the opens of the part~\cite[Ch.~III]{PicadoPultr2012}.} and ordered by the relation of being a part they form a complete lattice $\mathsf{S}(L)$.

Sublocales thus meet the demand of Sec.~\ref{sec:records} that a part be individuated by how it restricts the whole. A subset of a space can be given by listing its points, with no mention of the whole it belongs to, and it is determined by them, which is what made the mereology of Sec.~\ref{sec:standard_OMF} extensional. A sublocale, by contrast, is a mode of restriction of the whole and cannot be given without it, so that what a part is and where it stands in the whole are one fact. Nor is a sublocale settled by its points, for two of them may have exactly the same points and still differ. The mereology of a locale is not extensional.

The lattice $\mathsf{S}(L)$ is a \emph{coframe}, in which finite joins distribute over arbitrary meets, and in general it is not a Boolean algebra. Every open $U\in\mathcal{O}(L)$ carries an open part $\mathfrak{o}(U)\in\mathsf{S}(L)$, which identifies two opens $V$ and $W$ exactly when $U\wedge V=U\wedge W$, giving up every distinction outside~$U$, and a closed part~$\mathfrak{c}(U)\in\mathsf{S}(L)$, which identifies them exactly when $U\vee V=U\vee W$, giving up every distinction inside it. The two are disjoint and together cover the whole, so the open and the closed subsets of point-set topology all have their counterparts. A part, however, can be neither open nor closed.

Indeed, a part can even be nonzero and contain no point, as Sec.~\ref{sec:records} asked, and the real line already carries the paradigmatic example. This part of the line identifies two opens exactly when they have the same closure, as with $(0,2)$ and $(0,1)\cup(1,2)$, which differ only at the point $1$. Each open is thereby identified with the interior of its closure, the largest open with the same closure, and the opens equal to the interior of their closure, called \emph{regular}, are the opens of the part. The \emph{negation} $\neg U$ of an open $U$ is the largest open disjoint from it (on a space, the interior of its complement), so the interior of the closure of $U$ is $\neg\neg U$, and we write the part as $\mathbb{R}_{\neg\neg}$. It is \emph{dense}, in that it identifies no nonempty open with the empty one, just as a dense subset meets every nonempty open, and so it is not the empty part. Moreover, by Isbell's density theorem it is the smallest dense part of the line~\cite{Isbell1972}.

Yet $\mathbb{R}_{\neg\neg}$ itself has no points at all. For any real number $x$, the punctured line $\mathbb{R}\setminus\{x\}$ is open and has the same closure as the whole line, so in $\mathbb{R}_{\neg\neg}$ the two are one and the same open. Since $x$ lies in the whole line and not in the punctured one, as a point of $\mathbb{R}_{\neg\neg}$ it would have to lie in that open and not lie in it. No real number, then, is a point of $\mathbb{R}_{\neg\neg}$.\footnote{More generally, for a Hausdorff space $X$ the points of the sublocale $X_{\neg\neg}$ of its regular opens are exactly the isolated points of $X$, and the line has none.} The line thereby exhibits the sharing anticipated in Sec.~\ref{sec:records}. Every subset $A$ of the line determines a sublocale, the locale that $A$ defines as a subspace, which identifies two opens exactly when they meet $A$ in the same set and whose points are those of $A$. The rationals and the irrationals have no point in common and each is dense in the line, so the sublocale each determines is dense and contains $\mathbb{R}_{\neg\neg}$, as every dense sublocale does. As sublocales, the rationals and the irrationals share~$\mathbb{R}_{\neg\neg}$, while as subsets they share nothing.

The whole, from here on, is a locale $L_\Lambda$, and the records of Sec.~\ref{sec:records} are its parts, the sublocales $L_\lambda\in\mathsf{S}(L_\Lambda)$, so that the coframe $\mathsf{S}(L_\Lambda)$ takes over the role of $\Lambda$ in Sec.~\ref{sec:standard_OMF}. The records we will consider are finitely many, but $L_\Lambda$ cannot be finite. Some of them will be nonzero and have no points (Sec.~\ref{sec:pointlessness}), while every nonzero part of a finite locale has one, since its frame is the frame of opens of a finite topological space~\cite{Birkhoff1937}. Two lattices are in play, and they should not be conflated. The frame $\mathcal{O}(L_\Lambda)$ is a complete Heyting algebra and its logic is intuitionistic; its elements are the opens, which encode the distinctions present in the whole. The coframe $\mathsf{S}(L_\Lambda)$ is a complete co-Heyting algebra and its logic is in general paraconsistent;\footnote{App.~\ref{Appendix:Paraconsistency} defines both algebras and discusses the logic.} its elements are the records, and it is the coframe that carries the ontology.

\subsection{Localic ontological models}
\label{sec:localic_OMF}
Given the notion of reality of Sec.~\ref{sec:records} and the mathematics of Sec.~\ref{sec:pointfree} apt to represent it, we now define a localic ontological model of an operational theory, mirroring Definition~\ref{def:ontological_model}. One of its ingredients has no counterpart there. In the standard framework an ontic state specifies the system alone, so it answers every measurement of the theory, performed or not. A record instead registers which experimental procedures were performed, measurements among them, and which records a round performing a given procedure can realize is settled before any round is run. A localic model therefore assigns to each procedure~$\pi$, preparation or measurement, a region~$L_{\Lambda^{\pi}}$, so that a round performing~$P$ and~$M$ realizes a record below ${L_{\Lambda^{PM}}:=L_{\Lambda^{P}}\wedge L_{\Lambda^{M}}}$. Here $\mathcal{M}$ collects the measurements a single round can perform, as in Definition~\ref{def:ontological_model}, and measurements performed together are represented by the one they jointly constitute.
The records below $L_{\Lambda^{PM}}$ are the alternatives of a round performing $P$ and $M$, and the epistemic state of $P$ behaves as a probability over them, one measurement at a time. Nothing is asked of it at the join of the regions of two measurements that cannot be performed together, since no round realizes such a join (Sec.~\ref{sec:weight}).

\begin{definition}[Localic ontological model]
\label{def:localic_model}
A \emph{localic ontological model} of a physical theory, specified operationally as in Sec.~\ref{sec:standard_OMF}, consists of:
\begin{enumerate}[label=(\roman*)]
\item a locale $L_\Lambda$ whose sublocales $L_\lambda\in\mathsf{S}(L_\Lambda)$, the \emph{records}, each represent an ontic state; among them, a region $L_{\Lambda^{\pi}}$ for every procedure $\pi$ of the theory, the one within which a round performing $\pi$ realizes a record, and for every measurement $M\in\mathcal{M}$ a finite family $\mathcal{A}_{M}$ of pairwise disjoint records joining to $L_{\Lambda^{M}}$;
\item for every preparation $P\in\mathcal{P}$, a weight $q(L_\lambda|P)\in[0,1]$, the \emph{epistemic state} representing $P$, defined on sublattices $\mathcal{L}_{PM}$ of the records below $L_{\Lambda^{PM}}$, one for every measurement $M\in\mathcal{M}$,\footnote{It is a sublattice, closed under joins and meets, because the modular law of Eq.~\eqref{eq:ctx_modularity} involves the weights of the join and the meet of any two records it relates. We ask of $\mathcal{L}_{PM}$ only that it contain $0$ and $L_{\Lambda^{PM}}$, its meets with the members of $\mathcal{A}_{M}$, over which Eq.~\eqref{eq:localic_opdata} integrates, and its meets with the regions of the other measurements, on which Sec.~\ref{sec:weight} relies.} and monotone, modular and normalized on each,
\begin{equation}
\label{eq:ctx_modularity}
\begin{split}
q(S\vee T|P)+q(S\wedge T|P) &=q(S|P)+q(T|P), \\
q(L_{\Lambda^{PM}}|P) &=1,\qquad q(0|P)=0,
\end{split}
\end{equation}
for all $S,T\in\mathcal{L}_{PM}$;
\item for every measurement $M\in\mathcal{M}$ and outcome $k$, a \emph{response function} $q(k|M,L_\lambda)\in[0,1]$ defined at the members of $\mathcal{A}_{M}$, with $\sum_{k}q(k|M,L_\lambda)=1$ at each of them;
\end{enumerate}
such that the empirical predictions are reproduced by integration over records,
\begin{equation}
\label{eq:localic_opdata}
\wp(k|P,M)\;=\;\int \mathrm{d}L_\lambda\;q(L_\lambda|P)\,q(k|M,L_\lambda).
\end{equation}
\end{definition}

\noindent We treat one preparation $P$ at a time and work within its region, replacing every region $L_{\Lambda^{M}}$ by $L_{\Lambda^{PM}}$ and every member $L_\lambda$ of $\mathcal{A}_{M}$ by $L_{\Lambda^{PM}}\wedge L_\lambda$, which keeps the response of $L_\lambda$. These meets are the records over which Eq.~\eqref{eq:localic_opdata} integrates. They are pairwise disjoint and join to $L_{\Lambda^{PM}}$, where the weight is normalized. With $P$ fixed we suppress it in the notation, writing $q(L_\lambda)$ for the weight and $L_{\Lambda^{M}}$, $\mathcal{L}_{M}$ and $\mathcal{A}_{M}$ for the region, the lattice and the family so obtained.

Outside its region a measurement has no response: not a value left unspecified, but no value at all, since no round performing $M$ realizes a record there. Within the region the responses and the weight are organized as in Definition~\ref{def:ontological_model}, with the members of $\mathcal{A}_{M}$ in place of the ontic states, at which the responses are given, and $\mathcal{L}_{M}$ in place of the measurable subsets, on which the weight is defined. We set the integrand of Eq.~\eqref{eq:localic_opdata} to zero at every record outside $\mathcal{A}_{M}$, so that it is defined on all of $\mathsf{S}(L_\Lambda)$, while the response itself stays undefined there.

Clause~(i) specifies the records, the regions and the families before any round is run, so that what a round settles is only which record it realizes, as in Sec.~\ref{sec:records}. The weight then carries no measurement index, so a record on which it is defined for two measurements receives one value for both, which is measurement independence, Eq.~\eqref{eq:measurement_independence}, holding by construction. Nor does clause~(ii) ask anything of a preparation that Definition~\ref{def:ontological_model} does not. There too the weight is normalized on the region of every measurement; the conditions are invisible only because those regions all coincide with $\Lambda$.

Equation~\eqref{eq:localic_opdata} is the law of total probability of Eq.~\eqref{eq:opdata}, transposed to records. Its integral is the Choquet integral~\cite{Denneberg1994}, which is defined for monotone weights that need not be modular, coincides with the ordinary integral for weights that are, and asks of the integrand only its upper level sets, the parts where it is at least a given value, and of the weight only its values on them. Since the integrand is given at the members of $\mathcal{A}_{M}$ and not at points, each upper level set is a join of members of $\mathcal{A}_{M}$ and lies in the lattice $\mathcal{L}_{M}$ of the measurement performed, where clause~(ii) makes the weight modular and nothing departs from ordinary averaging. The family $\mathcal{A}_{M}$ is finite by clause~(i), as the ontic state spaces of Sec.~\ref{sec:standard_OMF} are, and its members are pairwise disjoint, so the integral is the sum $\sum_{L_\lambda\in\mathcal{A}_{M}}q(L_\lambda)\,q(k|M,L_\lambda)$ for any responses.\footnote{In what follows we take the responses to be deterministic, valued in $\{0,1\}$, so that each member of $\mathcal{A}_{M}$ answers exactly one outcome. Equation~\eqref{eq:localic_opdata} then reduces to the weight of the join of the members answering $k$; over the outcomes these joins are pairwise disjoint and join to $L_{\Lambda^{M}}$. In the Bell scenario deterministic responses lose no generality (Theorem~\ref{thm:localic_fine}).}

Replacing points by records does not by itself change the probability theory. If every procedure of the theory is performed on the whole of $L_\Lambda$, and the weight of a preparation is defined on one and the same lattice of records for every measurement, then clauses~(i) and~(iii) give every measurement a response throughout $L_\Lambda$, and clause~(ii) makes $q$ a single modular weight normalized on $L_\Lambda$, a finitely additive probability measure carried by records instead of by sets of points.

The probability theory departs from the standard one only where one of those two conditions fails. In this work it is the first that fails. The measurements are not performed on the whole of $L_\Lambda$, and two of them have regions neither of which contains the other. Like the responses, the weight is then defined only within regions, and a record spanning two of them, such as their join, lies in neither of their lattices. Weighing it would mean mixing the weights of the two regions in proportion to how likely each measurement was to be performed, which an ontological model does not specify. Equation~\eqref{eq:ctx_modularity} asks modularity and normalization of each region separately; we call this \emph{contextual modularity}, after the contexts of Sec.~\ref{sec:standard_bell} (the measurements a round can perform). Normalization is a fact each context has, not a fact the world has, and whether a global account can be supplied notwithstanding is addressed in Sec.~\ref{sec:weight}.  We write $L_{\Lambda^{c}}$, $\mathcal{L}_{c}$ and $\mathcal{A}_{c}$ for the region of the context $c$, the records on which its weight is defined, and the family at which its responses are given.

Definition~\ref{def:ontological_model} is the case of Definition~\ref{def:localic_model} in which every procedure is performed on the whole of $L_\Lambda$, the weight of a preparation is defined on one and the same lattice of records for every measurement, and every nonzero record has a point. The last condition makes the mereology extensional, since every record is then spatial and is the join of the one-point records below it, so that a record is nothing over and above the points in it,\footnote{If an open $U$ is not below an open $V$ while every point of $U$ lies in $V$, then $\mathfrak{o}(U)\wedge\mathfrak{c}(V)$ is a nonzero record with no points; so the hypothesis makes $L_\Lambda$ spatial, and, the parts of a record being the records below it, makes every record spatial too. A spatial locale is the join of the one-point parts below it~\cite{PicadoPultr2012}.} which is exactly the feature of Sec.~\ref{sec:standard_OMF} that Sec.~\ref{sec:records} rejected. Under the three conditions, setting $\Lambda:=\mathrm{pt}(L_\Lambda)$ recovers Definition~\ref{def:ontological_model} clause by clause, Eq.~\eqref{eq:ctx_modularity} becoming Eq.~\eqref{eq:modularity} and Eq.~\eqref{eq:localic_opdata} becoming Eq.~\eqref{eq:opdata}. The two frameworks are set side by side in Table~\ref{tab:frameworks}. Wherever two measurements are incompatible, this reduction is precluded by the axioms the next subsection presents.

\begin{table*}[t]
\centering
\setlength{\tabcolsep}{10pt}
\renewcommand{\arraystretch}{1.8}
\begin{tabular}{|l|l|l|}
\cline{2-3}
\multicolumn{1}{l|}{}
&
\parbox[c]{5.3cm}{\centering\textbf{Ontological models}}
&
\parbox[c]{5.7cm}{\centering\textbf{Localic ontological models}}
\\[2mm]
\hline
\parbox[c]{3.0cm}{\centering Ontic state space}
&
\parbox[c]{5.3cm}{\centering Measurable space $\Lambda$}
&
\parbox[c]{5.7cm}{\centering Coframe $\mathsf{S}(L_\Lambda)$}
\\[2mm]
\hline
\parbox[c]{3.0cm}{\centering Ontic states}
&
\parbox[c]{5.3cm}{\centering Points $\lambda$}
&
\parbox[c]{5.7cm}{\centering Sublocales $L_\lambda$}
\\[2mm]
\hline
\parbox[c]{3.0cm}{\centering Epistemic states}
&
\parbox[c]{5.3cm}{\centering Measures, modular and normalized on $\Lambda$, Eq.~\eqref{eq:modularity}}
&
\parbox[c]{5.7cm}{\centering Weights, modular and normalized on each $\mathcal{L}_{PM}$, Eq.~\eqref{eq:ctx_modularity}}

\\[2mm]
\hline
\parbox[c]{3.0cm}{\centering Response functions}
&
\parbox[c]{5.3cm}{\centering Distributions over outcomes at every point of $\Lambda$}
&
\parbox[c]{5.7cm}{\centering Distributions over outcomes at each member of $\mathcal{A}_{M}$}
\\[2mm]
\hline
\parbox[c]{3.0cm}{\centering Reproduction of predictions}
&
\parbox[c]{5.3cm}{\centering Law of total probability over points, Eq.~\eqref{eq:opdata}}
&
\parbox[c]{5.7cm}{\centering Law of total probability over records, Eq.~\eqref{eq:localic_opdata}}
\\[2mm]
\hline
\end{tabular}
\caption{The two frameworks side by side. Beyond what the table shows, the mereology of the left-hand column is extensional, a part being determined by the points in it; that of the right-hand column need not be, a sublocale being determined by how it restricts the whole, and the lattice of sublocales is in general not Boolean. On the left every procedure is performed on the whole of $\Lambda$; on the right each procedure $\pi$, preparation or measurement, is performed on its own region $L_{\Lambda^{\pi}}$, and different procedures may have different regions. A round performing $P$ and $M$ realizes a record below both regions, $\mathcal{L}_{PM}$ is the lattice of records below $L_{\Lambda^{PM}}$ on which the weight representing $P$ is modular and normalized, and $\mathcal{A}_{M}$ is the family of pairwise disjoint records, joining to $L_{\Lambda^{M}}$, at which the responses of $M$ are given. The right-hand column reduces to the left when every procedure is performed on the whole of $L_\Lambda$, the weight is defined on one and the same lattice of records for every measurement, and every nonzero record has a point.}
\label{tab:frameworks}
\end{table*}

\subsection{Shared records must be pointless}
\label{sec:pointlessness}

In Sec.~\ref{sec:records} we distinguished two kinds of records. Some register everything that occurred in a round, which procedures were performed and which outcomes they gave; others register only a part of a round, the identity of a procedure across the rounds in which it could have been performed. We now focus on the second kind, which is needed as soon as a measurement belongs to two contexts that no single round can realize together, and ask what such a record must be like in the framework of Definition~\ref{def:localic_model}.

Nothing in that definition makes one exist. Its clause~(i) gives each procedure a region and imposes no relation between the regions of different measurements, so two contexts sharing a measurement may have zero meet, and the sameness of that measurement is then carried only by its label, the same name, such as $A_0$, appearing in both contexts. Section~\ref{sec:records} argued that it must be carried by the ontology, since otherwise it would be a coincidence that the shared measurement yields the same statistics in both contexts, and, where a common cause is wanted, there would be nothing across the contexts for it to run through. We adopt this conclusion as an axiom.

\begin{axiom}[Relationality]
\label{ax:relationality}
Cross-contextual identity is grounded in, and only in, a shared measurement.\footnote{We state the axiom for a shared measurement, that being the case treated here; a preparation belonging to two ensembles that no single round realizes together would demand the same of that preparation.} If two contexts share a measurement, some nonzero part of the model lies below both, and this part is what carries the sameness of that measurement across them. If two incompatible contexts share none, no nonzero part lies below both.
\end{axiom}

\noindent The axiom has two halves. The first says that a shared measurement has an ontic correlate, one and the same part lying below both contexts. The second says that nothing else does, since two incompatible contexts with no measurement in common have nothing to be the same about. We call this \emph{relationality} because the sameness of a measurement in two contexts here consists in a part the two have in common, and so is a fact about neither taken by itself. It goes further than the relationality of Secs.~\ref{sec:records} and~\ref{sec:pointfree}, where the claim was that a part \emph{can} be individuated by how it restricts a whole; the axiom says which parts a scenario has, and this does not follow from any fact about locales.

A round realizes one context, and $L_\Lambda$ describes one round with all its alternatives, as $\Lambda$ does in Sec.~\ref{sec:standard_OMF}. The regions $L_{\Lambda^{c}}$ are the alternatives one round can realize.  Which outcome the context $c$ would have given, had it been performed, is answered by the members of the family $\mathcal{A}_{c}$ (Sec.~\ref{sec:localic_OMF}). The specific models we consider within the framework supply a second family. The \emph{outcome-fact region} $F^{c}_{o}$ is the part of $L_{\Lambda^{c}}$ at which $c$ was performed and gave $o$, which is the fact a record of the first kind registers. A record registering that $c$ gave $o$ settles what $c$ answers, so $F^{c}_{o}$ lies below the join of the members of $\mathcal{A}_{c}$ answering $o$. The fact regions, unlike that family, need not join to the context, and Definition~\ref{def:localic_model} imposes nothing on them. A record of the first kind lies in a fact region, registering both that its context was performed and what it gave. A record of the second kind can answer an outcome of $c$, as a member of $\mathcal{A}_{c}$, while lying in no fact region of $c$; it then registers no performance of $c$, and nothing prevents it from belonging to a second context as well.

The question this subsection answers is whether a part shared by two incompatible contexts can contain a point. We write $\mathrm{pt}(S)$ for the set of points lying in a part $S$ of $L_\Lambda$; a point lies in a meet exactly when it lies in both parts. A point decides every open by lying in it or not, and the opens encode every distinction present in the whole (Sec.~\ref{sec:pointfree}). It is therefore a complete description of a round, which is what an ontic state is in the standard framework (Sec.~\ref{sec:standard_OMF}). As a record (Sec.~\ref{sec:records}), a point registers a round that has occurred, with the context it performed and the outcome it obtained. This reading of points motivates two conditions, which the next two axioms state. First, the round a point registers realizes one context and gives one outcome of it, so a point does not report two results that no single round could produce together (Exclusivity). Second, that round did occur, so a point lying in a context registers that the context was performed and what it gave (Completeness).

\begin{axiom}[Exclusivity]
\label{ax:exclusivity}
No point lies in two outcome-fact regions whose facts no single round can realize together,
\begin{equation}
\label{eq:exclusivity}
\mathrm{pt}\bigl(F^{c}_{o}\bigr)\cap\mathrm{pt}\bigl(F^{c'}_{o'}\bigr)=\varnothing ,
\end{equation}
in the two cases: (i) $c=c'$ and $o\neq o'$, one round having one result; (ii) $c$ and $c'$ incompatible, one round realizing only one of the two contexts.
\end{axiom}

\begin{axiom}[Completeness]
\label{ax:completeness}
Every point of a context lies in some outcome-fact region of that context,
\begin{equation}
\label{eq:completeness}
\mathrm{pt}\bigl(L_{\Lambda^{c}}\bigr)\subseteq\bigcup\nolimits_{o}\mathrm{pt}\bigl(F^{c}_{o}\bigr).
\end{equation}
\end{axiom}

\noindent The standard ontological models framework has no fact regions of its own, since an ontic state specifies what each measurement would return and not which was performed. In its canonical model, Eq.~\eqref{eq:sectors}, the only possible fact regions are the sectors, and with them Completeness and clause~(i) of Exclusivity hold. Clause~(ii) is the departure. An ontic state carries an outcome for every measurement, incompatible ones included, and so lies in a sector of each of two contexts no single round realizes together; on the reading of a point as a completed round adopted above, it thereby reports an outcome-fact of each, which is what clause~(ii) forbids. Both axioms are conditions on points and say nothing about parts that contain none. Granted the two axioms, the records of the second kind can now be characterized.

\begin{theorem}[Shared records are pointless]
\label{thm:pointless}
Under Exclusivity and Completeness, $\mathrm{pt}\bigl(L_{\Lambda^{c}}\wedge L_{\Lambda^{c'}}\bigr)=\varnothing$ for any two incompatible contexts $c,c'$. In particular every part lying below two incompatible contexts is pointless.
\end{theorem}

\begin{proof}
Let $p$ be a point lying in $L_{\Lambda^{c}}\wedge L_{\Lambda^{c'}}$. It lies in each of the two contexts, so Completeness places it in $F^{c}_{o}$ for some $o$ and in $F^{c'}_{o'}$ for some $o'$, which clause~(ii) of Exclusivity forbids. A point lying in a part lies in every part above it, so the same holds of any part below the two contexts.
\end{proof}

\noindent Note that the proof uses no weight, no value of any response function, no disjointness of the regions, and no finiteness of the outcome set. It also obtains more than Sec.~\ref{sec:records} did, where the argument was that a record of the second kind, registering the identity of one procedure and not everything the round did, is not itself the complete description a point must be; here \emph{nothing} in such a record is a point, whatever it registers.

Were records sets of points, as in the standard framework, no record of the second kind could satisfy the three axioms. A nonempty part shared by two incompatible contexts would contain a point, which Completeness places in a fact region of each, against clause~(ii) of Exclusivity. Completeness and Exclusivity would thus leave two incompatible contexts no nonempty part in common, even when they share a measurement, which is what Relationality forbids. A nonzero part of a locale, by contrast, need not contain any point. Completeness can then hold at every point of a context while a nonzero part of that context lies outside all of its fact regions, and a record shared by two incompatible contexts has room to be such a part.

Three pictures are therefore available of what two incompatible contexts sharing a measurement can share, and the axioms close two of them. (Recall that a spatial part of a locale is determined by its points, so a part that is both spatial and pointless is $0$, and a pointless part that is nonzero cannot be spatial.) So the three are:
\begin{enumerate}[label=(\Alph*),leftmargin=*,itemsep=1pt,topsep=3pt]
\item their common part has a point: one element of reality registering an outcome for each of two incompatible contexts, which is the standard picture, and Exclusivity fails.
\item their common part is pointless and spatial, hence $0$: nothing carries the sameness of the shared measurement, and Relationality fails.
\item their common part is nonzero and pointless, hence non-spatial.
\end{enumerate}
Only (C) survives: Theorem~\ref{thm:pointless} closes (A), Relationality closes (B). An example of a nonzero pointless part is already in Sec.~\ref{sec:pointfree}, where the sublocales of the real line determined by the rationals and by the irrationals share no point and yet share a nonzero part.

Wherever two measurements are incompatible, then, the reduction of Sec.~\ref{sec:localic_OMF} to Definition~\ref{def:ontological_model} is indeed precluded. It asks that every procedure be performed on the whole of $L_\Lambda$ and that every nonzero record have a point. Where two incompatible contexts share a measurement, (C) denies the second; where they share none, Relationality keeps their regions apart, which denies the first. A record is then not determined by the points in it. The extensionality that Sec.~\ref{sec:records} rejected fails here, and not by stipulation. Pointlessness and non-vanishing enter differently: the first is derived, from Exclusivity and Completeness, the second posited, by Relationality. Neither entails the other, and both are needed. Without the first the shared part registers an outcome of each context, and without the second there is no shared part. The model of Sec.~\ref{sec:localic_bell}, on which the rest of the paper is built, shows that the three axioms can hold together, each shared record there being a nonzero pointless part and each fact region having points of its own.

\subsection{Modularity fails across incompatible contexts}
\label{sec:weight}

Across two incompatible contexts, modularity is not a well posed equation. It relates four values,
\[q(S)+q(T)=q(S\vee T)+q(S\wedge T),\]
and at $S=L_{\Lambda^{c}}$ and $T=L_{\Lambda^{c'}}$ three of them are already given. First, each context on its own weighs $1$. Second, their meet lies in both $\mathcal{L}_{c}$ and $\mathcal{L}_{c'}$ and carries no measurement index (Sec.~\ref{sec:localic_OMF}), so the two contexts agree on its value,
\begin{equation}
\label{eq:glueweight}
g_{cc'}:=q\bigl(L_{\Lambda^{c}}\wedge L_{\Lambda^{c'}}\bigr)\le 1,
\end{equation}
the bound holding because the meet lies below a region of weight $1$ and $q$ is monotone. The fourth value, $q(L_{\Lambda^{c}}\vee L_{\Lambda^{c'}})$, is of a different kind. When $g_{cc'}<1$, which is the case of interest, no context region lies above the join. Such a region would contain both $L_{\Lambda^{c}}$ and $L_{\Lambda^{c'}}$, and within it the weight is modular, so the join would receive $1+1-g_{cc'}>1$, more than the weight $1$ of the region containing it. The join, lying within no context region, therefore receives no weight under Definition~\ref{def:localic_model}, and no round realizes it, since that would take a round performing both incompatible contexts at once.

Modularity across incompatible contexts can then be asked only of an extension of $q$, a weight defined also at the join and giving it the inclusion--exclusion value. We call a \emph{completion} of $q$ any monotone $\hat q$ that extends it to a sublattice of $\mathsf{S}(L_\Lambda)$ containing the joins of context regions, and the next theorem shows that, wherever $g_{cc'}<1$, no completion valued in $[0,1]$ is modular.

\begin{theorem}[Completions of the contextual weights]
\label{thm:completions}
Let $c,c'$ be contexts, whose regions have weight $q(L_{\Lambda^{c}})=q(L_{\Lambda^{c'}})=1$, with $g_{cc'}<1$, and let $\hat q$ be a completion.
\begin{enumerate}[label=(\roman*),leftmargin=*,itemsep=1pt,topsep=3pt]
\item If $\hat q$ is modular on the pair $L_{\Lambda^{c}},L_{\Lambda^{c'}}$, then $\hat q\bigl(L_{\Lambda^{c}}\vee L_{\Lambda^{c'}}\bigr)=2-g_{cc'}>1$, so $\hat q$ takes a value outside the range Definition~\ref{def:localic_model} asks of a weight.
\item If $\hat q$ is valued in $[0,1]$, then $\hat q\bigl(L_{\Lambda^{c}}\vee L_{\Lambda^{c'}}\bigr)=1$, and the modular identity fails on the pair by exactly $1-g_{cc'}>0$.
\end{enumerate}
No completion is both modular and valued in $[0,1]$.
\end{theorem}

\begin{proof}
(i) Modularity gives $\hat q(L_{\Lambda^{c}}\vee L_{\Lambda^{c'}})=1+1-g_{cc'}$, which exceeds $1$ because $g_{cc'}<1$. (ii) Monotonicity and the bound $\hat q\le 1$ give $1=\hat q(L_{\Lambda^{c}})\le\hat q(L_{\Lambda^{c}}\vee L_{\Lambda^{c'}})\le 1$; the modular identity then reads $1+1=1+g_{cc'}$, failing by $1-g_{cc'}$.
\end{proof}

\noindent The proof is elementary. It uses no axiom of Sec.~\ref{sec:pointlessness}, nothing about the locale beyond Definition~\ref{def:localic_model}, no finiteness and, for clause~(i), not even monotonicity. Nor need the two contexts share anything, since at $g_{cc'}=0$ the identity fails by $1$. The argument for clause~(ii), moreover, gives more than the theorem states. A completion valued in $[0,1]$ gives weight exactly $1$ to every record above a context region, the top of its domain included, so it is normalized globally, although only normalization within each context was assumed. A modular completion, by contrast, cannot be scaled down into $[0,1]$ without giving every context region less than the weight $1$ that $q$ assigns it.

In the models of Sec.~\ref{sec:localic_bell}, where both kinds of completion exist, giving up modularity while keeping the weight valued in $[0,1]$ and keeping modularity while letting it exceed~$1$ are two conventions for one and the same physical situation. All completions agree wherever $q$ is defined and differ only at the join of two incompatible contexts, which no round realizes, and they make the same predictions, since Eq.~\eqref{eq:localic_opdata} evaluates the weight only at records below the region of the measurement performed. What is not a convention is that no single account is both modular and valued in $[0,1]$. In the standard framework the question does not arise, since every context region is $\Lambda$ itself, whose weight is $1$ by definition of a probability measure.

Notice that the amount by which the identity fails, ${1-g_{cc'}}$, which we call the \emph{defect} on the pair, decreases as two contexts come to share more and reaches its maximum of $1$ where they share nothing at all. Shared records therefore reduce the defect instead of carrying it. Where two contexts share the whole of their weight, so that $g_{cc'}=1$, the defect vanishes and the theorem no longer applies, with the weight of each context lying entirely on the record they share. The standard framework sits there, every context region being $\Lambda$, and total sharing is the price of modularity.

It is tempting to credit the failure of modularity to the pointless parts, on the thought that a weight carried by something with no points has nothing to add up. The thought fails twice. Pointless parts do not make a weight non-modular, since on a countably based \emph{fitted} locale (one satisfying a mild separation condition) any measure on the opens extends canonically to a modular weight on every part, pointless ones included~\cite{Simpson2012}.\footnote{Being modular, no such extension of a probability measure reproduces the contextual data of two incompatible contexts, since clause~(i) would give their join $2-g_{cc'}>1$ while the extension is bounded by the measure of the whole. The canonical construction supplies modularity at the price of the contextual weights.} Nor does non-modularity need pointless parts, since Theorem~\ref{thm:completions} asks nothing about points, and a model whose incompatible contexts shared a part with points in it would carry the same defect, of the same size. The mereology of Sec.~\ref{sec:pointfree} thus permits non-modularity without requiring it.\footnote{An independent, structural motivation for non-modular probability weights could come from taking the co-Heyting mereology of the lattice of sublocales as fundamental. In the classical setting, Kolmogorov's probability measures appear structurally as the ($\sigma$-additive) states of the effect algebra~\cite{foulis1994EffectAlgebrasUnsharpa} associated with the $\sigma$-algebra of measurable subsets of the ontic state space. In future work we aim to study a generalization of effect algebras that captures the co-Heyting mereology of the lattice of sublocales, just as effect algebras capture the fundamentally Boolean mereology of $\sigma$-algebras. The states of such generalized effect algebras should then give a structural justification for a theory of non-modular probability weights.} What makes it unavoidable, wherever two contexts share less than the whole of their weight, is the way Definition~\ref{def:localic_model} defines $q$, normalizing each context on its own, over the alternatives of a round (Sec.~\ref{sec:localic_OMF}), and assigning no value at the join of two of them.

\medskip
Sections~\ref{sec:pointlessness} and~\ref{sec:weight} have one premise in common, that some contexts cannot be performed in the same round, and each reflects it in the model in its own way. Reading a point as a complete record of a round gives Exclusivity and Completeness, and with them pointlessness. Normalizing the weight within each context, over the alternatives of one round, gives non-modularity. Each conclusion from the two subsections stands without the other. The framework is now complete and general, with nothing in it tied to a particular scenario, and the next subsection applies it to the scenario where local causality is at stake.

\subsection{Localic ontological models of the Bell scenario}
\label{sec:localic_bell}

Sections~\ref{sec:records}--\ref{sec:weight} are stated for an arbitrary scenario. We now bring them to the scenario of Sec.~\ref{sec:standard_bell}, where Alice and Bob choose settings $x,y\in\{0,1\}$, the four contexts $(x,y)$ are pairwise incompatible, neighbouring contexts share one measurement, and diagonal ones share none. Three things enter here in turn. The localic ontological models framework of Definition~\ref{def:localic_model} gives each context a region carrying a normalized weight, and imposes no relation between the regions of different contexts. The axioms of Sec.~\ref{sec:pointlessness} fix which contexts share a record and what kind of record that must be. They leave open how the sharing is distributed among the sectors of a context, and there we make one choice, which defines the family of models the rest of the paper works with. No such choice arises in the standard framework. There the elements of reality are the assignments $\lambda=(a_0,a_1,b_0,b_1)$, each answering every measurement, and the sectors follow from them by Eq.~\eqref{eq:sectors}, with nothing left to choose. A record, by contrast, answers only the measurements of the contexts it lies below, so the model has to specify what two contexts have in common. The weights and the response functions are characterized in Sec.~\ref{sec:go}; this subsection specifies the locale. Figure~\ref{fig:localicOM_Bell} shows the one we arrive at, and what follows can be read against it.

Neighbouring contexts share a measurement, and Relationality, which grounds cross-contextual sameness in a shared measurement and in nothing else, places a nonzero record below both. Since the two are incompatible, Theorem~\ref{thm:pointless} leaves that record without points. Diagonal contexts are incompatible as well but share no measurement, so no nonzero record lies below both. The pattern of sharing is therefore the ring of the simple story, each context tied to its two neighbours and to no other.

We now implement our choice of specific localic models by characterizing the records. For each context $(x,y)$ and each outcome pair $(a,b)$ the model posits a \emph{core} $K^{xy}_{ab}$, a record of the first kind (Sec.~\ref{sec:records}), registering that $(x,y)$ was performed and which outcome pair it gave. For each pair of neighbouring contexts and each outcome pair it posits a \emph{glue}, a record of the second kind, registering only the identity of the measurement the two contexts share,
\begin{equation}
\label{eq:glues}
G^{00\leftrightarrow 01}_{ab},\quad
G^{10\leftrightarrow 11}_{ab},\quad
G^{00\leftrightarrow 10}_{ab},\quad
G^{01\leftrightarrow 11}_{ab},
\end{equation}
one for each of the four neighbouring pairs, which share, in this order, the measurements $A_0$, $A_1$, $B_0$ and $B_1$. These thirty-two records are pairwise disjoint, so that the two glues incident on one sector, e.g., $G^{00\leftrightarrow 01}_{ab}$ and $G^{00\leftrightarrow 10}_{ab}$ in the context $(0,0)$, are distinct records with nothing in common. The models built on these records form what we call the family of localic ontological models with \emph{symmetric glues}, every neighbouring pair of contexts having one glue for each outcome pair.

A round performing $(x,y)$ realizes one of the records below the region of that context, and each sector of the context holds three of them: its own core, and one glue shared with each of the two neighbouring contexts,
\begin{equation}
\label{eq:localic_sectors}
L_{\Lambda^{xy}_{ab}}:=K^{xy}_{ab}\;\vee\;G^{xy\leftrightarrow x'y'}_{ab}\;\vee\;G^{xy\leftrightarrow x''y''}_{ab},
\end{equation}
where $(x',y')$ and $(x'',y'')$ are the two contexts neighbouring $(x,y)$. The four sectors of a context are disjoint and join to its region, $L_{\Lambda^{xy}}=\bigvee_{a,b}L_{\Lambda^{xy}_{ab}}$, as in Sec.~\ref{sec:standard_bell}. That makes twelve records to a context, four cores and eight glues, four shared with each neighbour.\footnote{The preparation is performed in every round, so its region is the whole of $L_\Lambda$, and the region of a context is that of the joint measurement $\{A_x,B_y\}$, in the notation of Sec.~\ref{sec:localic_OMF}.} These twelve are the family $\mathcal{A}_{xy}$ at which clause~(iii) of Definition~\ref{def:localic_model} gives the responses of the context.

Equation~\eqref{eq:localic_sectors} gives a glue the same label in both of the contexts it lies below. The glue $G^{00\leftrightarrow 01}_{ab}$, for instance, lies below $(0,0)$ and $(0,1)$, and its $b$ belongs to $B_0$ in one of them and to $B_1$ in the other. The subscript is therefore a label on records and by itself says nothing about outcomes. What a record answers in a context is determined by the response functions of that context, so a record of the sector labelled $(a,b)$ need not answer $(a,b)$ there. Throughout the models of Sec.~\ref{sec:family} every record answers its own label in each of the contexts it lies below; for the stabilizer strategy of Sec.~\ref{sec:quantum_strategies} some do not. In the canonical model of Sec.~\ref{sec:standard_bell}, instead, a sector collects the points that answer its label, Eq.~\eqref{eq:sectors}, so a point of the sector labelled $(a,b)$ answers $(a,b)$ by construction.

Two context regions meet in the glues they share,
\begin{equation}
\label{eq:localic_meets}
L_{\Lambda^{xy}}\wedge L_{\Lambda^{x'y'}}=\bigvee_{a,b}G^{xy\leftrightarrow x'y'}_{ab},
\end{equation}
with $(x',y')\ \text{neighbouring}\ (x,y)$. Diagonal contexts, which share none, meet in $0$.\footnote{Both follow from the disjointness of the thirty-two records: the joins are finite and $\mathsf{S}(L_\Lambda)$, being a coframe, is a distributive lattice, so the meet of two context regions is the join of the meets of their records, and distinct records meet in $0$.} Every three of the four contexts include a diagonal pair, so no nonzero record lies below three of them.

Let us compare these records with the canonical model of Fig.~\ref{fig:standardOM_Bell}. There the sameness of $A_0$ across the two contexts containing it is carried by the points those two contexts share, and every such point lies in one sector of every context, answering all four measurements at once (Sec.~\ref{sec:standard_bell}). Here it is carried by the meet of the two context regions, the join of the glues of Eq.~\eqref{eq:localic_meets}, which has no points and lies below two contexts and no more. The grid of contexts and sectors survives unchanged, and Fig.~\ref{fig:localicOM_Bell} keeps it. The change is in what sits in the overlaps, and how much. There two contexts share the whole of $\Lambda$, here four of their twelve records, or none.

So far we have characterized the records but not exhibited a locale that carries them. Give each of the thirty-two records $r$ a copy $J_r$ of the open interval $(0,1)$ of the real line, and let 
\begin{equation}
\label{eq:realization}
L_\Lambda:=\mathcal{O}\Bigl(\textstyle\bigsqcup_{r}J_r\Bigr)
\end{equation}
be the locale of their disjoint union. A core is the whole of its own copy, the open part $\mathfrak{o}(J_r)$, whose points are the reals of the interval. A glue is the smallest dense part of its own copy, $\mathfrak{o}(J_r)_{\neg\neg}$, the interval's version of the $\mathbb{R}_{\neg\neg}$ of Sec.~\ref{sec:pointfree}, which overlaps every open subinterval of $J_r$ and contains no point of any.\footnote{The locale itself is spatial, as Eq.~\eqref{eq:realization} is the frame of opens of a topological space. What has no points are some of its parts (the glues).} The outcome-fact regions of Sec.~\ref{sec:pointlessness} are made of cores: $F^{xy}_{o}$ is the join of the cores of $(x,y)$ that answer $o$ there. A glue enters none of them, since, as a record of the second kind, it answers an outcome in each of its two contexts without registering that either was performed. Where every core answers its own label, as throughout the models of Sec.~\ref{sec:family}, this is simply $F^{xy}_{ab}=K^{xy}_{ab}$.

\begin{proposition}
\label{prop:realization}
In the locale of Eq.~\eqref{eq:realization} the thirty-two records are pairwise disjoint, each core is spatial, and each glue is nonzero and has no points. Relationality, Exclusivity and Completeness all hold, whatever response functions are assigned to the members of $\mathcal{A}_{xy}$, so the axioms of Sec.~\ref{sec:pointlessness} are jointly satisfiable.
\end{proposition}

\begin{proof} Each glue lies in its copy, whose points are the reals of the interval. An open interval has no isolated point, so the glue identifies the interval with the interval minus any one of its points, as $\mathbb{R}_{\neg\neg}$ does for the line (Sec.~\ref{sec:pointfree}), and so none of those points is a point of the glue. The glue is nonzero, being dense in a nonzero copy. Each core is an open part of a topological space, hence spatial. Records on distinct copies meet in $0$, which gives their disjointness and with it Eq.~\eqref{eq:localic_meets} and Relationality. Exclusivity holds because the fact regions it compares, of two outcomes of one context or of two contexts, are joins of disjoint families of cores, which share no point. For Completeness, let $p$ be a point of $L_{\Lambda^{xy}}$. It lies on one copy, and the meet of $L_{\Lambda^{xy}}$ with that copy is the member of $\mathcal{A}_{xy}$ the copy carries, or $0$ if it carries none. That meet contains $p$, so it is neither $0$ nor a glue; it is a core of $(x,y)$, and $p$ lies in the fact region of the outcome that core answers.
\end{proof}

\noindent On this locale modularity fails whatever weight a model gives the records. Two diagonal contexts already show it. They meet in $0$, so their shared weight is $g_{cc'}=0$, and by Theorem~\ref{thm:completions} every completion valued in $[0,1]$ violates the modular identity $q(S)+q(T)=q(S\vee T)+q(S\wedge T)$ on their two regions by exactly $1$. The logic of records is likewise settled by this locale. In Boolean logic a proposition and its negation exclude one another. For records, the counterpart of the negation is the smallest record that completes a record to the whole, and the two may overlap, which is what makes the logic paraconsistent (App.~\ref{Appendix:Paraconsistency}). On this locale they overlap only in glues, so that records built from cores behave as in Boolean logic, and the departure sits exactly in the records that two incompatible contexts share.

None of the three conditions under which the framework reduces to the standard one (Sec.~\ref{sec:localic_OMF}, Table~\ref{tab:frameworks}) can hold here. The regions cannot all be the whole of $L_\Lambda$, since two diagonal contexts would then share everything, against Relationality, and with the regions distinct the weight is defined on a different lattice for each context, so the condition of a single lattice goes too. Nor can every nonzero record have a point, since the record two neighbouring contexts share has none. Notice that incompatibility is what rules the three conditions out. The argument against making every region the whole of $L_\Lambda$ needs two incompatible contexts that share no measurement, and the argument against points needs two that share one. With no incompatible contexts, clause~(ii) of Exclusivity has no instance and a shared record may have points. Every region may then be the whole of $L_\Lambda$, with the weight defined on a single lattice of records, so that contextual modularity, Eq.~\eqref{eq:ctx_modularity}, reduces to the modularity of Eq.~\eqref{eq:modularity}, and a standard ontological model satisfies all three axioms.

We now have a framework, three axioms that our notion of reality of Sec.~\ref{sec:records} asks of any model, and a family of localic ontological models of the Bell scenario, with symmetric glues, satisfying them. How is local causality stated for such models, and what constraints does it place on the statistics they can reproduce?


\begin{figure*}[t]
\centering
\includegraphics[width=0.95\linewidth]{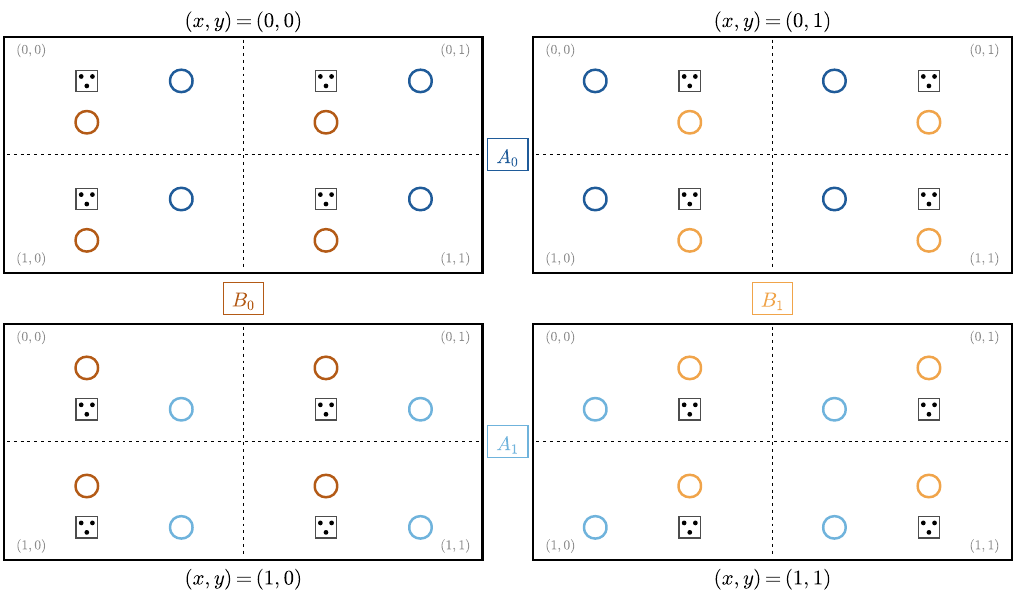}
\caption{\textbf{A localic ontological model of the Bell scenario, with symmetric glues.} The grid is that of Fig.~\ref{fig:standardOM_Bell}: the four contexts $(x,y)$ with Alice's setting by row and Bob's by column, and within each context the four sectors $L_{\Lambda^{xy}_{ab}}$ of Eq.~\eqref{eq:localic_sectors}, indexed in grey by the outcome pair $(a,b)$. Each sector holds three of the thirty-two records of the locale: one \emph{core} (square), which has points, drawn as dots inside it, and lies below its own context and no other; and two \emph{glues} (circles), which have no points, each lying below its own context and one of the two contexts neighbouring it. There are sixteen cores and sixteen glues, and every glue is drawn twice, once in each of the two contexts it lies below, on the side of its sector facing the other one; the two drawings of one glue therefore sit in the same sector of two neighbouring contexts, facing each other across the gap. The colour of a glue names the measurement its two contexts share, matching the label in the gap between them, which is coloured to match ($A_0$ and $A_1$ in dark and light blue, $B_0$ and $B_1$ in dark and light orange). Colour therefore names a pair of contexts and not an individual record. The four circles of one colour in a context are four distinct glues, one to a sector, and a glue is identified by its colour together with its sector and its context. In Fig.~\ref{fig:standardOM_Bell} two contexts sharing a measurement also share points, and those points are what carry the sameness of the measurement. Here they share glues, which carry it without having points at all. Diagonal contexts share no measurement and no record, so no nonzero record lies below three contexts. The figure shows the locale alone, before any weight or response functions. The drawing is schematic: the boxes are not disjoint subsets of a set, since each glue lies below both of the contexts whose boxes carry its two drawings, and a coframe with non-spatial parts admits no faithful picture on paper.}
\label{fig:localicOM_Bell}
\end{figure*}

\section{Bell's \emph{go} theorem}
\label{sec:go}

In the simple story, the web of records took the place of the plan as the common cause, with no spooky influence between Alice and Bob. At each record, Alice's answer depends on her own question alone and Bob's on his. Locality was never abandoned; what the story gave up was the single list answering all four measurements at once. This section makes that step of the story precise.  Section~\ref{sec:localic_local_causality} defines local causality in the localic ontological models framework of Sec.~\ref{sec:localic_framework}. Section~\ref{sec:family} then supplies the weights and the response functions that Sec.~\ref{sec:localic_bell} deferred. Every model in the resulting family of localic models with symmetric glues satisfies local causality, and their Bell--CHSH values run from the classical bound of Theorem~\ref{thm:chsh} to the algebraic maximum.  Section~\ref{sec:quantum_strategies} represents quantum strategies on the same locale, the strategy attaining the Tsirelson bound among them, which completes, for the Bell scenario, desideratum~(iii) of the introduction. The localic ontological models framework thus turns Bell's no-go theorem into a \emph{go} theorem.

\subsection{Local causality in the localic framework}
\label{sec:localic_local_causality}

Local causality rests on two principles, relativistic causality and Reichenbach's principle (Sec.~\ref{sec:local_causality}). Neither of these says whether an element of reality is a point of a set or a part of a locale, so both also apply to localic ontological models. The common cause is now the web of records of Sec.~\ref{sec:records}, weighted by the preparation, in place of $\Lambda$ with its epistemic state. It is fixed before any round is run and involves no setting, measurement independence holding by construction in Definition~\ref{def:localic_model}. A round realizes one record of the web, and local causality demands of that record what Definition~\ref{def:local_causality} demanded of~$\lambda$: once the record is completely specified, no correlation between the wings remains. Of the ingredients of Definition~\ref{def:localic_model}, only the response functions enter this demand, and clause~(iii) defines them at the members of $\mathcal{A}_{M}$ and nowhere else, here the family $\mathcal{A}_{xy}$ of each context.

\begin{definition}[(Localic) local causality]
\label{def:localic_local_causality}
A localic ontological model of the Bell scenario is \emph{locally causal} if
\begin{equation}
\label{eq:localic_factorizability}
q(a,b|x,y,L_\lambda)\;=\;q(a|x,L_\lambda)\,q(b|y,L_\lambda),
\end{equation}
for all $a,b,x,y$ and $L_\lambda\in\mathcal{A}_{xy}$, and if
\begin{equation}
\label{eq:localic_coherence}
q(a|x,L_\lambda)\;=\;q(a|x,L_{\lambda'})
\end{equation}
whenever $L_\lambda\in\mathcal{A}_{x0}$ and $L_{\lambda'}\in\mathcal{A}_{x1}$ overlap, $L_\lambda\wedge L_{\lambda'}\neq 0$, and likewise for Bob.
\end{definition}

Summing Eq.~\eqref{eq:localic_factorizability} over one party's outcome gives parameter independence, Eq.~\eqref{eq:parameter_independence}, now asked at a record, i.e., given the record, each party's outcome distribution depends on their own setting alone.\footnote{Factorizability is equivalent to parameter independence together with outcome independence, $q(a|b,x,y,L_\lambda)=q(a|x,y,L_\lambda)$. With the deterministic responses we work with (Sec.~\ref{sec:localic_OMF}) the second is automatic, the joint response being concentrated on one outcome pair and therefore equal to the product of its own marginals, so parameter independence is the whole of what Eq.~\eqref{eq:localic_factorizability} asks. The same holds of Definition~\ref{def:local_causality}, which is why Bell's 1964 argument could proceed from determinism together with locality (Sec.~\ref{sec:local_causality}).} Equation~\eqref{eq:localic_coherence} extends the same demand to the records lying below two contexts without being members of both. Where a member of $\mathcal{A}_{x0}$ and a member of $\mathcal{A}_{x1}$ overlap, their common part is a record of both contexts, and Eq.~\eqref{eq:localic_coherence} asks that the two members, at which the responses are given, answer $A_x$ alike. The responses of two contexts that share Alice's setting therefore agree for $A_x$ on all that the two have in common, whereas the measurements performed on Bob's side in those contexts are different ones and their responses are unrelated. Definition~\ref{def:local_causality} needs no counterpart of Eq.~\eqref{eq:localic_coherence}, because there all four contexts have their responses at one and the same collection of ontic states.

The demand of local causality is thus the same as in Sec.~\ref{sec:local_causality}, restated for records. What changes is the kind of object it constrains, and with it which contexts it relates to one another. In Definition~\ref{def:local_causality} the constraint falls on a $\lambda$ lying below all four contexts at once, so satisfying it there relates the responses of each context to those of the other three. Here, Definition~\ref{def:localic_local_causality} relates two contexts only through the records lying below both of them, and does not relate at all two contexts sharing no record. Which contexts share a record follows from the axioms of Sec.~\ref{sec:pointlessness}, and how much weight the shared records carry is taken up in Sec.~\ref{sec:family}.

Local causality plays out differently at the two kinds of record of Sec.~\ref{sec:records}. A glue registers only part of a round, the identity of the measurement its two contexts share, and can be realized whichever of the two is performed. There, Eq.~\eqref{eq:localic_factorizability} asks that the answer to that measurement be one and the same, and this is where local causality constrains a model (Sec.~\ref{sec:family}). A core registers that its context was performed and which outcome pair it gave, the pair being the one the responses of that context assign to it. Since a core lies below one context, parameter independence is trivially satisfied at it. We recall that the record a round realizes, core or glue, is settled with the round. Fixed in advance are the records, the weight the preparation gives them and what each answers, while nothing prior to the round determines which record the round is at. To ask for an element of reality that settles this before the round is to ask for a property possessed without being recorded, which Sec.~\ref{sec:records} rejected.

Inserting Eq.~\eqref{eq:localic_factorizability} into the reproduction formula~\eqref{eq:localic_opdata}, the correlations of a locally causal localic model take the form
\begin{equation}
\label{eq:localic_locality}
    \wp(a,b|x,y) \;=\; \sum_{L_\lambda\in\mathcal{A}_{xy}} q(L_\lambda)\, q(a|x,L_\lambda)\,q(b|y,L_\lambda),
\end{equation}
the integral of Eq.~\eqref{eq:localic_opdata} being the sum over $\mathcal{A}_{xy}$ (Sec.~\ref{sec:localic_OMF}). Summing Eq.~\eqref{eq:locality} over $b$ leaves an average over $\Lambda$, which is the same set of ontic states for every context, weighted by the same $p(\lambda)$, so $y$ disappears from it and no-signalling follows. The same sum over $b$ in Eq.~\eqref{eq:localic_locality} is equally well defined, the members of $\mathcal{A}_{xy}$ being pairwise disjoint and the weight modular within a context, and it leaves the weight of the join of those members that answer $a$ for $A_x$. What it does not leave is an expression free of $y$, since $\mathcal{A}_{xy}$ and $\mathcal{A}_{xy'}$ are different collections of records. Whether Bob selects $y$ or $y'$ alters no answer Alice's records give, and yet it alters which records her average runs over, so the two averages need not agree. No-signalling does not follow from Eq.~\eqref{eq:localic_factorizability} in the way it followed from Eq.~\eqref{eq:factorizability}.

This comes as no surprise in a reality made of records. Since records register the procedures performed, the statistics of two contexts are averages over different collections of records. Alice's marginals $\wp(a|x,y)$ and $\wp(a|x,y')$, which coincide when no-signalling holds, have in common only the glues lying below both contexts. Each of these carries the same weight in both contexts, by measurement independence, and one answer for $A_x$, by parameter independence, and so contributes the same amount to both marginals. The difference between the two marginals therefore comes entirely from the unshared records, i.e., the cores and the other glues of each context, whose answers to $A_x$ local causality leaves unconstrained.

No causal influence has been hidden in that difference. On Alice's side nothing depends on Bob's setting: her response function has only her own setting and the record among its arguments, the weight carries no setting index, and the records themselves are given before any round is run (clause~(i) of Definition~\ref{def:localic_model}). So when the unshared records contribute equally as well, and Alice's marginal comes out the same at $y$ as at $y'$, it is not because an influence from Bob's setting has been cancelled---there is no such influence in the model. 

The agreement of Alice's two marginals comes instead from a constraint on how one preparation distributes its weight among the records of the four contexts, namely that Alice's wing have a weight of its own whichever measurement Bob performs, and likewise for Bob. We call it \emph{marginal no-signalling}, after the marginal weight it asks of each wing. It is a condition on the model, which the reproduction formula turns into the no-signalling of the statistics, Eq.~\eqref{eq:nosignalling}, and which a standard locally causal ontological model satisfies identically. We state it in App.~\ref{Appendix:Fine}, and the models of the next subsection meet it for every value of the parameters specifying them. 
Notice that the constraint of marginal no-signalling is not arbitrary. A dependence of Alice's marginal on Bob's setting would be a correlation with neither a direct cause nor a common cause, which Reichenbach's principle forbids, so relativistic causality and Reichenbach's principle demand the constraint of a localic model just as they demand local causality. Reproducing the no-signalling statistics costs no unjustified adjustment of the ontology, which is the no-fine-tuning requirement the introduction demanded of a satisfactory realist explanation.

Local causality imposes the same causal structure on records as it does on points---the one drawn in Fig.~\ref{fig:dags}. A directed acyclic graph settles what the relata are and which of them each depends on, an edge from one node to another meaning that the first is an argument of the probability of the second. The graph does not settle what sort of object sits at each node, nor where the weight and the responses of a context are defined, all of $\Lambda$ on the left-hand panel and the region of the context on the right-hand panel. The change made in Sec.~\ref{sec:localic_framework} is of exactly that kind, which is why the two panels have the same graph. The right-hand panel does not, however, carry a Markov condition~\cite{Pearl}, which would relate the graph to a single joint distribution over its nodes, since no completion of the weight is a probability measure (Theorem~\ref{thm:completions}). A counterpart of the condition is expected from the composition rule discussed in Sec.~\ref{sec:programme}.

\begin{figure}[t]
\centering
\includegraphics[width=\linewidth]{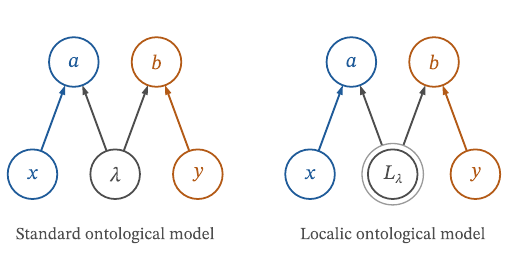}
\caption{\textbf{The causal structure of the Bell scenario, in both frameworks.} Each setting acts on its own wing, no edge joins the two wings, and no edge joins a setting to the common cause. The common cause differs between the panels: the ontic state is a point $\lambda\in\Lambda$ on the left and a record $L_\lambda\in\mathsf{S}(L_\Lambda)$ on the right, drawn with a second, outer circle to mark that a record can have parts where a point has none. Each panel is the causal structure of a locally causal model, local causality being imposed as Definition~\ref{def:local_causality} on the left and as Definition~\ref{def:localic_local_causality} on the right. Read as a classical causal model, the graph on the left entails no-signalling. On the right the weight carries no setting index and each record answers its own party's setting alone, but which records a round can realize is fixed by the context performed, and no-signalling then needs a further condition on the weight, marginal no-signalling, stated in App.~\ref{Appendix:Fine} and met by the models of Sec.~\ref{sec:family}.}
\label{fig:dags}
\end{figure}

\subsection{A family of locally causal models beyond the Bell bound}
\label{sec:family}

We now complete the localic ontological models with symmetric glues of Sec.~\ref{sec:localic_bell} by giving, for these models, the two ingredients of Definition~\ref{def:localic_model} that the locale and the records do not determine: the response functions and the epistemic state (the weight). We choose the two to deliver a family that is locally causal, reproduces no-signalling statistics, and takes every Bell--CHSH value from the classical bound to the algebraic maximum. Proposition~\ref{prop:realization} holds for any choice of the two ingredients, so the axioms of Sec.~\ref{sec:pointlessness} are secured in advance and nothing below has to be checked against them again.

Recall the records of those models, drawn in Fig.~\ref{fig:localicOM_Bell}. Each context has twelve: for each of its four outcome pairs a \emph{core}, lying below that context alone, and two \emph{glues}, each shared with a neighbouring context and having no points. The three records carrying one label make the sector of Eq.~\eqref{eq:localic_sectors}; the twelve make up $\mathcal{A}_{xy}$, the records at which Definition~\ref{def:localic_model} requires a response; and sixteen cores together with sixteen glues make thirty-two in total.

In this subsection, we choose response functions that give the same answer as the label. Specifically, we take each record to answer, in each context it lies below, the outcome pair labelling its sector there,
\begin{equation}
\label{eq:familyresponses}
q\bigl(a',b'\,|\,x,y,L_\lambda\bigr)\;=\;\delta_{a'a}\,\delta_{b'b} ,
\end{equation}
for every $L_\lambda\in\mathcal{A}_{xy}$ labelled $(a,b)$, so that a core $K^{xy}_{ab}$ answers $(a,b)$ in $(x,y)$ and a glue $G^{xy\leftrightarrow x'y'}_{ab}$ answers $(a,b)$ in both of its contexts. The join of the members of $\mathcal{A}_{xy}$ answering $(a,b)$ is then the sector $L_{\Lambda^{xy}_{ab}}$ itself, and since we consider deterministic responses, Eq.~\eqref{eq:localic_opdata} returns the weight of that sector, $\wp(a,b|x,y)=q(L_{\Lambda^{xy}_{ab}})$ (Sec.~\ref{sec:localic_OMF}). The statistics are read directly off the weight, which is what motivates the choice. Nothing forces it, of course. Labels are labels: what a record answers is what the response functions of a context assign to it (Sec.~\ref{sec:localic_bell}), and Sec.~\ref{sec:quantum_strategies} uses records that answer differently from their labels.

Local causality holds in every model of the family, for any weight that populates the records. Section~\ref{sec:localic_local_causality} showed that with deterministic responses the whole of Eq.~\eqref{eq:localic_factorizability} is parameter independence, Eq.~\eqref{eq:parameter_independence}, and that at a record lying below two contexts this demands the same response in both to the measurement those two contexts share. Our responses meet the demand by inspection. A core lies below one context and faces no demand at all. A glue lies below two, which share exactly one measurement, and it answers one and the same label in both, so the party performing the shared measurement returns the same outcome in both contexts. The other party's two answers belong to two different measurements, which Eq.~\eqref{eq:parameter_independence} leaves unrelated. Equation~\eqref{eq:localic_coherence} asks nothing further, given that two members of different families are either one and the same glue or disjoint.

What remains to choose is the epistemic state, which by Definition~\ref{def:localic_model} is the weight a preparation gives to the records. We take one that treats the four contexts alike, giving one weight to the two sectors whose label satisfies $a\oplus b=xy$, and so wins the CHSH game of Eq.~\eqref{eq:pwin} in that context, and another to the two that lose. Such a state is fixed on the sectors by a single number: writing $w$ for the weight of a winning sector and $\ell$ for that of a losing one, normalization within a context reads $2w+2\ell=1$ and leaves $\ell=\tfrac12-w$. We parametrize that number by $\eta$, so that $\eta=0$ places the model at the classical bound and $\eta=1$ at the algebraic maximum:
\begin{equation}
\label{eq:wl}
w:=\frac{3+\eta}{8},\qquad \ell:=\frac{1-\eta}{8},\qquad \eta\in[0,1].
\end{equation}

A glue sits in the sector carrying its label in each of its two contexts, and the two contexts need not agree on whether that label wins. Three of the four have $xy=0$ and win on $a=b$, and the fourth wins on $a\neq b$. The pairs $00\leftrightarrow01$ and $00\leftrightarrow10$ therefore join contexts that agree, and a glue of theirs can sit in a winning sector of both; the pairs $01\leftrightarrow11$ and $10\leftrightarrow11$ join contexts that disagree, and every glue of theirs sits in a winning sector of one context and a losing sector of the other. Two numbers thus suffice, and we set
\begin{equation}
\begin{split}
\label{eq:glueweights}
&q\bigl(G^{00\leftrightarrow 01}_{ab}\bigr)=q\bigl(G^{00\leftrightarrow 10}_{ab}\bigr)=u\,\delta_{ab},
\\ &
q\bigl(G^{01\leftrightarrow 11}_{ab}\bigr)=q\bigl(G^{10\leftrightarrow 11}_{ab}\bigr)=v\,(1-\delta_{ab}),
\end{split}
\end{equation}
placing $u$ on the labels that win in both contexts of the first two pairs, and $v$ on the labels that win in $(1,1)$. Each glue receives one value, which serves both of the contexts below which it lies (Sec.~\ref{sec:localic_OMF}).

Each core then carries what its sector needs beyond the two glues that sector holds,
\begin{equation}
\label{eq:coreweights}
\begin{array}{c|cccc}
q\bigl(K^{xy}_{ab}\bigr) & (0,0) & (0,1) & (1,0) & (1,1) \\ \hline
a=b & w-2u & w-u & w-u & \ell \\
a\neq b & \ell & \ell-v & \ell-v & w-2v
\end{array}
\end{equation}
and Fig.~\ref{fig:family_weights} shows the whole allocation on the grid of Fig.~\ref{fig:localicOM_Bell}. These are non-negative exactly when
\begin{equation}
\label{eq:ranges}
0\le u\le \frac{w}{2},\qquad 0\le v\le \ell .
\end{equation}
The bound on $u$ comes from $(0,0)$, where two glues of weight $u$ share a winning sector of weight $w$. The bound on $v$ comes from $(0,1)$ and $(1,0)$, where a glue of weight $v$ has to fit inside a losing sector of weight $\ell$; since $\ell\le w/2$, this is tighter than what $(1,1)$ asks of the same glue.

Adding the weights of the three records of each sector gives the statistics. In the context $(0,0)$ a winning sector holds a core of weight $w-2u$ and two glues of weight $u$, and a losing sector holds a core of weight $\ell$ and two glues of weight $0$; the other three contexts work the same way, and in every case
\begin{equation}
\label{eq:family_table}
\wp(a,b|x,y)\;=\;q\bigl(L_{\Lambda^{xy}_{ab}}\bigr)\;=\;
\begin{cases} w, & a\oplus b=xy, \\[2pt] \ell, & \text{otherwise.}\end{cases}
\end{equation}
Normalization within each context is $2w+2\ell=1$ and needs no compensating negative contribution anywhere. Each value of $a$ meets one winning and one losing sector, and likewise each value of $b$, so both marginals equal $w+\ell=\tfrac12$ in every context.

The correlators are $E_{00}=E_{01}=E_{10}=2(w-\ell)=(1+\eta)/2$ and $E_{11}=-(1+\eta)/2$, so that
\begin{equation}
\label{eq:family_S}
S\;=\;2+2\eta .
\end{equation}
As $\eta$ runs over $[0,1]$, Eq.~\eqref{eq:family_table} traces the segment joining the statistics at the classical bound, where $w=3/8$, to those of the Popescu--Rohrlich box of Sec.~\ref{sec:bell_nogo}, where $w=1/2$, and $S$ runs from $2$ to $4$ with it. At $\eta=\sqrt2-1$ it takes the Tsirelson value $2\sqrt2$, and Eq.~\eqref{eq:family_table} gives there exactly the statistics of the optimal quantum strategy, as Sec.~\ref{sec:quantum_strategies} shows.

Among the localic models of the family we set aside those in which the glues carry no weight. Suppose every record lying below two contexts had weight zero, $u=v=0$. The records of positive weight below a context would then lie below that context alone, and distinct contexts would have disjoint sets of them, each summing to one. Nothing would relate the statistics of the four contexts. Any four distributions, one to a context, signalling ones included, would be reproduced, and the identifications the experiment itself makes, that $A_0$ is one and the same measurement in two of the contexts, would have no counterpart entering the statistics. Such a model corresponds to four worlds, one to a context, joined only by records that carry no weight. It restates the statistics of each context on its own and says nothing about how the four hang together---which is what a realist account of a Bell scenario has to supply---and a value $S>2$ obtained there would leave local causality with nothing to constrain.\footnote{The case in which just one of $u$ and $v$ vanishes is set aside for a similar reason. The context $(0,0)$ has $(0,1)$ and $(1,0)$ as neighbours and shares $2u$ with each, while $(1,1)$ has the other two and shares $2v$ with each. At $u=0$ the context $(0,0)$ therefore shares no record of positive weight with any other. Local causality then leaves its statistics free with respect to those of the other three, and its contribution to $S$ can be set at will; at $v=0$ the same holds of $(1,1)$.} Zero glue weight is the opposite extreme from the standard ontological models framework, where two contexts share the whole of their weight and the defect of Theorem~\ref{thm:completions} vanishes (Sec.~\ref{sec:weight}). The models of Bell's \emph{go} theorem accordingly give positive weight to records below every neighbouring pair of contexts.

\begin{theorem}[Bell's \emph{go} theorem]
\label{thm:go}
For every Bell--CHSH value $S\in[2,4]$ there is a localic ontological model of the Bell scenario, satisfying the axioms of Sec.~\ref{sec:pointlessness} and giving positive weight to records below every neighbouring pair of contexts, that is locally causal and reproduces no-signalling statistics with that value of~$S$.
\end{theorem}

\begin{proof}
For the values $S\in[2,4)$ everything is already in place. The localic ontological models with symmetric glues defined above have non-negative weights by Eq.~\eqref{eq:ranges}, and the twelve of each context sum to~${2w+2\ell=1}$. Each glue receives one value, serving both contexts below which it lies, and the responses are deterministic and normalized at each member of~$\mathcal{A}_{xy}$. The twelve records of a context are pairwise disjoint and $\mathsf{S}(L_\Lambda)$ is distributive, so summation over them extends the weight to the lattice they generate. On that lattice the weight is monotone and modular, and it is normalized at the context region, which is what clause~(ii) of Definition~\ref{def:localic_model} asks; the same lattice contains the meets with the other context regions, by Eq.~\eqref{eq:localic_meets}. The axioms of Sec.~\ref{sec:pointlessness} hold by Proposition~\ref{prop:realization}, local causality by the verification of parameter independence given above, and every neighbouring pair shares $2u$ or $2v$, both positive.  Equation~\eqref{eq:family_table} gives both marginals the value $\tfrac12$ in every context, which is Eq.~\eqref{eq:nosignalling}, and Eq.~\eqref{eq:family_S} gives $S=2+2\eta$. Letting $\eta$ run over $[0,1)$ attains every value in $[2,4)$.

The localic models above reproduce the case $S=4$ only with $\eta=1$, where $\ell=0$ and Eq.~\eqref{eq:ranges} forces $v=0$, which is the case we have just set aside. The need for $v=0$ is not an artefact of how we parametrized the weight. It arises for every weight on this locale, as long as each record answers the outcome pair labelling its sector. To see this, recall that the weight of a sector is then the probability of the pair labelling it, so that a record answering a pair of probability zero carries no weight. In the Popescu--Rohrlich statistics the pairs with $a\neq b$ have probability zero in $(0,1)$ and in $(1,0)$, and those with $a=b$ have probability zero in $(1,1)$. A glue on a pair of contexts containing $(1,1)$ carries one label into both of them, and whichever label it carries is null in one of the two. Every such glue therefore has weight zero, and both pairs containing $(1,1)$ are left unpopulated.

Letting a record answer something other than its label removes the obstruction. The label of a record is not an outcome commitment, and what a record answers in a context is determined by the response functions of that context (Sec.~\ref{sec:localic_bell}). On the pair $01\leftrightarrow 11$ the shared measurement is $B_1$, so parameter independence pins Bob's answer across the two contexts and leaves Alice's two free; let one glue of that pair answer $(0,0)$ in $(0,1)$ and $(1,0)$ in $(1,1)$, and a second answer $(1,1)$ and $(0,1)$. Each answers a winning pair in each of its two contexts, and each returns one and the same outcome for $B_1$. Do the same on $10\leftrightarrow 11$ with the parties exchanged: the shared measurement is $A_1$, so Alice's answer is pinned and Bob's two are free. On the two pairs containing $(0,0)$ the glues with $a=b$ already answer a winning pair in both contexts, and the cores answer their labels throughout. Give two glues of every pair the weight $t$, give the two cores of each context that answer a winning pair the weight $\tfrac12-2t$, and give the remaining records the weight $0$. Every context weighs $4t+2(\tfrac12-2t)=1$, every neighbouring pair shares $2t$, the statistics are $\wp(a,b|x,y)=\tfrac12\,\delta_{a\oplus b,\,xy}$ with both marginals $\tfrac12$, and $S=4$. For $0<t\le\tfrac14$ this model is locally causal, shares positive weight on every neighbouring pair, and satisfies the axioms of Sec.~\ref{sec:pointlessness} by Proposition~\ref{prop:realization}.
\end{proof}

The family of localic models above reproduces the statistics of the Bell scenario (Eq.~\eqref{eq:family_table}) with no fine-tuning. Their causal structure is the standard local one. The record a round realizes screens off both outcomes, each outcome depends on its own party's setting alone, and no setting is an argument of any weight (Fig.~\ref{fig:dags}). They do not posit causal influences that are not reflected in the operational statistics---influences whose effects stay hidden only by setting the parameters to particular values.\footnote{A localic ontological model can also be fine-tuned (though the models of this family are not). Take any member with $u+v\le\ell$ and let the glue $G^{00\leftrightarrow 01}_{00}$ answer $(0,0)$ in $(0,0)$ and $(1,0)$ in $(0,1)$, so that Alice returns a different outcome for $A_0$ according to Bob's setting, an arrow from Bob's setting to Alice's outcome absent from the models above. Resetting two cores of $(0,1)$, $K^{01}_{00}$ to $w$ and $K^{01}_{10}$ to $\ell-v-u$, leaves Eq.~\eqref{eq:family_table} and both marginals unchanged, so the arrow leaves no trace in the statistics. The arrow is what makes the model fine-tuned. An influence is posited, and the two cores take the values that hide it from the statistics.} The localic models are nonetheless not standard classical causal models~\cite{Pearl}, and cannot be recast as ones without changing their causal structure. Were the records the values of a single variable carrying a setting-independent distribution, Theorem~\ref{thm:chsh} would apply and $S$ could not exceed $2$. Localic ontological models thus provide a new framework for causal models, in which the web of records is a common cause of a new kind.

The non-modularity of the weight across incompatible contexts, established in general in Sec.~\ref{sec:weight}, shows itself in the family at the joins of the context regions. For two incompatible contexts $c$ and $c'$, modularity would give $\hat q(L_{\Lambda^{c}}\vee L_{\Lambda^{c'}})=2-g_{cc'}$, which exceeds $1$, while a weight valued in $[0,1]$ gives $1$; by Theorem~\ref{thm:completions} the modular identity therefore fails by the defect $1-g_{cc'}$. Here, that defect is $1$ on each of the two diagonal pairs, which share nothing, and $1-2u$ or $1-2v$ on the four neighbouring pairs. 
The glue weights thus fix the defect on each neighbouring pair, yet the Bell--CHSH value does not depend on the precise values they take. Equation~\eqref{eq:family_table} depends on $w$ and $\ell$ alone, $u$ and $v$ cancelling in every sector, so all the models of the family give the same statistics, and the same $S$, for all $u$ and $v$ allowed by Eq.~\eqref{eq:ranges}. The glues give local causality something to constrain, as the cases set aside above showed. How much they must carry is not decided by the statistics, which leave $u$ and $v$ free, but must be decided by some further principle, as we discuss in Sec.~\ref{sec:Discussion}.

Theorem~\ref{thm:go} also points to a new version of Fine's theorem. The latter, Theorem~\ref{thm:fine}, says that the statistics reproduced by a locally causal ontological model are exactly those obtained by marginalizing a single joint distribution over $A_0$, $A_1$, $B_0$ and $B_1$, and no such distribution exists beyond the Bell bound. Here a single weight on the records again yields the statistics of all four contexts, each by summing over the records of that context, and the statistics so obtained are the no-signalling ones instead of the local ones, once the records answering a measurement are asked to weigh alike in both contexts containing it. What has changed is the weight: it is defined on the records, not on global assignments, and it is not modular across incompatible contexts, so it is not a Kolmogorov distribution and carries no bound with it. Appendix~\ref{Appendix:Fine} makes the correspondence precise, stating and proving the \emph{localic Fine's theorem}.

The second half of the simple story of Sec.~\ref{sec:Intro} now has its mathematics. The story said that the chances of two incompatible contexts do not add into a single total, and that their total exceeds one by exactly what the two contexts do not share. What they share is fixed in the models we characterize here by the weight, Eq.~\eqref{eq:glueweights} giving a neighbouring pair the total $g_{cc'}=2u$ or $2v$, and Theorem~\ref{thm:completions} turns that number into the defect. The story said that without records the four contexts would be four unrelated descriptions with no shared past to account for the correlations, and that with a plan they would be a single description to which Bell's limit applies---the cases $u=v=0$ and $g_{cc'}=1$. It said, finally, that none of this calls for an influence between Alice's region and Bob's, which is local causality, holding for every model above, for any weight that populates the records. The simple story is thus realized, in Theorem~\ref{thm:go}, by explicit locally causal localic ontological models that achieve every value of $S$ from the classical bound to the algebraic maximum.

\begin{figure*}[t]
\centering
\includegraphics[width=0.95\linewidth]{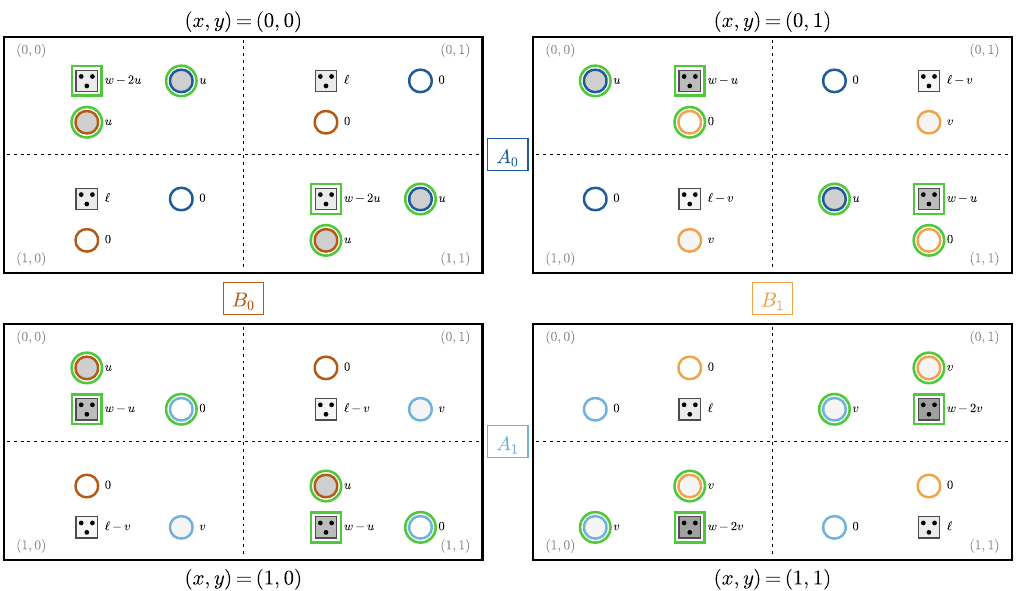}
\caption{\textbf{Weight and response functions of the localic ontological models with symmetric glues.} The figure supplies to Fig.~\ref{fig:localicOM_Bell} the weight each record carries, from Eqs.~\eqref{eq:glueweights} and~\eqref{eq:coreweights}, and the response functions of Eq.~\eqref{eq:familyresponses}, so that label and answer coincide at every record of this figure, which is a choice and not a rule. The weight is written beside the record and shown as its shade, white for a record of weight zero and darker the more weight the record carries; the symbols are exact for every $\eta$, $u$ and $v$, while the shades are drawn at $\eta=\sqrt2-1$, $u=\sqrt2/8$ and $v=\ell/2$. A record is highlighted in green when the response functions of its context send it to an outcome pair winning the game of Eq.~\eqref{eq:pwin}, so that the highlighted records of a context are the support of its winning outcome pairs. Notice that a glue carries one weight and so one shade in both of its drawings, but it answers in each of its two contexts separately, and an outcome pair that wins in one may lose in the other. The three records of a winning sector are all highlighted and add to $w$, and the three of a losing sector add to $\ell$, which is Eq.~\eqref{eq:family_table}; the highlighted records of a context add to $2w$, the weight with which it is won. At $\eta=\sqrt2-1$, where $2w=\cos^2(\pi/8)$, the figure is the localic representation of the optimal quantum strategy of Sec.~\ref{sec:quantum_strategies}.}
\label{fig:family_weights}
\end{figure*}

\subsection{Localic representations of quantum strategies}
\label{sec:quantum_strategies}

Theorem~\ref{thm:go} shows that there is a locally causal model at every Bell--CHSH value from the classical bound to the algebraic maximum, the Tsirelson value among them. A reader may still be unsatisfied with this, on the ground that we have not yet proven that quantum theory itself can be reproduced by localic models of this kind, and that if it cannot be, little of interest remains. Reproducing quantum theory would mean assigning a weight to each state and response functions to each measurement once and for all, and returning the Born rule in every scenario in which they appear. Nothing established so far does that, and neither does the localic Fine's theorem of App.~\ref{Appendix:Fine}, which builds a separate model for each statistics given in advance, and only in the Bell scenario. While a localic model of the whole of quantum theory is a matter for future work, we show here how three quantum strategies in the Bell scenario are represented on the locale of Sec.~\ref{sec:localic_bell}. A strategy is a state together with a measurement for each setting, with which Alice and Bob play the game of Eq.~\eqref{eq:pwin}, a round being won when $a\oplus b=xy$; we keep that language throughout. Each of the three is checked against the Born rule, and each shows something the other two do not: an optimal quantum strategy, a stabilizer strategy, and a strategy on a product state.

Let us start from a quantum strategy achieving the Tsirelson bound, which we refer to as \emph{the optimal quantum strategy}. Let the state be $|\Phi^+\rangle=(|00\rangle+|11\rangle)/\sqrt2$, Alice's measurements the Pauli observables $A_0=\sigma_z$ and $A_1=\sigma_x$, and Bob's $B_0=(\sigma_z+\sigma_x)/\sqrt2$ and $B_1=(\sigma_z-\sigma_x)/\sqrt2$. The correlators are $\langle A_xB_y\rangle=(-1)^{xy}/\sqrt2$ and both marginals are uniform, so the Born rule gives
\begin{equation}
\label{eq:optimal_born}
\wp_{\Phi^+}(a,b|x,y)\;=\;\frac14\Bigl(1+\frac{(-1)^{a+b+xy}}{\sqrt2}\Bigr),
\end{equation}
which equals Eq.~\eqref{eq:family_table} at $w=(2+\sqrt2)/8$, that is, at $\eta=\sqrt2-1$. The weight of Eqs.~\eqref{eq:glueweights} and~\eqref{eq:coreweights} at that value of $\eta$, read by the response functions of Sec.~\ref{sec:family}, therefore reproduces this strategy exactly and returns $S=2\sqrt2$ by Eq.~\eqref{eq:family_S}. Figure~\ref{fig:family_weights} is drawn at that value of $\eta$, so it is at once the picture of the family and the picture of this strategy, its shades being the weight $q_{\Phi^+}$ that represents the state. Each context is won with weight $2w=\cos^2(\pi/8)$.

The optimal strategy does not pick out a single model. Equation~\eqref{eq:family_table} depends on $w$ and $\ell$ and on nothing else, so Eq.~\eqref{eq:optimal_born} constrains $\eta$ alone, and every pair $(u,v)$ satisfying Eq.~\eqref{eq:ranges}, with both positive as in the models of Theorem~\ref{thm:go}, represents this state and these measurements equally well. Moving weight between the core of a sector and its two glues leaves the weight of the sector unchanged, and in the family every record answers its label (Eq.~\eqref{eq:familyresponses}), so each operational probability is the weight of a sector. Such a move changes only how much of the state's weight lies on records shared by two contexts, an amount the next strategy constrains.

Keeping the state and giving Bob Alice's measurements produces a strategy that stays well inside the bound of Theorem~\ref{thm:chsh}. We call it \emph{the stabilizer strategy}. With $A_0=B_0=\sigma_z$ and $A_1=B_1=\sigma_x$ on $|\Phi^+\rangle$ the Born rule gives
\begin{equation}
\label{eq:stab_born}
\wp^{\sigma}_{\Phi^+}(a,b|x,y)\;=\;
\begin{cases} \tfrac12\,\delta_{ab}, & x=y,\\[2pt] \tfrac14, & x\neq y,\end{cases}
\end{equation}
with $E_{00}=E_{11}=1$ and $E_{01}=E_{10}=0$, hence $S=0$. The two parties measure the same observable in the contexts $(0,0)$ and $(1,1)$, and $|\Phi^+\rangle$ is a $+1$ eigenstate of $\sigma_z\otimes\sigma_z$ and of $\sigma_x\otimes\sigma_x$ alike, which is what makes these the stabilizer measurements. The outcomes therefore agree with certainty in those two contexts and are uncorrelated in $(0,1)$ and $(1,0)$. The game is thus won with certainty in $(0,0)$, half the time in $(0,1)$ and $(1,0)$, and never in $(1,1)$, where $a\oplus b=xy$ asks for outcomes that differ. The optimal quantum and stabilizer strategies prepare the same state and differ only in Bob's measurements. We therefore represent $|\Phi^+\rangle$ by the same weight $q_{\Phi^+}$ in both, and ask more: that everything the two strategies share---Alice's two measurements included---be represented in the same way, so that Bob's response functions carry the entire difference between $S=2\sqrt2$ and $S=0$.

The weight is thereby given in advance. With deterministic responses each record answers a single outcome pair in each of its contexts, so the response function of a pair takes the value one at the records answering it and zero elsewhere. These records form its \emph{support}, and by Eq.~\eqref{eq:localic_opdata} their total weight is the probability of the pair. No record can therefore weigh more than the probability of the pair it answers. In the context $(0,1)$ every outcome pair has probability $\tfrac14$, while the core $K^{01}_{00}$ weighs $w-u$ and is the heaviest record there. Hence $w-u\le\tfrac14$, that is $u\ge w-\tfrac14=\sqrt2/8$, which with Eq.~\eqref{eq:ranges} confines $u$ to $\bigl[\sqrt2/8,\,(2+\sqrt2)/16\bigr]$. We set $u=\sqrt2/8$, where
\begin{equation}
\label{eq:stab_point}
w-2u=\ell,\qquad w-u=\ell+u=\tfrac14 ,
\end{equation}
and leave $v$ free in $(0,\ell)$; Figs.~\ref{fig:family_weights} and~\ref{fig:stabilizer} are drawn at $v=\ell/2$.\footnote{Setting $u=\sqrt2/8$ is a choice. With deterministic responses, Eq.~\eqref{eq:stab_born} is met exactly at those $(u,v)$ allowed by Eq.~\eqref{eq:ranges} that satisfy one of $u=\sqrt2/8$, $u+v=\tfrac14$ with $v\ge\ell/2$, or $u-v=\sqrt2/8$ with $v\le\ell/2$, a condition imposed by the contexts $(0,1)$ and $(1,0)$ alone. With the stochastic responses that Definition~\ref{def:localic_model} also permits, Eq.~\eqref{eq:stab_born} is met at every admissible $(u,v)$.}

Given the statistics to be reproduced, keeping Alice's response functions unchanged determines Bob's. Alice answers, at every record and in every context, the $a$ that labels it, which is what she does in Sec.~\ref{sec:family} and so in the optimal strategy. The statistics leave open what a record of weight zero answers, and we let it answer its label throughout. In the contexts $(0,0)$ and $(1,1)$ Eq.~\eqref{eq:stab_born} gives weight only to the outcome pairs with $a=b$, so at every record of positive weight there, glues included, Bob answers what Alice answers. Every glue of positive weight lies below one of those two contexts, and therefore answers a pair with $a=b$ there. Four of these glues share a measurement of Bob's, and for those parameter independence entails that Bob repeats the same answer in the neighbouring context.

Consider now the context $(0,1)$, in which the supports of the four outcome pairs weigh $\tfrac14$ each. Its records of positive weight labelled with $a=0$ are two cores, labelled $(0,0)$ and $(0,1)$ and weighing $w-u$ and $\ell-v$, and two glues, $G^{00\leftrightarrow01}_{00}$ weighing $u$ and $G^{01\leftrightarrow11}_{01}$ weighing $v$. The glue $G^{01\leftrightarrow11}_{01}$ shares $B_1$ with the context $(1,1)$, where Bob answers $0$, so parameter independence gives him the same answer here. Alice answers $0$ as well, so that glue lies in the support of the outcome pair $(0,0)$ and leaves $\tfrac14-v$ of that support to be found among the two cores and the glue $G^{00\leftrightarrow01}_{00}$. The core weighing $w-u$ is too heavy for what is left, since $w-u=\tfrac14$ by Eq.~\eqref{eq:stab_point} and $v>0$, while the core weighing $\ell-v$ and the glue $G^{00\leftrightarrow01}_{00}$ together weigh $(\ell-v)+u=\tfrac14-v$. The support of the outcome pair $(0,0)$ therefore holds that core and that glue, and the core weighing $w-u$ is the only record of positive weight answering the outcome pair $(0,1)$.

The same argument divides the records labelled with $a=1$. It divides the context $(1,0)$ as well, where the glue already determined by parameter independence is $G^{00\leftrightarrow10}_{00}$, of weight $u$. The support of the outcome pair $(0,0)$ needs $\tfrac14-u=\ell$ beyond that glue, the core weighing $w-u$ is again too heavy, and the core weighing $\ell-v$ and the glue $G^{10\leftrightarrow11}_{01}$ together weigh $(\ell-v)+v=\ell$. In every context, then, each glue of positive weight answers a pair with $a=b$, and at the eight cores of the contexts $(0,1)$ and $(1,0)$ Bob answers the reverse of the $b$ that labels the record. Parameter independence holds at every glue, and local causality with it, since Alice gives each glue its label in both of its contexts and, where the shared measurement is Bob's, his answer was carried across.

The model is drawn in Fig.~\ref{fig:stabilizer}. Coarse-graining the four outcomes of the context $(x,y)$ into won and lost, according to $a\oplus b=xy$, yields a two-outcome measurement whose winning response function is $\sum_{a\oplus b=xy}q(a,b|x,y,L_\lambda)$. With deterministic responses this function takes the value one at some of the twelve records and zero elsewhere, its support is what the figure highlights in green, and the context is won with the total weight of the highlighted records.

A record is drawn in the sector it belongs to, and answers the label of that sector unless a second pair is written in grey beneath its weight, in which case it answers that one. In Fig.~\ref{fig:family_weights} no grey pair appears, since there every record answers its label, and the winning support is accordingly the union of the two winning sectors. In Fig.~\ref{fig:stabilizer} twenty grey pairs appear, and the winning support may then take part of a sector but not all of it. In the context $(0,1)$, for instance, it holds the two glues of the sector labelled $(0,0)$ but not that sector's core, which carries $(0,1)$ in grey and so answers a losing pair. A record of weight zero answers its label too, and by Eq.~\eqref{eq:glueweights} the empty glues lying below the context $(0,0)$ are labelled with $a\neq b$ and are highlighted nowhere. Those below $(1,1)$ are labelled with $a=b$ and are highlighted in the contexts $(0,1)$ and $(1,0)$ but not in $(1,1)$.

Reading the four contexts in turn shows where the strategy gains and where it loses. In the context $(0,0)$ the measurements return only $a=b$, which wins there. The support of the outcome pair $(0,0)$ holds the three records of the sector labelled $(0,0)$ together with the core of the sector labelled $(0,1)$, which carries $(0,0)$ in grey, and the support of the outcome pair $(1,1)$ is built the same way from the sectors labelled $(1,1)$ and $(1,0)$. Each of the two weighs $(w-2u)+2u+\ell=w+\ell=\tfrac12$, and between them they hold every record of positive weight, the four left unhighlighted carrying none, so the context is won with weight $1$.

In the contexts $(0,1)$ and $(1,0)$ each of the four supports weighs $\tfrac14$: the two winning ones hold a core and one glue from each neighbour, $(\ell-v)+u+v=\ell+u=\tfrac14$, and the two losing ones a single core of weight $w-u=\tfrac14$, once records of weight zero are set aside. Eight records are highlighted there and four are not, and of those four the two carrying weight are the heaviest on the panel, so eight highlighted markers carry half the weight and the context is won with weight $\tfrac12$. In the context $(1,1)$ the measurements return only $a=b$ once more, and there $a=b$ loses. The six records answering $(0,0)$ weigh $\ell+(w-2v)+2v=\tfrac12$, the six answering $(1,1)$ weigh the other half, and the winning response function takes the value one nowhere, so nothing at all is highlighted. Every round in that context yields an outcome, and the outcome always loses. The four contexts are won with weight $1,\tfrac12,\tfrac12,0$, averaging $\tfrac12$ and returning $S=0$ through Eq.~\eqref{eq:pwin}.

Figures~\ref{fig:family_weights} and~\ref{fig:stabilizer} carry the same shading and differ only in what the records answer. One state, one locale, one weight, and one pair of response functions for Alice: between $S=2\sqrt2$ and $S=0$ there is no difference at all in the ontological representation of the state, nor in that of the two measurements the strategies share. The whole of the difference is in what Bob's measurements make of the same records. The state plants the weight and the measurements harvest it, as the introduction put it, and here one planting is harvested twice. The optimal measurements take $\cos^2(\pi/8)$ of it from every context, the stabilizer measurements all of it from $(0,0)$, half from each of $(0,1)$ and $(1,0)$, and none at all from $(1,1)$.

The last strategy we consider is what we call \emph{the product state strategy}. It keeps the measurements of the optimal strategy and changes the state. With $|\Psi\rangle=|00\rangle$, Alice's outcome is $0$ with certainty at setting $0$ and uniform at setting $1$, Bob's is $0$ with probability $2w$ at either setting, and the two are independent, so $E_{00}=E_{01}=1/\sqrt2$ and $E_{10}=E_{11}=0$, giving $S=\sqrt2$. On the locale of Fig.~\ref{fig:localicOM_Bell} the weights are
\begin{equation}
\label{eq:product_weight}
q_{\Psi}\bigl(K^{xy}_{ab}\bigr)=\wp_{\Psi}(a,b|x,y),\qquad q_{\Psi}(G)=0
\end{equation}
for every glue $G$, and every record answers its label. These are the response functions the optimal strategy already gave to these four measurements, so between that strategy and this one only the weight has changed, and it now populates no glue. Each context is normalized, local causality holds for the reason it holds in Sec.~\ref{sec:family}, and each sector weighs what the Born rule asks of it, the cores alone carrying the statistics. As in Fig.~\ref{fig:family_weights}, no record of Fig.~\ref{fig:productstate} departs from its label, so the winning support of every context is again the union of two whole sectors. The four contexts are won with weight $2w,2w,\tfrac12,\tfrac12$, averaging to $\tfrac12+\sqrt2/8$ and returning $S=\sqrt2$.

This model is not among those of Theorem~\ref{thm:go}, which give positive weight below every neighbouring pair. The framework does not forbid such a model, and the theorem simply makes no claim about it, for the reason given in Sec.~\ref{sec:family}. Where no shared record carries weight, nothing relates the four contexts, and local causality, which holds at every record here as elsewhere, has no shared weight through which to reach the statistics. Nothing is lost by leaving such models aside. The statistics of $|\Psi\rangle$ come from a single distribution over the outcomes of all four measurements, so by Theorem~\ref{thm:fine} a standard locally causal ontological model reproduces them as well; the localic account and the standard one agree on every prediction and differ in what they posit. A localic model with positive weight below every neighbouring pair reproduces them too, as it does every no-signalling statistics (Theorem~\ref{thm:localic_fine}), so asking for positive shared weight leaves out representations, such as the one above, and none of the no-signalling statistics.

The three strategies treated here do not answer the objection that opened this subsection, as was said at the outset. Their role is to indicate concretely how paradigmatic quantum strategies can be represented by locally causal localic ontological models, and how much such representations already constrain. One locale carried the optimal strategy at the Tsirelson value, a stabilizer strategy on the same state, and a product state under the same measurements as the first, each checked against the Born rule. Unlike the models of the localic Fine's theorem, the first two representations were not built one statistics at a time. The first strategy fixed the sector weights of $q_{\Phi^+}$; the second, served by the same weight and the same response functions for Alice and with deterministic responses, confined the glue weight $u$ to the interval $[\sqrt2/8,(2+\sqrt2)/16]$; and at $u=\sqrt2/8$ Bob's response functions were forced, not chosen. Rigidity of that kind will shape how a locally causal localic ontological model of quantum theory is built.

\begin{figure*}[t]
\centering
\includegraphics[width=0.95\linewidth]{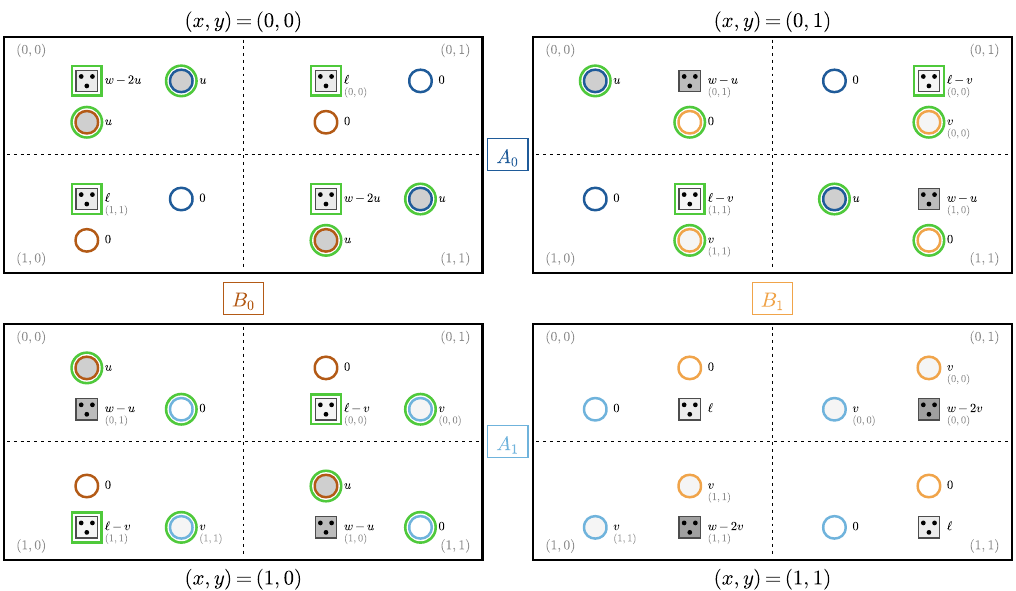}
\caption{\textbf{The stabilizer strategy.} Symbols, shades and geometry are those of Fig.~\ref{fig:family_weights}, drawn at the same $\eta=\sqrt2-1$, $u=\sqrt2/8$ and $v=\ell/2$: one weight on one locale, with only the response functions changed. A record answers the label of the sector it is drawn in, unless a pair is printed in grey beneath its weight, in which case it answers that pair. Twenty grey pairs appear, and every one of them departs in Bob's outcome alone. Highlighting follows the rule of Fig.~\ref{fig:family_weights}, applied context by context to these answers. The four contexts are won with weight $1,\tfrac12,\tfrac12,0$, and $S=0$.}
\label{fig:stabilizer}
\end{figure*}

\begin{figure*}[t]
\centering
\includegraphics[width=0.95\linewidth]{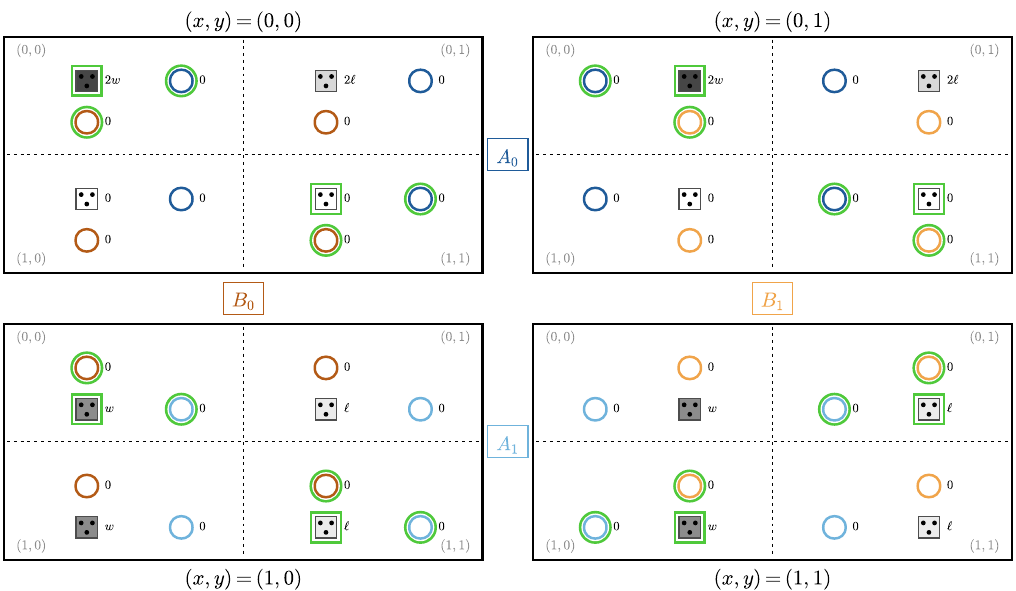}
\caption{\textbf{The product state strategy.} The locale, the symbols and the conventions are those of Figs.~\ref{fig:localicOM_Bell} and~\ref{fig:family_weights}. The response functions are the ones of Fig.~\ref{fig:family_weights}. Only the state has changed, so every record answers its label here too and the highlighting falls in the same places. Every glue is white, the cores alone carrying the statistics of Eq.~\eqref{eq:product_weight}. Records of weight zero are drawn and not deleted, the locale here being the locale of the other two figures; what changes with the state is which of its records are populated. The four contexts are won with weight $2w,2w,\tfrac12,\tfrac12$, and $S=\sqrt2$.}
\label{fig:productstate}
\end{figure*}

\section{Discussion}
\label{sec:Discussion}

\subsection{Redefining realism}
\label{sec:takeaway}

Bell's theorem assumes a notion of realism, mathematically represented by the standard ontological models framework, under which no account of the quantum statistics satisfies local causality. The localic ontological models framework represents a different notion of realism, under which, in the Bell scenario, such an account exists. Its elements of reality are no longer points of a measurable space---individual units defined by answering every measurement at once. They are sublocales of a locale, the records---relational units defined as modes of restriction of the whole. Records carry the sameness of a measurement across two incompatible contexts without those pairwise identities composing into an assignment across all of them. The common cause becomes a web of such records, and the models of Theorem~\ref{thm:go} keep the standard local causal structure while the Bell--CHSH value runs from the Bell bound to the algebraic maximum.

A crucial insight that made the result possible is that local causality need not be tied to the standard notion of realism. The Bell bound requires one collection of ontic states to serve all four contexts, each ontic state carrying a response to every measurement. Local causality is then the factorizability condition of Eq.~\eqref{eq:locality} over that collection, and Theorem~\ref{thm:chsh} derives the bound from the condition and the single collection together. In the simple story each such ontic state was a plan: one list answering every question either party might be asked, all at once, across contexts no round can jointly realize. Yet neither relativistic causality nor Reichenbach's principle, the two principles from which local causality is derived, asks for such a list. They demand that no outcome depend on the distant setting and that the correlation between the outcomes have a common cause in the shared past, conditioning on which it disappears. A web of records supplies the shared past without supplying the list, and both demands continue to hold, leading again to a factorizability condition (Definition~\ref{def:localic_local_causality}).

Theorem~\ref{thm:go} is proved by building localic models obeying specific constraints. They satisfy the three axioms of Sec.~\ref{sec:pointlessness} (Proposition~\ref{prop:realization}), which ground the sameness of a measurement shared by two contexts in a part the two have in common, and which require a point to report an outcome of the context it lies in and none of any context incompatible with that one. From the axioms a first key feature of these localic models follows: point-freeness. What two incompatible contexts share must exist and have no point (Theorem~\ref{thm:pointless}), which no space can provide. The second key feature is that the probability theory has to depart from Kolmogorov's. Each context is normalized on its own and nothing fixes a value at the join of two of them, so wherever two contexts share less than the whole of their weight, no account of the weight is both modular and valued in $[0,1]$. Staying inside $[0,1]$ costs the modular identity exactly what the two contexts do not share (Theorem~\ref{thm:completions}). 

Apart from these two features, the localic framework behaves like the standard one. It applies to any operational theory, and its definition contains the standard framework as the special case in which every procedure lives in the whole locale, the weight is defined on one and the same lattice of records for every measurement, and every nonzero record has a point (Table~\ref{tab:frameworks}), a case the axioms exclude wherever two measurements are incompatible. On one locale, Sec.~\ref{sec:quantum_strategies} represents three paradigmatic quantum strategies, a first indication of how a single localic ontological model might reproduce the whole of quantum theory.

The fine-tuning problem that motivates this work does not arise here. No weight and no response function takes a distant setting among its arguments, so there is no influence present in the ontology which is not reflected in the operational statistics. A realist committed to the no-fine-tuning principle now has a concrete way of keeping it. What has to go is the picture of reality made of independently individuated units, together with the classical probability theory that picture carries with it.

\subsection{The first step of a research programme}
\label{sec:programme}

Theorem~\ref{thm:go} entails that the quantum statistics in the Bell scenario need not be at odds with local causality. Whether the whole of quantum theory admits of a locally causal localic model, and what properties such a model would have, is the question it leaves open. Older questions would be worth revisiting in such a model, among them the measurement problem~\cite{Maudlin1995} and whether the quantum state is ontic or epistemic~\cite{Harrigan2010}. The present work is therefore the first step of a research programme whose ultimate aim is a reformulation of quantum theory that provides a realist account free of fine-tuning, and with it an unambiguous picture of the world the theory describes.

A localic ontological model of the whole of quantum theory would provide one locale, with a weight for each state, response functions for each measurement, and a representation of each transformation, fixed once and returning the Born rule in every scenario in which they appear. The main obstacle we see at present, and the one this work circumvents, is composition, namely how two localic models combine into a model of the pair and how the composition structure of quantum theory is then recovered. Quantum theory builds a bipartite system by the tensor product of Hilbert spaces, and a localic model needs a rule building a locale for the pair from a locale for each part. The rule we expect is built on the localic product, for two reasons. Composing systems that are spatial can produce parts with no points~\cite{Isbell1972}, so the composition could supply on its own the point-freeness that Theorem~\ref{thm:pointless} derives from the axioms. And the localic product is built similarly to the tensor product of Hilbert spaces. Both are generated by pairs consisting of one element from each factor, and what is imposed on the pairing is that it respect, in each argument separately, the structure that factor carries, which is addition of vectors in the Hilbert case and joins of opens in the localic one~\cite{PicadoPultr2012}.

Such a composition rule would also address the derivation of no-signalling discussed in Sec.~\ref{sec:localic_local_causality} and App.~\ref{Appendix:Fine}. In the standard framework that derivation needs two ingredients, local causality and the single collection of ontic states serving all four contexts. Local causality survives the move to records; the single collection does not, and the derivation goes with it. Nor can the derivation be recovered by requiring that the records answering $a$ to $A_x$ be the same collection whichever measurement Bob performs, and likewise for Bob, since that gives the four contexts one and the same family of records, and with it a standard locally causal model and the Bell bound. What is needed is a derivation from elsewhere in the ontology, and composition is the natural place. Under such a rule, the part where Alice answers $a$ to $A_x$ is determined by the locale of her wing alone, and if the rule also assigns to every state of the pair a state of her wing, the counterpart of a reduced density matrix, her marginal is the weight that state gives the part, the same for both of Bob's settings. Quantum theory delivers no-signalling from the way subsystems compose and how effects normalize, so a localic model composed in this way would derive it on the same grounds, and not only reproduce it.

A further question, which such a rule might also settle, is what fixes the parameters $u$ and $v$, left open in Sec.~\ref{sec:family}. This is where we would look for an answer to why quantum theory stops at $S=2\sqrt{2}$ instead of reaching $S=4$. The glue weights do not determine $S$, as Sec.~\ref{sec:family} showed, but the two constrain each other. For instance, combining Eqs.~\eqref{eq:wl}, \eqref{eq:family_S} and~\eqref{eq:ranges} gives $v\le(4-S)/16$, and since $v$ is carried by the glues that the context $(1,1)$ shares with its neighbours, the further a model goes beyond $S=2$ the less weight those glues can hold, until at $S=4$ they hold none. Read the other way, a lower bound on that weight is an upper bound on $S$ in this family, and requiring $v\ge(2-\sqrt{2})/8$ would give exactly $S\le 2\sqrt{2}$.
Since $2v$ is the weight the context $(1,1)$ has in common with each of its neighbours, a lower bound on it asks two incompatible contexts to keep a minimum amount in common. One may then speculate that quantum theory's own limits on incompatibility appear in a localic model as a lower bound on the glue weight, with the Tsirelson bound following from it. Were that so, the operational limits of quantum theory would be due to a bound on how much can be recorded of incompatible procedures at once, as those of relativity are due to a bound on how fast a signal can travel. We leave this speculation for future investigation.

Not every next step is as demanding as providing a localic model of the whole of quantum theory or understanding the origin of the Tsirelson bound. The mechanism at work here---records carrying the identity of what two incompatible contexts share while nothing composes those pairwise identities into a single object serving all the contexts at once---is not special to a Bell scenario. It should apply wherever incompatible contexts have to be reconciled: in the GHZ scenario, where no inequality is needed~\cite{GHZ}, in Kochen--Specker scenarios~\cite{KochenSpecker1967,CabelloProof,Budroni2022}, and where the contexts are made of preparations, as in proofs of preparation contextuality~\cite{Spekkens2005,pusey2018robust,CataniUncertainty2022,Selby2023,Khoshbin_2024}. The Bell scenario itself can be read as a prepare-and-measure scenario, in which Alice's setting and outcome steer a preparation on Bob's side and Bob measures it~\cite{liang2010specker,CoeckeKissinger2017,Henaut2018,CataniFaleiro2023,Wright_2023}; read that way, a proof of nonlocality becomes a proof of preparation contextuality, and the contexts to be reconciled are operationally equivalent preparation procedures, for instance the different convex decompositions of the maximally mixed state of a qubit. Nonlocality is also a special case of contextuality in the sheaf-theoretic~\cite{abramsky2011sheaf,ABShane} and graph-theoretic~\cite{CSW,AcinFritzLeverrierSainz2015,AmaralTerraCunha2018,Wagner2022b} accounts. We expect the ingredients this work uses for the Bell scenario to serve in these settings too, and the fine-tuning that accompanies contextuality~\cite{CataniLeifer2023,Cavalcanti2018} to dissolve as it does here.

Extended Wigner's friend no-go theorems~\cite{frauchiger2018quantum,bong2020strong,brukner2018no,wiseman2023thoughtful,schmid2023review} are a further target. A key assumption in these results, in one form or another, is what is known as the absoluteness of observed events, that an observed event is a single objective fact and not relative to anything or anyone. Our models suggest replacing it, since an outcome would have ontic status exactly while a record of it is present, so that a friend's outcome would cease to be a fact of the world once that record is erased.\footnote{The relation between extended Wigner's friend scenarios and fine-tuning is studied in~\cite{ying2023relating}. In~\cite{Walleghem2026} it is argued that absoluteness of observed events is indeed the assumption to give up if one adheres to the principle of no fine-tuning.} These theorems rest on there being no joint distribution over all the outcomes involved, including outcomes that have been erased, from which the empirical probabilities could be recovered. A localic model would turn them into go theorems as it turns Bell's, since one weight on the records, non-Kolmogorov and not a distribution over assignments to all contexts at once, returns the statistics of every context (App.~\ref{Appendix:Fine}).

So far we have used the localic framework to reproduce the kinematics of an operational theory, by considering prepare-and-measure scenarios. The key ontological move has been to take relationality as primary, giving centrality to parts, defined as modes of restriction of the whole. The framework will also have to reproduce the dynamics of an operational theory. One option is to follow the recipe of the standard framework, representing a transformation by a map carrying weights to weights and leaving the ontic state space untouched. A second, which we find more promising, is to make dynamics internal to the elements of reality. The move is of the same kind as making relationality internal to them: a record is then not defined and individuated independently of what comes before and after it, but as an interconnected part of a process. This calls for a sequential composition of parts, one recorded after another, which a locale cannot supply, since its only product is the meet, and a commutative and idempotent product keeps no order. A natural candidate is quantales~\cite{AbramskyVickers1993}, which generalize locales by replacing the meet with an associative product not necessarily commutative or idempotent, while the joins remain as they are. The primary ontological constituents then become processes of recording rather than just records, which places the notion of reality this framework endorses in the tradition of process philosophy~\cite{Whitehead1929,SeibtProcess}, where relationality and dynamics are both primary, as advocated in~\cite{CataniWhitehead2025}.

\subsection{Morals}
\label{sec:morals}

In a late paper titled ``On the impossible pilot wave''~\cite{BellPilotWave}, John Bell wrote a short history of impossibility proofs for hidden variable models of quantum theory. It is the paper where he recounts how the de Broglie--Bohm pilot wave theory made him realize that von Neumann's proof of the impossibility of hidden variables was flawed, and how the pilot wave being nonlocal led him to ask whether \emph{local} hidden variable models are possible, the question his own no-go theorem answered in the negative.

Bell's theorem is one of the main examples of how the methodology of no-go theorems has been key to making progress in the foundations of quantum theory and to identifying its genuine nonclassical features. As argued at length in \cite{ToyFieldTheory} and \cite{catani2022replycommentwhyinterference}, that methodology consists first in articulating what, precisely, one takes to be the principles defining classicality, then in formalizing these principles mathematically, and finally in proving a rigorous no-go theorem that demonstrates a contradiction between the formalized principles and the aspects of the quantum phenomenology in question. In this work we complement this with the methodology of \emph{go} theorems. Once the no-go theorem is proven, one keeps the principles defining classicality, such as the absence of fine-tuning~\cite{Wood2015,CataniLeifer2023} or the principle of Leibnizianity~\cite{Spekkens2019Leibniz}, and replaces the assumption concerning the mathematical framework in which they are expressed, with one equally general and equally precise, so that the contradiction dissolves. Bell himself left a remark pointing this way, singling out the impossibility proofs that concern local causality and expressing the wish that they too may one day be ``illuminated, perhaps harshly, by some simple constructive model''.

Bell called the closing section of that paper \emph{Morals}. He discusses there the lessons to take from this history of impossibility proofs and from his own theorem. One arrives at the end of it knowing what quantum reality cannot be, without feeling any closer to what it is. And one is left with these words, as famous as they are empowering:
\begin{quote}
Long may Louis de Broglie continue to inspire those who suspect that what is proved by impossibility proofs is lack of imagination.
\end{quote}
This work is an attempt to compensate for that lack of imagination, by rethinking what an element of reality is and turning an impossibility proof into a possibility proof.

\begin{acknowledgments}

We thank Chris Heunen for helpful discussions.

L.C.\ thanks Andrew Moshier and Alexander Kurz, who introduced him to point-free topology, and Matt Leifer and Rob Spekkens, who introduced him to the research programme of going beyond ontological models, when this project began in 2019. He also thanks the participants of the \href{https://basic-research.org/}{BRCP} event ``Stromboli Summer Research Intensive 2025'', in particular Jan Głowacki, Guilherme Franzmann, and Johannes Fankhauser, for numerous discussions on this work, and Markus Müller for encouraging him to pursue this line of research.

L.C.\ acknowledges funding by national funds through FCT -- Fundação para a Ciência e a Tecnologia, I.P., under the support UID/50014/2025 (\url{https://doi.org/10.54499/UID/50014/2025}).
N.v.d.S.\ has been partially funded by the French National Research Agency (ANR) within the framework of ``Plan France 2030'', under the research projects EPIQ ANR-22-PETQ-0007, HQI-Acquisition ANR-22-PNCQ-0001 and HQI-R\&D ANR-22-PNCQ-0002.
R.S.B.\ acknowledges funding from FCT -- Fundação para a Ciência e a Tecnologia through the project KaleidosQope, 2023.13603.PEX.

The authors used Claude Opus~5--5.5 and Claude Fable~5.1 (Anthropic), and GPT-5.4--GPT-6 (OpenAI) in the preparation of this manuscript, to assist with the development and verification of the proofs, to improve the presentation, to edit the language, and to produce the figures. The output was checked by the authors at every step, and they take full responsibility for the content of the paper.

\end{acknowledgments}

\bibliography{bib_new}

\appendix

\section{Localic Fine's theorem}
\label{Appendix:Fine}

Fine's theorem, Theorem~\ref{thm:fine}, both identifies the statistics that a locally causal ontological model can reproduce and shows that the canonical model of Sec.~\ref{sec:standard_bell}, with its sixteen plans and its deterministic responses, reproduces all of them. We prove the counterpart of the two claims for localic models, the weight on the records taking the place of the joint distribution and the thirty-two records of Sec.~\ref{sec:localic_bell} that of the plans, though not of their responses, a difference explained after the theorem.

Let us first recall that summing Eq.~\eqref{eq:localic_locality} over $b$ gives Alice's marginal in a locally causal localic model as a sum over the records of the context performed, $\wp(a|x,y)=\sum_{L_\lambda\in\mathcal{A}_{xy}}q(L_\lambda)\,q(a|x,L_\lambda)$, and that for the two contexts containing $A_x$ the sum runs over different collections, which have only the shared records in common (Sec.~\ref{sec:localic_local_causality}). The theorem below needs the two sums to agree, the condition Sec.~\ref{sec:localic_local_causality} named marginal no-signalling.

\begin{definition}[Marginal no-signalling]
\label{def:marginal_nosignalling}
A localic ontological model of the Bell scenario satisfies \emph{marginal no-signalling} if, for each outcome of each measurement, the records answering that outcome weigh the same in the two contexts containing that measurement, that is, for all $a$ and $x$,
\begin{equation}
\label{eq:marginal_nosignalling}
\sum_{L_\lambda\in\mathcal{A}_{x0}}q(L_\lambda)\,q(a|x,0,L_\lambda)\;=\;\sum_{L_\lambda\in\mathcal{A}_{x1}}q(L_\lambda)\,q(a|x,1,L_\lambda),
\end{equation}
where $q(a|x,y,L_\lambda)=\sum_b q(a,b|x,y,L_\lambda)$ is the response to $A_x$ in the context $(x,y)$, and likewise for Bob, with $\mathcal{A}_{0y}$ and $\mathcal{A}_{1y}$ in place of the two collections above, for all $b$ and $y$.
\end{definition}

\noindent Equation~\eqref{eq:marginal_nosignalling} has the same credentials as local causality, being motivated by the two principles from which Sec.~\ref{sec:local_causality} derives the latter, relativistic causality and Reichenbach's principle. Alice's outcome and Bob's setting are spacelike separated, so relativistic causality implies that neither is a direct cause of the other. Nothing in the common past stands behind the two either. Alice's outcome depends on the common past only through the record realized. Which records a round can realize does depend on the context performed, and so on Bob's setting among other things, but not causally, the context being part of what a record registers (Sec.~\ref{sec:records}). Measurement independence, which Definition~\ref{def:localic_model} has by construction, lets no weight depend on a setting, and clause~(i) of the same definition gives the records before any round, so that between them they exclude the superdeterministic and retrocausal mechanisms through which the common past could connect the two (Sec.~\ref{sec:standard_OMF}). A dependence of Alice's marginal on $y$ would then be a correlation explained neither by a direct influence nor by a common cause, which is what Reichenbach's principle forbids.

Parameter independence and marginal no-signalling formalize the absence of an influence of Bob's setting on Alice's outcome, the first at each record (Definition~\ref{def:localic_local_causality}) and the second after averaging over the records of a context. Parameter independence asks that no record change its answer to $A_x$ with $y$, and marginal no-signalling that the records answering a given outcome of $A_x$ weigh the same in the two contexts, whichever records those turn out to be. A localic ontological model can satisfy local causality and the three axioms of Sec.~\ref{sec:pointlessness} and still signal,\footnote{Take the locale and the responses of Sec.~\ref{sec:family}, with every record answering its own label and every glue carrying the weight $t=1/16$, and let three contexts have the statistics $\wp(a,b|x,y)=1/4$ while $(0,1)$ has $5/16$ on its two pairs with $a=0$ and $3/16$ on its two with $a=1$. Every core weighs $1/8$, except in $(0,1)$, where the cores weigh $3/16$ and $1/16$; each context weighs one, every neighbouring pair shares $4t=1/4$, and Proposition~\ref{prop:realization} gives the axioms. No setting is an argument of any weight or response, and yet Alice's marginal $\wp(0|0,y)$ is $1/2$ at $y=0$ and $5/8$ at $y=1$, so the model signals from Bob's wing to Alice's. The axioms of Sec.~\ref{sec:pointlessness} constrain which records are nonzero and which points they carry, and say nothing about the weight.} so Eq.~\eqref{eq:marginal_nosignalling} is independent of both and has to be required separately. It asks two collections of records to weigh the same, a weaker condition than being the same records. Demanding the same records of every measurement, Bob's as well as Alice's, would give the four contexts one and the same family of records, and with it a standard locally causal model, to which Theorem~\ref{thm:chsh} applies (Sec.~\ref{sec:programme}). The models of Sec.~\ref{sec:family}, by contrast, meet the weaker condition at every value of $S$ up to the algebraic maximum. We impose Eq.~\eqref{eq:marginal_nosignalling} here and expect, as anticipated in Sec.~\ref{sec:programme}, a rule composing two localic models into a model of the pair to deliver it, as the tensor product delivers no-signalling in quantum theory.

Local causality still constrains the marginals, through the records the two contexts share, and the constraint can be made quantitative. With deterministic responses, let $R^{xy}_a$ be the join of the members of $\mathcal{A}_{xy}$ answering $a$ for $A_x$, so that $\wp(a|x,y)=q(R^{xy}_a)$ (Sec.~\ref{sec:localic_OMF}).

\begin{proposition}[Local causality and the marginals]
\label{prop:marginal_bound}
In a locally causal localic model of the Bell scenario with deterministic responses, $R^{x0}_a\wedge L_{\Lambda^{x1}}=R^{x1}_a\wedge L_{\Lambda^{x0}}$, and
\begin{equation}
\label{eq:marginal_bound}
\bigl|\wp(a|x,0)-\wp(a|x,1)\bigr|\;\le\;1-g_{cc'},\quad c=(x,0),\ c'=(x,1);
\end{equation}
likewise for Bob. On the thirty-two records of Sec.~\ref{sec:localic_bell} the bound is attained at every value of $g_{cc'}$.
\end{proposition}

\begin{proof}
Let $N:=L_{\Lambda^{c}}\wedge L_{\Lambda^{c'}}$, so that $R^{xy}_a\wedge L_{\Lambda^{xy'}}=R^{xy}_a\wedge N$ for $y\neq y'$, each $R^{xy}_a$ lying below its own context. By distributivity, $R^{x0}_a\wedge N$ is the join of the meets of the members of $\mathcal{A}_{x0}$ answering $a$ with the members of $\mathcal{A}_{x1}$, and Eq.~\eqref{eq:localic_coherence} makes each nonzero meet of this kind lie below a member of $\mathcal{A}_{x1}$ answering $a$, hence below $R^{x1}_a\wedge N$. Exchanging the two contexts gives the reverse inequality, and with it the first claim. The record $R^{x0}_a\wedge N$ therefore lies in both $\mathcal{L}_{x0}$ and $\mathcal{L}_{x1}$, each of which contains $N$ (Definition~\ref{def:localic_model}), and carries a single weight $s_a$ in both. Modularity within $\mathcal{L}_{xy}$ gives $\wp(a|x,y)-s_a=q(R^{xy}_a\vee N)-q(N)$, a number between $0$ and $1-g_{cc'}$, since $N\le R^{xy}_a\vee N\le L_{\Lambda^{xy}}$ and the weight is monotone. The two marginals thus differ by at most $1-g_{cc'}$. On the thirty-two records, with the responses of Sec.~\ref{sec:family}, a weight giving $g_{cc'}$ to one glue of the pair, $1-g_{cc'}$ to a core of $c$ answering $a$ and to a core of $c'$ answering the other outcome, and $1$ to one core of each remaining context, the other records carrying none, attains the bound.
\end{proof}

\noindent Marginal no-signalling is thus a condition on the unshared records alone, the shared ones contributing alike to both marginals by local causality together with measurement independence. The bound is the defect of Theorem~\ref{thm:completions} on the pair, both being the weight that each context does not share with the other. At $g_{cc'}=1$ the bound vanishes, and this is the derivation of no-signalling in the standard framework, where the two contexts share everything; what makes the unshared weight vanish there is the globality of $\Lambda$ and not local causality (Sec.~\ref{sec:localic_local_causality}).

The proposition also makes precise what Sec.~\ref{sec:localic_local_causality} meant by a weight of Alice's wing, local causality supplying the record that carries it and marginal no-signalling its value. By the first claim and distributivity, $(R^{x0}_a\vee R^{x1}_a)\wedge L_{\Lambda^{xy}}=R^{xy}_a$ for both $y$, so the join $R^{x0}_a\vee R^{x1}_a$ is the record of Alice answering $a$ to $A_x$ whichever measurement Bob performs, its part in each context being the records that answer $a$ there. As a join across two incompatible contexts it lies in general in neither lattice $\mathcal{L}_{xy}$, and marginal no-signalling is exactly the condition that its two parts weigh the same, $\wp(a|x,0)=\wp(a|x,1)$, which gives it a weight read off either context. This weight is the counterpart of the reduced density matrix; the standard framework has it from the single $\Lambda$, and here it is what a composition rule has to supply (Sec.~\ref{sec:programme}). With Eq.~\eqref{eq:marginal_nosignalling} required, the correspondence announced in Sec.~\ref{sec:family} can be stated.

\begin{theorem}[Localic Fine's theorem]
\label{thm:localic_fine}
For the statistics $\wp(a,b|x,y)$ of the Bell scenario, the following are equivalent:
\begin{enumerate}[label=(\roman*)]
\item $\wp(a,b|x,y)$ admits of a locally causal localic ontological model satisfying marginal no-signalling;
\item $\wp(a,b|x,y)$ is no-signalling, Eq.~\eqref{eq:nosignalling};
\item $\wp(a,b|x,y)$ admits of a locally causal localic ontological model satisfying marginal no-signalling and the axioms of Sec.~\ref{sec:pointlessness}, on the locale of Eq.~\eqref{eq:realization}, with the thirty-two records of Sec.~\ref{sec:localic_bell}, deterministic responses, and positive weight on the records shared by each neighbouring pair of contexts.
\end{enumerate}
\end{theorem}

\noindent Clause~(i) is permissive about which localic models are admitted. Any locale is allowed, with any records and with responses that need not be deterministic, no axiom of Sec.~\ref{sec:pointlessness} is imposed, and two contexts may share any part of the model or none at all. Besides local causality, which restricts no statistics on its own, the only constraint it contains is Eq.~\eqref{eq:marginal_nosignalling}. The most substantive part of the theorem is (ii)$\Rightarrow$(iii), which reaches all the no-signalling statistics on the thirty-two records, and not only the one-parameter family of Eq.~\eqref{eq:family_table}. The direction (i)$\Rightarrow$(ii) is immediate, since the reproduction formula turns Eq.~\eqref{eq:marginal_nosignalling}, a condition on the weight and the responses, into Eq.~\eqref{eq:nosignalling}, a condition on the statistics. Theorem~\ref{thm:fine} is proved along the same cycle, but with the work in the opposite place: its clauses~(ii) and~(iii) differ by a change of notation and its main content lies in (i)$\Rightarrow$(ii).

\begin{proof}
(iii)$\Rightarrow$(i) holds because a model of the third kind is a model of the first.

(i)$\Rightarrow$(ii). Within a context the integral of Eq.~\eqref{eq:localic_opdata} is the sum over the members of $\mathcal{A}_{xy}$ (Sec.~\ref{sec:localic_OMF}), so Eq.~\eqref{eq:marginal_nosignalling} reads $\wp(a|x,0)=\wp(a|x,1)$, and likewise for Bob, which is Eq.~\eqref{eq:nosignalling}.

(ii)$\Rightarrow$(iii). Let $\wp(a,b|x,y)$ be no-signalling. The locale of Eq.~\eqref{eq:realization} and its thirty-two records are given by clause~(iii) itself, so what has to be supplied is what each record answers in each context below which it lies, and what each record weighs. We let every core answer the outcome pair labelling it, in the one context below which it lies, and we let every glue return one and the same outcome for the measurement its two contexts share, which is required by parameter independence, leaving free the outcome it returns for the other measurement of each of the two. It remains to choose these outcomes and the weight of every record. We start with the glues, where Eq.~\eqref{eq:nosignalling} is used, and then continue with the cores, each one taking what is left of the entry it answers in its context. We write $m$ for the smallest positive entry of $\wp(a,b|x,y)$, which exists because the four entries of a context sum to one.

Take the pair $00\leftrightarrow01$, whose contexts share $A_0$; the other three pairs are treated identically, with the parties exchanged on the two that share one of Bob's measurements. Choose an outcome $a$ of $A_0$ with $\wp(a|x{=}0)>0$, which exists because the two values of that marginal sum to one, and which Eq.~\eqref{eq:nosignalling} makes independent of Bob's setting. Since $\sum_{b}\wp(a,b|0,0)=\wp(a|x{=}0)=\sum_{b}\wp(a,b|0,1)$, each of the two contexts has an outcome of Bob's at which the entry is positive, say $\wp(a,b|0,0)>0$ and $\wp(a,b'|0,1)>0$. We let one glue of the pair answer $(a,b)$ in $(0,0)$ and $(a,b')$ in $(0,1)$, returning $a$ for the shared $A_0$ in both, and give it the weight $m/2$. We choose the other three glues of the pair to answer their own labels, which likewise returns one outcome for $A_0$ in both contexts, and to carry no weight.

Two glues of positive weight now lie below each context, one from each of the two neighbouring pairs that context belongs to, and each of them answers there an outcome pair of positive probability. We give the core $K^{xy}_{ab}$ the weight $\wp(a,b|x,y)$ less the weight of the glues that answer $(a,b)$ in $(x,y)$. In this way, no core weight is negative; at most two glues of positive weight answer $(a,b)$ in $(x,y)$, each of weight $m/2$, and they do so only when $\wp(a,b|x,y)$ is positive and hence at least $m$.

With the construction above, the records of $(x,y)$ answering $(a,b)$ carry $\wp(a,b|x,y)$ between them, which is Eq.~\eqref{eq:localic_opdata} for deterministic responses, and the twelve records of the context carry $\sum_{a,b}\wp(a,b|x,y)=1$. Summation over those twelve extends the weight to a monotone modular weight on $\mathcal{L}_{xy}$, normalized at the context region, as in the proof of Theorem~\ref{thm:go}. Local causality is parameter independence here, for deterministic responses (Sec.~\ref{sec:localic_local_causality}), and it constrains only the records lying below two contexts, namely the glues, which meet it by the choice of outcomes above, and Eq.~\eqref{eq:localic_coherence} adds nothing, the thirty-two records being pairwise disjoint. The axioms of Sec.~\ref{sec:pointlessness} hold by Proposition~\ref{prop:realization}, which is independent of the responses, each neighbouring pair carries the weight $m/2$, and Eq.~\eqref{eq:marginal_nosignalling} holds because its two sums are the marginals of~$\wp$.
\end{proof}

\noindent Theorem~\ref{thm:localic_fine} stands to Theorem~\ref{thm:go} as Theorem~\ref{thm:fine} stands to Theorem~\ref{thm:chsh}. The statistics that Theorem~\ref{thm:fine} characterizes as reproducible in the standard framework obey the bound of Theorem~\ref{thm:chsh}, and those that Theorem~\ref{thm:localic_fine} characterizes as reproducible in the localic one are the no-signalling statistics, which include the ones reproduced by the models of Theorem~\ref{thm:go} at every value of $S$ up to the algebraic maximum. Theorem~\ref{thm:fine} and Theorem~\ref{thm:localic_fine} also agree in fixing a carrier once and for all, the sixteen plans in the one case and the thirty-two records in the other, on which every statistics of the set each theorem names is reproduced.

Two things differ. The first is how the canonical model is obtained. In a locally causal ontological model with epistemic state $p(\lambda)$, multiplying at a given $\lambda$ the response values of the four measurements gives each global assignment $(a_0,a_1,b_0,b_1)$ a probability, so that every $\lambda$ carries a distribution over the sixteen of them. Averaging these distributions over $p(\lambda)$ leaves a single one, which Theorem~\ref{thm:fine} takes as the epistemic state of the canonical model. Here, no record answers all four measurements, so that route is closed, and the model of clause~(iii) is instead read off the statistics, its weights being the entries $\wp(a,b|x,y)$ distributed among cores and glues and the outcomes its glues return for the unshared measurements being chosen where those entries are positive. The construction takes from a model of clause~(i) only the fact that those entries are no-signalling. Unlike the responses of the canonical model, Eq.~\eqref{eq:standard_responses}, those of the glues are therefore not fixed once for all statistics, and with four glues per neighbouring pair they cannot be. A glue of the pair $00\leftrightarrow01$ lies below both contexts and so answers three measurements, $A_0$, $B_0$ and $B_1$, its responses forming a triple $(a_0,b_0,b_1)$. There are eight such triples, and the canonical model has a sector overlap for each of them, while this locale has four glues, which with fixed responses would answer at most four triples. The statistics of a global assignment whose triple is left unanswered would then give every glue of the pair an outcome of probability zero in one of its two contexts, hence no weight, against clause~(iii). Eight glues per neighbouring pair would allow the responses to be fixed once for all statistics; four suffice for each statistics on its own and carry the models of Sec.~\ref{sec:family}.

The second difference is what restricts the reproducible statistics to the set each theorem names. In Theorem~\ref{thm:fine} the restriction comes from local causality together with the single collection of ontic states that serves all four contexts. A localic model need not have such a collection, so local causality no longer restricts anything on its own (Sec.~\ref{sec:family}). The restriction to the no-signalling statistics comes instead from Eq.~\eqref{eq:marginal_nosignalling}, in the hypothesis of (i)$\Rightarrow$(ii).

Theorem~\ref{thm:fine} is recovered from the locally causal localic models that meet the first two conditions of the reduction of Table~\ref{tab:frameworks}, namely that every procedure be performed on the whole of $L_\Lambda$ and that the weight be defined on one and the same lattice of records for every measurement. Under these two conditions the meets of one member from each of the four families are pairwise disjoint and, $\mathsf{S}(L_\Lambda)$ being distributive, join to $L_\Lambda$, so their weights sum to one, the weight being modular on the lattice they lie in; each nonzero meet lies below exactly one member of every family and answers in each context what that member answers. Equation~\eqref{eq:localic_factorizability} makes the joint answer at a meet the product of the two parties' answers, and Eq.~\eqref{eq:localic_coherence} makes Alice's answer to $A_x$ there the same in both contexts containing $A_x$, the two members above the meet overlapping in it, and likewise for Bob. The meets, with these weights and responses, are thus a locally causal ontological model, Definition~\ref{def:local_causality}, reproducing the same statistics by modularity. Equation~\eqref{eq:marginal_nosignalling} is then an identity, holding in every such model and imposed on none, and the statistics characterized are no longer the no-signalling ones but those of Theorem~\ref{thm:fine}. The third condition, that every nonzero record have a point, makes the meets sets of points and plays no part in this. What separates the two theorems is what separates the two frameworks, whether one collection of elements of reality serves all four contexts.

\section{The paraconsistent logic of records}
\label{Appendix:Paraconsistency}

While the parts of the ontic state space of the standard framework are subsets, and the logic they carry is Boolean, the records of the localic models of Sec.~\ref{sec:localic_bell} carry a paraconsistent logic, whose departure from the classical rules falls exactly on the records two incompatible contexts share.

The reading of a lattice as a logic is the one the standard framework already uses. There a proposition is a subset $A\subseteq\Lambda$, and to assert it is to say that the ontic state of the round lies in $A$; conjunction is intersection, disjunction is union, and $A$ implies $B$ exactly when $A\subseteq B$. Here, a proposition is a record~$S$, and to assert it is to say that the record the round realizes is a part of~$S$; conjunction is the meet, disjunction is the join, and~$S$ implies~$T$ exactly when $S\le T$. The order is inclusion in both cases, written $\subseteq$ between subsets and $\le$ between parts. What changes is that a record is not fixed by the points it contains (Sec.~\ref{sec:pointfree}), which leaves room for records without complements, and it is their presence that changes the logic.

Besides the order, which says whether one proposition implies another, a logic has a connective for ``if $U$ then $V$'', a proposition that can itself be asserted and combined with others. In a lattice it is the \emph{implication} $U\to V$, the largest element whose meet with $U$ lies below $V$, i.e., the weakest proposition that together with $U$ implies $V$. A complete lattice in which it exists for all $U$ and $V$ is a \emph{complete Heyting algebra}~\cite{PicadoPultr2012,Vickers1989}, and there the implication is characterized by
\[
W\wedge U\le V
\quad\iff\quad
W\le U\to V,
\]
so that $W$ and $U$ together imply $V$ exactly when $W$ implies $U\to V$. The frame of opens is a complete Heyting algebra, since by Eq.~\eqref{eq:frame} the join of all opens whose meet with $U$ lies below $V$ again has this property.\footnote{In fact, a complete lattice is a frame if and only if it is a complete Heyting algebra.} Taking the whole for $W$, $U$ implies $V$ exactly when $U\to V$ is the whole, which is how the connective agrees with the order. Among the subsets of a set $X$ the implication is $(X\setminus U)\cup V$, the points at which $U$ fails or $V$ holds, and among the opens of a topological space $X$ it is the interior of that subset. The case $V=0$ is the negation $\neg U:=U\to0$ used in Sec.~\ref{sec:pointfree}, the largest open disjoint from $U$, so that $U\wedge\neg U=0$.

Since every open lies above $0$, an open and its negation together imply every open. That rule, ${U\wedge\neg U\le V}$ for every open $V$, is \emph{explosion}, and a logic is \emph{paraconsistent} when it fails, i.e., a contradiction can be tolerated without everything else following from it~\cite{PriestParaconsistent}. Boolean logic satisfies explosion for the same reason. There negation is \emph{complementation}, under which a subset and its complement exclude one another, which gives explosion, and between them cover the whole, which is the \emph{law of the excluded middle}. Negation in the frame keeps the first property, $U\wedge\neg U=0$, and loses the second: for $U=(0,\infty)$ on the real line the negation is $\neg U=(-\infty,0)$, and $U\vee\neg U$ is the line with the point $0$ removed, not the whole line. That is intuitionistic logic.

In the coframe of records, excluded middle and explosion change places. Finite joins distribute over arbitrary meets there (Sec.~\ref{sec:pointfree}), according to the law dual to Eq.~\eqref{eq:frame}, so the records $R$ with $T\le S\vee R$ are closed under meets and have a least member, the \emph{subtraction} $T\smallsetminus S$, the smallest record that covers $T$ together with $S$. The universal property here is the following:
\[
T\le S\vee R
\quad\iff\quad
T\smallsetminus S \le R.
\]
A complete lattice equipped with such a subtraction is a \emph{complete co-Heyting algebra}~\cite{Rauszer1974,Urbas1996}. In place of the negation~$\neg$, a co-Heyting algebra supplies the \emph{supplement} $\cneg S:=L_\Lambda\smallsetminus S$, the smallest record that covers the whole together with $S$. For the supplement, the law of the excluded middle now holds:~${S\vee\cneg S=L_\Lambda}$. However, a record and its supplement need not exclude one another, and what they have in common is the \emph{boundary} ${\partial S:=S\wedge\cneg S}$~\cite{Lawvere1991,ReyesZolfaghari1996}, so named because for closed subsets of a space it is the topological boundary. A record~$B$ is called a \emph{boundary} if~$\partial B = B$, or equivalently~$\cneg B = L_\Lambda$.

Explosion in the coframe reads $S\wedge\cneg S\le T$ for every record $T$, and taking $T=0$ shows this to be the condition $\partial S=0$, that a record and its supplement exclude one another. Suppose $S$ has a complement $S'$. The supplement $\cneg S$ is the smallest record covering the whole together with $S$, so it lies below $S'$, and since $S$ excludes $S'$ it excludes the supplement as well, which makes the boundary $\partial S$ vanish. Suppose instead that $S$ has no complement. Were the supplement to exclude $S$, it would be a complement of $S$, covering the whole together with $S$ as every supplement does. So it does not exclude $S$, and the boundary $\partial S$ is nonzero. Explosion therefore holds at a record exactly when that record has a complement. Which records lack one depends on the locale. On the locale of Eq.~\eqref{eq:realization}, a join of some of the thirty-two records of Sec.~\ref{sec:localic_bell} lacks one exactly when one of them is a glue (Proposition~\ref{prop:boundaries}), and a glue lacks one because a nonzero pointless part of a spatial locale never has a complement, as the next lemma shows. Its proof uses the \emph{one-point sublocale}~$\mathfrak{b}(p)$ that a point~$p$ of a locale~$L$ determines, the part whose only point is~$p$~\cite[\S III.10.1]{PicadoPultr2012}.

\begin{lemma}
\label{lem:pointless_boundary}
    If $L$ is a spatial locale, then any pointless sublocale~$G\in\mathsf{S}(L)$ satisfies $\partial G=G$.
\end{lemma}
\begin{proof}
Let $G$ be a pointless sublocale of a spatial locale~$L$. By definition, $L\smallsetminus G$ is the least sublocale whose join with~$G$ is~$L$. We show that only $L$ itself has this property.

Suppose $G\vee T=L$ for some sublocale $T$ of $L$. Every one-point sublocale $\mathfrak{b}(p)$ of $L$ satisfies ${\mathfrak{b}(p)\wedge G=0}$, since the only sublocales below $\mathfrak{b}(p)$ are $0$ and $\mathfrak{b}(p)$ itself, and $G$ is pointless. Hence
\[
\mathfrak{b}(p)=\mathfrak{b}(p)\wedge L = \mathfrak{b}(p)\wedge(G\vee T)=\mathfrak{b}(p)\wedge T\le T,
\]
so $T$ contains the same points as~$L$. Spatiality implies that $L$ is the join of its one-point sublocales \cite[\S VI.3.1]{PicadoPultr2012}, so $T=L$. Therefore $L\smallsetminus G=L$, and
\[
\partial G=G\wedge(L\smallsetminus G)=G\wedge L = G.
\qedhere
\]
\end{proof}

\begin{proposition}[Glues are the boundaries]
\label{prop:boundaries}
Consider the set of thirty-two records of Sec.~\ref{sec:localic_bell}, realized in the locale $L_\Lambda$ of Eq.~\eqref{eq:realization}, and let $\mathcal G$ be its set of glues. For every family~$\mathcal F$ of records:
\[
\partial\left(\bigvee_{R\in\mathcal F}R\right)
=
\bigvee_{G\in\mathcal F\cap\mathcal G}G.
\]
In particular, every core satisfies $\partial K^{xy}_{ab}=0$, while every glue satisfies $\partial G=G$ and $\cneg G=L_\Lambda$. For each context~$(x,y)$ we have
\begin{equation}
\label{eq:boundaries}
\partial L_{\Lambda^{xy}}
=
\bigvee_{x'y'}
\left(
L_{\Lambda^{xy}}\wedge L_{\Lambda^{x'y'}}
\right),
\end{equation}
where the join runs over the two contexts~$(x',y')$ neighbouring~$(x,y)$.
\end{proposition}

\begin{proof}
Fix the family~$\mathcal F$. Let $C$ be the join of its cores and $\Gamma$ the join of its glues:
\[
C:= \bigvee_{K\in \mathcal{F}\setminus \mathcal{G}}K
\qquad\text{and}\qquad
\Gamma:= \bigvee_{G\in \mathcal{F}\cap\mathcal{G}}G,
\]
so that the join over~$\mathcal{F}$ is~$C\vee \Gamma$. By construction, each core is identified with an open sublocale, so any join of cores will be an open sublocale, meaning we can write~$C=\mathfrak{o}(V)$ for some~$V\in\mathcal{O}(L_\Lambda)$. In turn,~$C$ has complement~${C'=\mathfrak{c}(V)}$. The glues in $\mathcal F$ are disjoint from the cores, so~$\Gamma\le C'$. Moreover, since glues are pointless by Proposition~\ref{prop:realization}, $\Gamma$ is pointless as a finite join of pointless sublocales.

Since $L_{\Lambda}$ is spatial, Lemma~\ref{lem:pointless_boundary} gives $\partial\Gamma=\Gamma$, meaning equivalently that~$\cneg\Gamma=L_{\Lambda}$. For any sublocale $T$ of~$L_\Lambda$ it thus follows that
\[
(C\vee\Gamma)\vee T=L_{\Lambda}
\iff C\vee T=L_{\Lambda}
\iff C'\le T.
\]
Hence $C'$ is the least sublocale covering the whole together with $C\vee\Gamma$, so $\cneg(C\vee\Gamma)=C'$, and hence
\[
\partial(C\vee\Gamma)=(C\vee\Gamma)\wedge C'=\Gamma,
\]
which is precisely the first claimed equation. Taking a single core gives $\partial K^{xy}_{ab}=0$, and taking a single glue gives $\cneg G=L_{\Lambda}$ and $\partial G=G$. Finally, a context region is the join of its four cores and eight glues, and Eq.~\eqref{eq:localic_meets} identifies the join of those glues with the join of its overlaps with its two neighbours, yielding Eq.~\eqref{eq:boundaries}.
\end{proof}

\noindent Equation~\eqref{eq:boundaries} is where explosion fails. Take the sector $L_{\Lambda^{xy}_{ab}}$ of one context and the sector $L_{\Lambda^{x'y'}_{ab}}$ of a neighbouring one. The supplement of a sector is the complement of its core, by the proof of Proposition~\ref{prop:boundaries}, and the second sector lies below it, being disjoint from that core, so the second implies the negation of the first. By Theorem~\ref{thm:pointless} no point lies in both sectors, so no complete description of a round satisfies the two propositions together; by Eq.~\eqref{eq:localic_sectors} their conjunction is nonetheless the glue $G^{xy\leftrightarrow x'y'}_{ab}$, which is not $0$, so not everything follows from the two of them. What the two propositions share, then, is no point but the record carrying the sameness of the measurement that belongs to both contexts (Axiom~\ref{ax:relationality}). The contradiction is tolerated, and it is tolerated in what two contexts that no round performs together have in common. Among subsets a nonzero conjunction without a common point is impossible, since a proposition is fixed by the points that satisfy it, so two propositions no point satisfies together have empty conjunction and everything follows from them.

On the locale of Eq.~\eqref{eq:realization} the glues lack a complement only because of the points that lie in no context region. Since the glues have no points, finite distributivity places every point of a context region in one of its cores. The other points of $L_\Lambda$ are the reals of the glue copies, which lie in no fact region, receive no weight and no response, and are realized by no round. A record covering the whole together with a glue contains every point, by the proof of Lemma~\ref{lem:pointless_boundary}, among them the reals of the glue's own copy. By spatiality these join to the copy, which contains the glue, and this is why the glue is its own boundary. Without these points the glues would have complements, since the join of the four context regions contains none of them, and within it each of the thirty-two records is complemented by the join of the other thirty-one. On this locale the glues thus fail explosion as pointless parts of a spatial whole, and the records that fail it are the shared ones only because the construction makes no other record pointless. A realization in which each glue is a part of the cores of both its contexts, as the sublocales of the rationals and of the irrationals share $\mathbb{R}_{\neg\neg}$ (Sec.~\ref{sec:pointfree}), would make the sharing itself produce the failure, and we leave it to future work.

A record shared by two incompatible contexts has no points (Theorem~\ref{thm:pointless}), in a spatial locale a nonzero record with no points has no complement, and explosion fails at a record exactly when it has none. On the locale of Eq.~\eqref{eq:realization} the logic of records is Boolean on everything built from cores and paraconsistent exactly on the glues (Proposition~\ref{prop:boundaries}). A contradiction is tolerated in the record two incompatible contexts share.

\section{Comparison with other works}
\label{Appendix:Comparison}
 
In this appendix we compare the localic framework with other ways of going beyond the standard ontological models framework and of addressing the problem of fine-tuning. Two questions are asked of each approach. The first is whether it shares our aim, which Sec.~\ref{sec:Intro} states as three desiderata: a precise notion of reality, applicability to any operational theory, and the reproduction of the quantum statistics without fine-tuning. The second is which of the following four commitments of a locally causal ontological model the approach gives up.
\begin{itemize}
\item[(A)] \textit{Individual units.} The elements of reality are individual units, each individuated on its own, independently of any other and of any whole it may belong to, and thus mathematically represented by points of a measurable space; one collection of them serves every context, each carrying a response to every measurement (Secs.~\ref{sec:Intro} and~\ref{sec:takeaway}).
\item[(B)] \textit{Absoluteness of outcomes.} Each round realizes one context, and each measurement performed yields a single outcome, an objective fact not relative to anything or anyone (Sec.~\ref{sec:programme}).
\item[(C)] \textit{Kolmogorov probability.} The statistics are recovered through classical (Kolmogorov) probability theory over that collection.
\item[(D)] \textit{Common-cause structure.} The model satisfies local causality and measurement independence.
\end{itemize}
 
This work gives up (A). Its elements of reality are relational units, each context having a collection of records of its own, and what two contexts share is a record they have in common. Giving up (A) carries the failure of (C) with it, since each context is normalized on its own and the contexts share less than the whole of their weight, so that no account of the weight is both modular and valued in $[0,1]$ (Theorem~\ref{thm:completions}). It keeps (D). Its only consequence for (B) comes from taking records as the elements of reality, since an outcome has ontic status exactly while a record of it is present. In the Bell scenario (B) is kept. The outcome obtained is what the record the round realizes answers in the context performed, and it is an objective fact; what a record lacks is an answer in the contexts it does not lie below, which is a departure from (A) rather than a relativity of the observed event. The situation is different in extended Wigner's friend scenarios, where a record of a performed measurement can be erased, and there the framework is expected to depart from (B) as well, as discussed in Sec.~\ref{sec:programme}.
 
\paragraph*{Giving up local causality or measurement independence.}
The interpretations of quantum theory that can be cast in the standard framework keep (A), (B), and (C) and give up local causality, and as interpretations of one theory they do not aim at theory-independence. They exhibit the fine-tuning in its original form and accept it. The pilot-wave theory of de Broglie and Bohm~\cite{Bohm1952} is the best-known example. Its ontic state is the pair of the wave function and the particle positions, and the outcome on Bob's side depends on Alice's setting, so parameter independence fails. That dependence leaves no operational trace because the positions are distributed according to the Born rule. Collapse theories~\cite{GRW1986,Bassi2013} can also be cast as ontological models, with the wave function as the complete state, evolving by a stochastic modification of the Schr\"odinger equation that is nonlinear for the individual state vector while remaining linear for the statistical operator, which keeps no-signalling intact. They too give up local causality, and they change the predictions of quantum theory.
 
Other proposals also change the predictions of quantum theory, and they do so in order to explain the fine-tuning of ontological models as emergent from physical processes that have not yet been observed. Valentini's reading of the pilot-wave theory is the paradigmatic case, where no-signalling only holds in quantum equilibrium (the Born distribution of positions) and fails outside it. A relaxation mechanism leads to the equilibrium~\cite{Valentini1991,Valentini2002}, and the absence of signalling is then no fine-tuning of the nonlocal influence, since signalling is possible in principle, only not in equilibrium. Reference~\cite{CataniGalleyGonda2024} puts forward a proposal of the same shape for generalized contextuality, where the operational equivalences that correspond to ontological distinctions emerge from a yet undiscovered process of information erasure. What these proposals ask of the world is that quantum theory fail in some regime nobody has reached, in ways that stand in tension with relativity, whether by permitting signalling outright or by admitting measurements finer than quantum theory allows. Explaining a fine-tuning as emergent is thus a different enterprise from removing it, which is what this work does, leaving the quantum predictions as they stand.
 
Retrocausal~\cite{PriceWharton2015,WhartonArgaman2020} and superdeterministic~\cite{tHooft2016,HossenfelderPalmer2020} accounts keep local causality and give up instead the measurement independence that we take the standard framework to assume, departing from it in that one respect while keeping (A), (B), and (C). The challenge they face is the fine-tuning. Meeting it asks for a clear formalization of the realist framework in which these accounts are defined, and several works address the question~\cite{Almada2016,LeiferPusey2017,Adlam2022Retrocausality,Adlam2021Teleological}.

\paragraph*{Giving up the absoluteness of outcomes.} The Everett interpretation~\cite{Everett1957,Wallace} gives up (B), and with it the standard framework, since a round does not yield a single outcome. Each outcome of non-zero amplitude occurs, in a world of its own within the universal quantum state, and Bell's argument then cannot run, its derivation requiring that a quantum experiment have a single outcome~\cite{Bacciagaluppi2002,BrownTimpson2016}. Commitments (A), (C), and (D) are defined for a locally causal ontological model, and each would have to be reformulated before it could be stated in this picture. How probability is to be understood when more than one outcome is realized has long been debated, and we do not enter into the question here~\cite{SaundersEtAl2010}. The account differs from ours in the picture of the world it describes and in interpreting the quantum formalism instead of offering a theory-independent framework.

A similar assessment applies to relational quantum mechanics~\cite{Rovelli1996,DiBiagioRovelli2021} and QBism~\cite{QBism1,QBism2}. They give up (B) by keeping one outcome per measurement and relativizing it, in the first case to the physical system with respect to which the event occurs, in the second to the agent whose experience it is. Whether the Bell scenario is thereby described locally and without fine-tuning is under debate~\cite{MartinDussaud2019,Pienaar2019,Pienaar2020}, and how to formalize the notion of reality these approaches endorse remains work in progress~\cite{Fuchs2017,Pienaar2024,AdlamRovelli2023,MartinDussaud2023,DeBrotaFuchsSchack2026}.
 
\paragraph*{Beyond classical causal models.} The fine-tuning problem is posed by Wood and Spekkens~\cite{Wood2015} within the framework of classical causal models~\cite{Pearl}, the counterpart of ontological models in causal modelling. One route beyond the standard ontological models framework is therefore to generalize the classical causal models framework. The present work can also be seen from this perspective (Fig.~\ref{fig:dags}), and the approaches differ in what sits at the nodes and at the arrows. Quantum causal models~\cite{QCC,BarrettLorenzOreshkov2020} put quantum systems at the nodes, the common cause of the Bell correlations becoming a system in an entangled state, and define the causal relations through unitary evolution. The fine-tuning does not arise in these models. They leave open, however, what a quantum system at a node really is, so they do not by themselves supply the clear notion of reality that our first desideratum asks for.
 
Recent work addresses this. Ormrod and Barrett argue for a notion of reality that resides in the dynamics~\cite{ormrod2024quantum}. The primitive is a unitary circuit, taken as an observer-independent fact. Relative to a subset of its systems, the causal structure singles out sets of mutually exclusive alternatives, represented by projectors, and one alternative from each set occurs, with a probability computed from the circuit alone. An event is then relative to the subset one starts from, and the proposal gives up (B) as the approaches above do. Being formulated for unitary circuits, it does not aim at theory-independence.
 
Another approach in which causal structure is ontologically fundamental, and one that shares the aim of going beyond the ontological models framework, is that of causal-inferential theories~\cite{SchmidSelbySpekkens2020}. The framework separates causation from inference, giving each a process-theoretic formulation of its own, and it takes the classical theory of inference to be Bayesian probability together with Boolean logic. The authors' motivation is to preserve Leibnizianity~\cite{Spekkens2019Leibniz}, a principle closely related to the requirement of no fine-tuning adopted here. The realist causal-inferential theory they call for is conservative with respect to causal structure, keeping the common-cause structure of the Bell scenario (Fig.~\ref{fig:dags}). Its notions of influence and inference, by contrast, are both intrinsically nonclassical. The authors set out criteria such a theory would have to meet and leave their exact formalization open. The localic models of the Bell scenario can be read in this light, since they depart from the classical theory of inference in both of its components, Boolean logic giving way to the paraconsistent logic of records (App.~\ref{Appendix:Paraconsistency}) and Kolmogorov probability to a non-modular weight, while the causal structure is left as it stands. A localic ontological model of quantum theory, were it built, would be a candidate for what they ask for. Phrasing our framework in their language is a natural next step.
 
A further line of research on causal structure also rests on point-free topology. Ordered locales~\cite{HeunenVanDerSchaaf,vanderschaaf2024PointFreeSpacetimes} replace the point-set description of the causal order of spacetime with a lattice of regions, and the two constructions are companions in taking regions rather than points as primitive. How the relativistic setting treated there combines with the quantum one treated here we leave for future research. 
 
\paragraph*{Changing the algebra of events or the probability.} The sheaf-theoretic account of contextuality and nonlocality~\cite{abramsky2011sheaf,ABShane} applies, like ours, to any operational theory. It represents the statistics of each context as a local section of a presheaf over the contexts and identifies nonclassicality with the obstruction to gluing them into a global section, a distribution over joint assignments to all measurements at once. The object whose absence it diagnoses is the one that (A), (B), and (C) together supply. The account stops at the diagnosis, since it does not aim at a new ontological framework, whereas a localic model supplies in place of the missing object one weight on the records, not a distribution over joint assignments, that returns the statistics of every context (App.~\ref{Appendix:Fine}).
 
Partial Boolean algebras~\cite{KochenSpecker1967,AbramskyBarbosa2021} formalize the event structure directly, as a family of Boolean algebras, one per context, overlapping in the events the contexts share, with meets and joins defined only between events lying in a common context. The localic framework extends these operations to every pair of sectors, whatever their contexts, and the nonzero meets of sectors of neighbouring contexts are exactly the glues, the records on which local causality does its work (Sec.~\ref{sec:localic_bell}). In quantum logic~\cite{BirkhoffvonNeumann1936} the operations are likewise extended to every pair of events, but the result is an orthomodular lattice that takes the quantum proposition structure as primitive, leaving it unexplained. While quantum logic gives up distributivity, every departure of this work, records without points, the non-modular weight, and the paraconsistency of App.~\ref{Appendix:Paraconsistency}, takes place inside a distributive lattice.
 
Topos approaches to quantum theory are the closest to this work in their mathematics, but they pursue a different aim. In the programme of Isham, Butterfield, and D\"oring~\cite{IshamButterfield1998,DoeringIsham2008} the state space is replaced by the spectral presheaf over the poset of contexts, whose lack of global elements is equivalent to the Kochen--Specker theorem. Propositions are its clopen subobjects, and a quantum state is represented by a measure on them defined context by context~\cite{Doering2009}. In the Bohrification programme of Heunen, Landsman, and Spitters~\cite{HeunenLandsmanSpitters2009} the state space is a locale, internal to a topos of functors on the contexts, whose lack of points again reflects the Kochen--Specker theorem, and a quantum state becomes a probability valuation on it. Both programmes thus work with a state space without global points and with per-context probability, and both build their logic on a Heyting algebra, the clopen subobjects in one case and the opens of the internal locale in the other, so that the logic is intuitionistic. The algebra of clopen subobjects is moreover bi-Heyting, as D\"oring later showed, so it carries a paraconsistent negation alongside the intuitionistic one on which the programme was built~\cite{Doering2016}. Here the propositions are records, which form a coframe, and the logic App.~\ref{Appendix:Paraconsistency} derives from it is paraconsistent, the intuitionistic logic living on the frame of opens instead. The aim and the objects differ as well. The topos constructions take the operator algebra of quantum theory as their input, treat contextuality as a feature of the world to be represented rather than dissolved, and do not address the problem of fine-tuning. Their state object is a presheaf, or an internal locale, over the poset of contexts ordered by inclusion, whereas here there is one locale, the contexts are its parts, and what two contexts share is their meet. Heunen, Landsman, and Spitters describe their construction as turning quantum theory into a classical theory, at the price of reasoning about it in an intuitionistic logic, in which the law of the excluded middle (for any proposition~$p$ we have~$p\vee\neg p$) does not hold in general.
The localic framework keeps the usual rules of reasoning and changes instead what the elements of reality are, so that its non-Boolean logic is the logic of the records themselves.
 
Since the Bell inequalities are the conditions for a single Kolmogorov probability space to serve all four contexts, as Pitowsky showed~\cite{Pitowsky1989,Pitowsky1994}, exceeding the Bell bound by departing from Kolmogorov probability alone is a natural route, and it has already been explored. Every no-signalling statistics, up to the algebraic maximum and so beyond the Tsirelson bound, is reproduced by a signed measure over the joint assignments~\cite{abramsky2011sheaf,AbramskyBrandenburger2014,AlSafiShort2013}, so a model of that kind exists for every correlation we exhibit, including ones outside the quantum set. The departure from (C) in this work differs from that route in three respects. The weights remain non-negative. The failure of the modular identity is confined to the joins of incompatible contexts, which no experiment reaches, and its size is what those contexts do not share (Theorem~\ref{thm:completions}). And it is derived rather than posited, following from each context being normalized on its own while the contexts share only part of their weight. Signed measures show how far a departure from (C) can reach, and the localic models show how such a departure can follow from a change in what the elements of reality are.

\paragraph*{Other directions.} Müller's algorithmic idealism goes beyond the ontological models framework by giving fundamental physics, quantum theory included, a different role: that of telling an observer what they should expect to see next rather than which reality accounts for the operational statistics~\cite{Muller2020,Muller2026}. Fankhauser, Gonda, and De les Coves stay within the ontological models framework and enrich it. The standard framework leaves the observer outside the model, represented only through the preparations and measurements performed, while they keep a classical ontology and model the observer as a physical system within it, asking what such an observer can learn about the world~\cite{Fankhauser2024,Gonda2026}. Neither approach sets out to address the problem of fine-tuning.
 
Among the approaches that posit elements of reality and depart from one of the four commitments, the departure is made at (B), (C), or (D), and (A) is either kept or left without a precise reformulation. The topos approaches are an exception, but in reformulating (A) they do not address the problem of fine-tuning. The localic framework reformulates (A) within an ontological model, which is what permits a locally causal account of the Bell correlations (Theorem~\ref{thm:go}).

\end{document}